\documentclass{theoretics}

\ThCSauthor[1]{Amey Bhangale}{amey.bhangale@ucr.edu}[0000-0002-3878-9241]
\ThCSauthor[2]{Subhash Khot}{khot@cs.nyu.edu}[0009-0007-9246-4011]
\ThCSauthor[3]{Dor Minzer}{dminzer@mit.edu}[0000-0002-8093-1328]

\ThCSaffil[1]{Department of Computer Science and Engineering, University of California, Riverside}
\ThCSaffil[2]{Department of Computer Science, Courant Institute of Mathematical Sciences, New York University}
\ThCSaffil[3]{Department of Mathematics, Massachusetts Institute of Technology}

\ThCSthanks{This article was invited from STOC 2025 \cite{BKMSTOC2025}. Amey Bhangale is supported by the Hellman Fellowship award; Subhash Khot is supported by the NSF Award CCF-1422159, NSF CCF award 2130816, and the Simons Investigator Award; Dor Minzer is supported by the NSF CCF award 2227876 and NSF CAREER award 2239160.}
\ThCSshortnames{A. Bhangale, S. Khot, D. Minzer}
\ThCSshorttitle{On Approximability of Satisfiable $k$-CSPs: V}
\ThCSyear{2026}
\ThCSarticlenum{9}
\ThCSreceived{Jul 29, 2025}
\ThCSrevised{Jan 9, 2026}
\ThCSaccepted{Feb 22, 2026}
\ThCSpublished{May 20, 2026}
\ThCSdoicreatedtrue
\ThCSkeywords{Constraint Satisfaction Problems, Hardness of Approximation, Analysis of Boolean Functions.}

\title{On Approximability of Satisfiable \texorpdfstring{$k$}{k}-CSPs: V}

\usepackage{bm}
\usepackage{bbm}
\usepackage[breakable]{tcolorbox}
\usepackage{blindtext}

\newcommand{\symm}{\textsc{Mildly-Symmetric}}
\newcommand\E{\mathop{\mathbb{E}}}
\newcommand\card[1]{\left| {#1} \right|}
\newcommand\sett[2]{\left\{ \left. #1 \;\right\vert #2 \right\}}

\newcommand\set[1]{{\left\{ #1 \right\}}}
\newcommand\Prob[2]{{\Pr_{#1}\left[ {#2} \right]}}

\newcommand\cProb[3]{{\Pr_{#1}\left[ \left. #3 \;\right\vert #2 \right]}}

\newcommand\Expect[2]{{\mathop{\mathbb{E}}_{#1}\left[ {#2} \right]}}

\newcommand\cExpect[3]{{\mathbb{E}_{#1}\left[ \left. #3 \;\right\vert #2 \right]}}
\newcommand\norm[1]{\| #1 \|}

\newcommand\skipi{{\vskip 10pt}}

\newcommand\spn{{\sf spn}_{\mathbb{N}}}

\newcommand\inner[2]{\langle{#1},{#2}\rangle}
\newcommand\eps{\varepsilon}

\renewcommand\geq{\geqslant}
\renewcommand\leq{\leqslant}

\newcommand{\rom}[1]{\uppercase\expandafter{\romannumeral #1\relax}}

\def\ggg{\gtrsim}
\def\lll{\lesssim}

\newcommand{\truncerr}{\varrho}

	\newcommand{\V}[1]{\boldsymbol{#1}}
	
	\newcommand{\calP}{\mathcal{P}}
	\newcommand{\calC}{\mathcal{C}}
	\newcommand{\calG}{\mathcal{G}}
	\newcommand{\calF}{\mathcal{F}}
	\newcommand{\calL}{\mathcal{L}}
	\newcommand{\bF}{\boldsymbol{F}}
	\newcommand{\bcalF}{\boldsymbol{\mathcal{F}}}

	\newcommand{\bcalX}{\boldsymbol{\mathcal{X}}}
	\newcommand{\bg}{{\V g}}
	\newcommand{\scale}{\mathbf{Scale}}

	\newcommand{\dict}{\mathbf{Dict}}
	\newcommand{\Round}{\mathbf{Round}}
	\newcommand{\simplex}{\blacktriangle}
	
    \newcommand{\alg}{\mathcal{ALG}}
	
	\newcommand{\calV}{\mathcal{V}}
	\newcommand{\inst}{\Upsilon}
	\newcommand{\val}{\mathbf{val}}
	\newcommand{\supp}{\mathtt{supp}}
	\newcommand{\x}{\mathbf{x}}
	
	\newcommand{\opt}{\textsc{Opt}}

\SetKwComment{Comment}{$\triangleright$\ }{}
\SetKwComment{tcp}{$\triangleright$\ }{}

\begin{document}
\maketitle

\begin{abstract}
We propose a framework of algorithm vs.~hardness for all Max-CSPs and demonstrate it for a large class of predicates. This framework extends the work
of Raghavendra \cite{Rag08}, who showed a similar result for almost satisfiable Max-CSPs.

Our framework is based on a new {\em hybrid approximation algorithm}, which uses a combination of the Gaussian elimination technique (i.e., solving a system of linear equations over an Abelian group) and the semidefinite programming relaxation. We complement our algorithm with a matching dictator vs.~quasirandom test that has perfect completeness.

The analysis of our dictator vs.~quasirandom test is based on a novel invariance principle, which we call the {\em mixed invariance principle}. Our mixed invariance principle is an extension of the invariance principle of Mossel, O'Donnell and Oleszkiewicz \cite{MOO} which plays a crucial role in Raghavendra's work. The mixed invariance principle allows one to relate $3$-wise correlations over discrete probability spaces with expectations over spaces that are a mixture of Gaussian spaces and Abelian groups, and may be of independent interest.
\end{abstract}

 \section{Introduction}
\subsection{Constraint Satisfaction Problems}
The class of constraint satisfaction problems (CSPs in short) consists of some of the most studied computational problems in artificial intelligence, database theory, logic, graph theory, and computational complexity. Given a predicate  $P : \Sigma^k \rightarrow \{0,1\}$ for some finite alphabet $\Sigma$, a $P$-CSP instance consists of a set of variables $X = \{x_1, x_2, \ldots, x_n\}$ and a collection of {\em local} constraints $C_1, C_2, \ldots, C_m$, each one of the form $P(x_{i_1}, x_{i_2}, \ldots, x_{i_k}) = 1$. Here and throughout, we refer to the parameter $k$ as the arity of the CSP. For a class of predicates $\mathcal{P}\subseteq \{P\colon \Sigma^k\to\{0,1\}\}$, an instance of $\mathcal{P}$-CSP consists of a set of variables $X$ and a collection of constraints $C_1,\ldots,C_m$ each one of the form $P(x_{i_1}, x_{i_2}, \ldots, x_{i_k}) = 1$ for some $P\in\mathcal{P}$. The value of an instance $\inst$, denoted by $\val(\inst)$,  is the maximum fraction of the constraints that can be satisfied by an assignment to the variables.

The most natural decision problem associated with instances of $\mathcal{P}$-CSP is the satisfiability problem: given an instance $\inst$ of $\mathcal{P}$-CSP, determine if it is satisfiable, i.e., if there exists an assignment $A\colon X\to\Sigma$ satisfying all of the constraints of $\inst$. In a relaxation of this problem called the Max-$\mathcal{P}$-CSP problem (which is most relevant to this paper), one is given an instance $\inst$ of $\mathcal{P}$-CSP, and the task is to efficiently find an assignment to the variables that satisfies as many of the constraints as possible.
An $\alpha$-approximation algorithm is a polynomial-time algorithm that, given an instance $\inst$, finds an assignment satisfying at least $\alpha\cdot \opt(\inst)$ fraction of the constraints, where $\opt(\inst)$ is the value of the optimal assignment.
The study of CSPs has driven some of the most influential developments in theoretical computer science,
including
the theory of NP-completeness~\cite{Cook71, Levin73},
the PCP theorem~\cite{AroraLMSS1998, AroraS1998, FGLSS96},
the development of semidefinite programming algorithms~\cite{GW95, KMS98, Zwick98, Zwick99, Rag08, CharikarMM09},
the Unique Games Conjecture and its consequences~\cite{Khot02UG,Rag08},
the dichotomy theorem~\cite{FederV98,Bulatov17,Zhuk20} and more. Below, we
elaborate on two of these topics.

\paragraph{The Dichotomy Theorem:} a systematic study of the complexity of solving CSPs was started by Schaefer in 1978~\cite{Schaefer78}, who showed that for every predicate $P$ over a Boolean alphabet, the problem of checking satisfiability of a $P$-CSP instance is either in P or is NP-complete.
Note that this is not a trivial
statement: by Ladner's theorem~\cite{Lander}, if P$\neq$NP, then there are languages that are not in P nor are NP-hard; these are often called NP-intermediate problems.
Thus, another way of stating Schaefer's theorem is that CSPs over Boolean
alphabets cannot be NP-intermediate. Feder and Vardi~\cite{FederV98} conjectured that Schaefer's theorem holds
for all finite alphabets, and this statement is often referred to
as the Dichotomy Conjecture.
Much effort has gone into studying the dichotomy conjecture, mainly
using the tools of abstract algebra. The conjecture was recently resolved by Bulatov and Zhuk (independently)~\cite{Bulatov17, Zhuk20}, who proved that, indeed,
for any family of predicates $\mathcal{P}$, the decision problem $\mathcal{P}$-CSP is either in P or is NP-complete.

\paragraph{Approximating Almost Satisfiable Instances:} the complexity of approximating {\em almost satisfiable} instances is
rather well understood by now. Here and throughout, we say that an
instance $\inst$ is almost satisfiable
if $\opt(\inst)\geq 1-\eps$ where
$\eps>0$ is a small constant.

Some of the theory
here is based on the PCP theorem~\cite{FGLSS96,AroraS1998,AroraLMSS1998}.
As an example, an important result of
H\r{a}stad~\cite{Has01} states
that for all $\eps>0$, given an instance $\inst$ of the Max-$3$-Lin problem promised to have $\opt(\inst)\geq 1-\eps$, it is NP-hard to find an assignment satisfying at least $1/2+\eps$ fraction of the constraints. Here, the Max-3LIN
problem is the problem Max-$\mathcal{P}$-CSP where
$\mathcal{P} = \{P_0,P_1\}$
and $P_a\colon\mathbb{F}_2^3\to\{0,1\}$ is defined by $P_a(x,y,z) = 1_{x+y+z = a}$. In fact, H\r{a}stad's hardness result~\cite{Has01} also applies to the {\em decision version} of the problem: for every $\eps>0$, it is NP-hard to distinguish between the cases $\val(\inst)\geq 1-\eps$ and $\val(\inst)\leq 1/2+\eps$.

Getting a more comprehensive understanding of the approximability of almost satisfiable CSPs requires a stronger PCP characterization of NP, known as the Unique Games Conjecture (UGC)~\cite{Khot02UG}. Assuming UGC, Raghavendra~\cite{Rag08} showed that for every collection of predicates $\mathcal{P}$ and $\eps>0$, there is a constant $\beta_{\mathcal{P}}$ such that:
\begin{enumerate}
    \item Algorithm: there exists
    a polynomial-time algorithm that given an instance $\inst$ of $\mathcal{P}$-CSP promised
    to have $\opt(\inst)\geq 1-\eps$, outputs an assignment satisfying at least $\beta_{\mathcal{P}}$ fraction of the constraints. Clearly, this also means that there exists a polynomial-time algorithm that distinguishes instances with a value at least $1-\eps$ from instances with a value at most almost $\beta_{\mathcal{P}}$.
    \item Hardness: for all $\delta>0$, given an instance $\inst$, it is NP-hard to distinguish between the case that $\val(\inst)\geq 1-\eps-\delta$, and the case that $\val(\inst)\leq \beta_{\mathcal{P}}+\delta$.
\end{enumerate}
In words, Raghavendra's result
asserts that for every collection of predicates, the approximability
of almost satisfiable instances exhibits
a dichotomy between approximation
ratios that can be achieved in polynomial time, and those that
are NP-hard (assuming UGC).

\subsection{Approximating Satisfiable Instances}
The complexity of approximating satisfiable CSPs is much more complicated and remains mostly unsolved as of now, even under reasonable conjectures in the style of UGC~\cite{Khot02UG}. The proof of the dichotomy theorem implies that for $\mathcal{P}$ for which $\mathcal{P}$-CSP is NP-complete, there exists a constant $0<\delta_{\mathcal{P}}<1$ such that it is NP-hard to distinguish satisfiable instances from instances with value at most $\delta_{\mathcal{P}}$.  However, unlike the almost-satisfiable case, we do not know tight inapproximability results for satisfiable instances for every $\mathcal{P}$. There are only a few collections $\mathcal{P}$ for which we know the existence of $\alpha_{\mathcal{P}}$ for which  efficient $\alpha_{\mathcal{P}}$ approximation of Max-$\mathcal{P}$-CSP is possible, while  $\alpha_{\mathcal{P}}+\eps$ approximation is NP-hard. There are even fewer collections $\mathcal{P}$ for which we know the
value of $\alpha_{\mathcal{P}}$. An example for such problem is the $3$-SAT a problem, for which H\r{a}stad~\cite{Has01} proved that $\alpha_{3SAT} = 7/8$ works (and more generally $\alpha_{kSAT} = 1-1/2^k$ for the $k$-SAT problem). Another example is the NTW predicate\footnote{The accepting assignments of the predicate NTW are $\{(0,0,0), (0,0,1), (0,1,0), (1,0,0), (1,1,1) \}$.}~\cite{OW_dict_09, OW_hard_09, KSaket06}, for which the optimal threshold of $\alpha_{NTW} = 5/8$ was shown by H{\aa}stad~\cite{Hastad14}. The
problem of determining the existence of the ratio $\alpha$
and pinning it down gets very difficult very quickly though; see~\cite{DBLP:conf/soda/BrakensiekHPZ21} for the case of the NAE predicate.

\subsection{The Dichotomy Approximation Conjecture}
Motivated by the dichotomy theorem and Raghavendra's theorem discussed above, in~\cite{BKMcsp1} the authors
suggested the following statement, referred to as the \emph{Approximation Dichotomy Conjecture}:
\begin{conjecture}[Approximation Dichotomy Conjecture]
\label{conj:ADC}
    For all $k\in\mathbb{N}$ and for all collections of $k$-ary predicates $\mathcal{P}$ for which $\mathcal{P}$-CSP is NP-hard, there exists a constant $\alpha_{\mathcal{P}}$ such that:
    \begin{enumerate}
        \item Algorithm: there is a polynomial-time algorithm that distinguishes satisfiable instances of Max-$\mathcal{P}$-CSP from instances with value at most $\alpha_{\mathcal{P}}$.
        \item Hardness: for all
        $\eps>0$, it is NP-hard to
        distinguish satisfiable instances of Max-$\mathcal{P}$-CSP from instances with value at most $\alpha_{\mathcal{P}} + \eps$.
    \end{enumerate}
\end{conjecture}
In words, the approximation dichotomy conjecture states that the
complexity of the approximating
Max-$\mathcal{P}$-CSP exhibits a
rapid phase transition between
approximation ratios that can
be achieved by polynomial-time
algorithms, and approximation ratios that are NP-hard to achieve. In other words, the approximation problem is never NP-intermediate.

\newcommand{\define}[1]{{\color{cyan}#1}}
A sequence of works~\cite{BKMcsp1,BKMcsp2,BKMcsp3,BKMcsp4} made progress on the case $k=3$ of Conjecture~\ref{conj:ADC}, focusing on the ``hardness'' part; we elaborate on these works below. The goal of the current paper is to make progress
on the ``algorithmic'' part of
Conjecture~\ref{conj:ADC}.
In particular, we propose an approximation algorithm
for a wide class of CSPs,
and develop tools to bridge between its performance and the ``hardness'' side of the conjecture.

\subsubsection{Why \texorpdfstring{$3$}{3}-ary Predicates?}
In full generality, Conjecture~\ref{conj:ADC}
is likely
to be very difficult to settle. Settling it, even for certain classes of CSPs, requires one to study associated,
very general analytical problems.
These problems include within
them (as subcases) the inverse
theorems for Gowers' uniformity
norms over finite fields~\cite{bergelson2010inverse,tao2010inverse,tao2012inverse}. In fact, the resulting analytical
problem for $k$-ary predicates
implies a generalization of the
inverse theorem for Gowers' $U^{k-1}$-norms. For instance,
the main result of~\cite{BKMcsp4}
is a solution to this analytical problem for $k=3$, and it
can be used to study the density of restricted $3$-AP free sets~\cite{BKM}, a result which was previously unknown. As the
difficulty in the study of Gowers'
uniformity norms sharply increases
when $k$ is large, it stands to  reason that the associated analytical
problems we study also become
more challenging as $k$ increases. Hence, we focus on the simplest
case that is not understood, $k=3$, which is already highly non-trivial.
It is also reasonable
to expect that the resolution of the analytical problem for general $k$ will ultimately proceed by induction
on $k$, in which scenario the case $k=3$ will be the base
case of this induction.\footnote{
Many of the proofs of the
inverse theorem for Gowers' uniformity norms indeed proceed in this fashion.}

We remark that prior works,
and more specifically Raghavendra's theorem~\cite{Rag08} as well as
the dichotomy theorem~\cite{Bulatov17,Zhuk20},
have side-stepped these issues. First, due to the imperfect completeness
Raghavendra requires only fairly
simple inverse theorems (corresponding to distributions with full support), and the argument is essentially the same for all $k$.
Second, in the context
of the dichotomy theorem, one does
not have to worry about preserving
approximation ratios, and the
difference between this case
and our case is analogous to the
difference between the relatively elementary Hales-Jewett theorem~\cite{hales2009regularity} and the more challenging density Hales-Jewett theorem~\cite{DHJ1,DHJ2,DHJ3}.

\subsubsection{Approximation Algorithms vs Hardness Reductions}
To characterize the right approximation threshold for Max-$\mathcal{P}$-CSP, one has to work both on the algorithmic front as well as on the hardness front, and make them meet. In this section
we discuss the way it is manifested in almost satisfiable
CSPs, and the difference between
that and our setting of satisfiable
CSPs.

\paragraph{Dictatorship tests
as evidence for hardness:}
dictatorship tests
are one of the most important
components in hardness of
approximation results. A function $f: \Sigma^n \rightarrow \Sigma$ is called a dictatorship function if there is $i\in[n]$ such that $f(x) = x_i$ for all $x$. A dictatorship test is a procedure that queries $f$ at a few (correlated) locations randomly, and based on these values decides if $f$ is a dictator function or {\em far} from any dictator function. For brevity, we often refer to the latter type of functions as {\em quasirandom} functions.

We briefly describe the notion of being {\em far} from dictator functions here. The influence of a coordinate $i$ on a function $f\colon (\Sigma^n, \nu^{n})\to \mathbb{C}$ is the amount it affects the value of $f$ at a random input,
i.e.,
\[
I_i[f] = \Expect{\substack{(x_1,\ldots,x_n)\sim \nu^{n}\\x_i'\sim \nu}}
{\card{f(x_1,\ldots,x_i,\ldots,x_n) - f(x_1,\ldots,x_i',\ldots,x_n)}^2}.
\]
 A function is called far from dictatorship functions if, for every coordinate $i$, the influence of the coordinate $i$ in $f$ is small. There are three important properties of a dictatorship test that are useful in getting the hardness of approximation result for Max-$\mathcal{P}$-CSP:
\begin{enumerate}
    \item The completeness $c$, which is the probability that
    the test accepts any dictatorship
    function.
    \item The soundness $s$, which is the probability that the test accepts any function which is far from being a dictatorship.
    \item The check that the
    test makes: if the dictatorship test makes
    $k$ queries, say to the points $x(1),\ldots,x(k)$, and performs
    a check of the form
    $P(f(x(1)),\ldots,f(x(k))) = 1$, for $P\in \mathcal{P}$, then it can be used to
    get a hardness result for Max-$\mathcal{P}$-CSP.
\end{enumerate}
Typically, a
dictatorship test with parameters
$0<s<c\leq 1$ that uses a collection of predicates $\mathcal{P}$, can be converted to
an NP-hardness result for the
following promise problem: given
an instance $\inst$ of CSP-$\mathcal{P}$ promised to be at least $c$-satisfiable, find an assignment
satisfying at least $s$ fraction
of the constraints. While this transformation is far from being automatic in general,\footnote{Often times one requires a plausible complexity-theoretic assumption, such as the Unique-Games Conjecture~\cite{Khot02UG} or
the Rich-2-to-1 Games Conjecture~\cite{BravermanKM21}.}
dictatorship tests are often
thought of as strong evidence
towards this form of hardness
of approximation result.
Thus, the hardness part of Conjecture~\ref{conj:ADC}  requires dictatorship
tests with $c=1$, which we refer
to as perfect completeness.

\paragraph{Dictatorship tests with perfect completeness:} the
papers~\cite{BKMcsp1,BKMcsp2,BKMcsp3, BKMcsp4} develop tools to analyze dictatorship tests
with perfect completeness. This
case turns out to be considerably
more complicated than the case of
imperfect completeness. Whereas
in the latter case, one can change
the dictatorship test slightly
(to gain desirable features from it) at only a mild cost in
the completeness parameter,
this is unaffordable once we
take $c=1$. In particular,
the paper~\cite{BKMcsp4} gives
a ``quasirandom'' versus ``dictatorship'' result for
a wide class of $3$-ary predicates.

\paragraph{Matching the dictatorship test via an algorithm:}
to match the performance of the
dictatorship test, an approximation algorithm has to gain insights from the soundness analysis of the test. At a high
level, the algorithm has to be
able to utilize any quasirandom function in an algorithmic way.

In the context
of Raghavendra's theorem~\cite{Rag08},
the class of quasirandom functions consists of
low-degree functions with
no influential coordinates~\cite{Mossel}.
To utilize these functions algorithmically, the powerful
invariance principle of~\cite{MOO}
is used, asserting that these
type of functions essentially
come from Gaussian space.
As Gaussian samples
can be produced algorithmically
by solving the semi-definite
programming relaxation and multiplying by vector-valued
Gaussian random variables,
this gives rise to an algorithm.

In the context of satisfiable CSPs, the class of quasirandom
functions is more rich in general.
In fact, this class may even depend on the
specific predicate in question. In many cases of interest this class includes low-degree
functions as well as Fourier characters, and it is
not immediately clear how to use
these types of functions algorithmically. The invariance
principle of~\cite{MOO} fails
for these types of functions,
and more information
(besides the one gained from the semi-definite programming relaxation) seems necessary
for an algorithm.
The main contribution of the
current work is to propose
such an algorithm and prove
a generalization of the invariance
principle that facilitates the
use of this type of functions algorithmically.

\subsection{Our Main Result}
We now state the main result
of this paper, and towards this
end we require a few definitions.
We begin with the notion of
Abelian embeddings, also
referred to as linear embeddings.
\begin{definition}
\label{def:lin_embed1}
Let $\Sigma_1,\ldots,\Sigma_k$
be finite alphabets and let
$A\subseteq \prod\limits_{i=1}^{k}\Sigma_i$.
For an Abelian group $(G,+)$, we say maps
$\sigma_i\colon \Sigma_i\to G$ for
$i=1,\ldots, k$
form an Abelian embedding of $A$
if
\[
(a_1,\ldots,a_k)\in A
\implies
\sum\limits_{i=1}^k\sigma_i(a_i) = 0_G.
\]
We say $A$ is Abelianly
embeddable if there are maps
$\sigma_i$ that are not all
constant that form an Abelian
embedding of $A$. We say that a distribution $\mu$ over $\prod\limits_{i=1}^{k}\Sigma_i$ is Abelianly embeddable if ${\sf supp}(\mu)$ is.
\end{definition}
The notion of
Abelian embeddings is central
to the study
of satisfiable CSPs, and the
existence of embeddings into a
group $G$ should be thought of
as hinting towards a subspace-type structure inside $A$.
For technical reasons though,
we are only able to handle Abelian groups $G$ that are finite. Thus, we say that $A$
admits a $(\mathbb{Z},+)$ embedding if it has an Abelian embedding into the group $(\mathbb{Z},+)$.
With this in mind, we
wish to define the class of predicates
we handle in the current paper, and we begin by
giving a few examples.
\begin{enumerate}
    \item {\bf Vector valued punctured $3$-Lin:} let $p$ be a prime number, let $m\geq 2$ and take $\Sigma =\mathbb{F}_p^m\setminus \{\vec{0}\}$. Take the collection of predicates $\mathcal{P} = \{P\}$, where the set of satisfying assignments of $P$ consists of tuples $(x, y, z)$ such that  $x, y, z\in \mathbb{F}_p^m\setminus \{\vec{0}\}$, each pair of $x,y,z$ is linearly independent and $x+y+z=\vec{0}$. It is easy
    to observe
    that $P$ admits an Abelian
    embedding into  $\mathbb{F}_p^m$. There are
    other Abelian embeddings that
    are induced by it, such as
    $\sigma_1(x) = \inner{\alpha}{x}$, $\sigma_2(y)=\inner{\alpha}{y}$, $\sigma_3(z) = \inner{\alpha}{z}$ for any
    $\alpha\in\mathbb{F}_p^m$.

    The predicate $P$ does not admit a $(\mathbb{Z},+)$
    embedding, and furthermore,
    there is a group action on
    $\Sigma$ that preserves satisfying
    assignments of $P$. Namely,
    for each invertible $M\in\mathbb{F}_p^{m\times m}$
    we can take $\tau_M\colon \mathbb{F}_p^m\to\mathbb{F}_p^m$
    defined as $\tau_M u = Mu$.
    It is easy to see that for
    each such $M$, if $P(x,y,z) = 1$, then
    $P(\tau_M(x), \tau_M(y), \tau_M(z)) = 1$. Using this,
    one can show that the collection $\mathcal{P}$ is {\symm} as defined below, and our results therefore apply to $\mathcal{P}$.

    \item {\bf $3$-Uniform hypergraph $p$-strong coloring:} let $p\geq 5$ be a prime, take $\Sigma = \mathbb{F}_p$ and consider the
    predicate $P\colon \Sigma\to\{0,1\}$ where $P(x,y,z) = 1_{x,y,z\text{ are distinct}}$ and $\mathcal{P} = \{P\}$. The problem Max-$\mathcal{P}$-CSP is also known~\cite{BrakensiekG14} as the $p$-strong coloring problem: an instance can be viewed as a hypergraph (whose hyperedges are the triplets that are involved in some constraint), and a solution can be viewed as coloring the graph maximizing the fraction of hyperedges with all distinct colors.

    The predicate $P$ does not
    admit Abelian embeddings,
    however its support does
    contain subsets that
    do admit Abelian embeddings.
    The collection $\mathcal{P}$
    can be seen to be {\symm}
    (see Section~\ref{sec:proof_that_symm} for a proof), and so
    our result applies to $\mathcal{P}$.
    \item {\bf $3$-Uniform hypergraph $3$-rainbow coloring:} consider the previous example except that we take $p=3$. This problem is also known~\cite{GuruswamiL18} as the rainbow hypergraph coloring problem, where a hyperedge is properly colored if {\em all} the colors are present. This time the predicate $P$ has Abelian embeddings, such as
    $\sigma,\gamma,\phi\colon \Sigma\to \mathbb{F}_3$ defined as $\sigma(x) = x$, $\gamma(y) = y-1$, $\phi(z) = z-2$. In fact, this is even a $(\mathbb{Z},+)$ embedding.
    We also note that
    there is a group action $\{\tau_{a,b}\}_{a,b\in \mathbb{F}_3, a\neq 0}$
    on $\Sigma$ defined as $\tau_{a,b}(u) = au + b$
    that preserves the satisfying assignment of $P$.

    Because $P$ admits a $(\mathbb{Z},+)$ embedding, the collection $\mathcal{P}$ is not {\symm} as defined below, and our result does not
    apply to it.
    As discussed in Sections~\ref{sec:on_symm} and~\ref{section:rel_work}
    we view the $(\mathbb{Z},+)$-embedding obstruction as an
    issue of a technical nature (as opposed to the issue of pairwise connectivity, which appears much more fundamental).
    Therefore, we suspect that
    a variant of our result should
    hold for this predicate as
    well.
\end{enumerate}

Having looked at a few examples and non-examples, we now formally define the class of {\symm}
predicates as follows:
\begin{definition} [{\symm} predicates]
\label{def:symm_p}
    A family of predicates
    $\mathcal{P}\subseteq \{P:\Sigma^3 \rightarrow \{0,1\}\}$
    is called {\symm} if there are actions $\tau_1, \tau_2, \ldots, \tau_\ell : \Sigma \rightarrow \Sigma$ such that:
    \begin{enumerate}
        \item For every $P\in\mathcal{P}$, every $i\in[\ell]$ and every satisfying assignment $\sigma \in \Sigma^3$ of $P$, the assignment $(\tau_i(\sigma_1),\tau_i(\sigma_2),\tau_i(\sigma_3))$ is a satisfying assignment of $P$.
        \item For every $P\in\mathcal{P}$ and for every satisfying assignment
        $\sigma \in \Sigma^3$ of $P$, the set 
        \[
        \{(\tau_i(\sigma_1), \tau_i(\sigma_2), \tau_i(\sigma_3)) \mid i\in [\ell]\}\subseteq \Sigma^3
        \]
        does not have a $(\mathbb{Z},+)$-embedding.
    \end{enumerate}
\end{definition}
In words, a collection of predicates $\mathcal{P}$ is
called {\symm} if there are
maps on the alphabet $\Sigma$
that both (1) preserve the satisfying
assignments of all predicates
in $\mathcal{P}$, and (2) the orbit of each satisfying assignment under the maps
is rich enough that it doesn't
admit any $(\mathbb{Z},+)$ embedding.

\skipi

Our main result with regard to
Conjecture~\ref{conj:ADC} is
the following statement.
\begin{theorem}\label{thm:main}
      Let $\mathcal{P}$ be a collection of {\symm} $3$-ary predicates. Then there exists $\alpha_{\mathcal{P}}$
      such that for all $\eps>0$:
      \begin{enumerate}
          \item Hardness: there is
          a dictatorship
          vs quasirandom  test using $\mathcal{P}$ with perfect completeness and
          soundness $\alpha_{\mathcal{P}}+\eps$.
          \item Algorithm:
          there exists a polynomial-time algorithm that distinguishes between satisfiable instances of $\mathcal{P}$-CSP from instances of $\mathcal{P}$-CSP with value at most $\alpha_{\mathcal{P}}$.
      \end{enumerate}
 \end{theorem}
\paragraph{Organization:}
the rest of this introductory
section is organized as follows.
In Section~\ref{sec:appx_alg_csp}
we discuss the algorithmic approach
for CSPs, and in Section~\ref{sec:new_approx_alg} we
present our \emph{hybrid algorithm}.
In Section~\ref{sec:analysis_of_hybrid} we discuss the analysis of
the hybrid algorithm and in
Section~\ref{section:mixed_inv_intro} we discuss our main technical
contribution, the {\em mixed
invariance principle}.
In Section~\ref{section:rel_work} we
discuss other related works.

\subsection{Approximation Algorithms for CSPs}\label{sec:appx_alg_csp}
In this section we discuss two algorithmic
techniques that are vital
towards our hybrid algorithm,
and are used in
Raghavendra's theorem and
in the dichotomy theorem.

\subsubsection{Raghavendra's Algorithm: Semi-definite
Programming Relaxations}
For almost satisfiable CSPs,
Raghavendra~\cite{Rag08} showed that {\bf semi-definite programming} (SDP) based algorithms give the optimal approximation algorithms (assuming the UGC). More specifically, we can write down a basic SDP relaxation of a given instance of Max-$\mathcal{P}$-CSP as shown in Table~\ref{fig:sdp+gaussianE}. Here, $\mathcal{V}$ is the set of
variables of the instance and
$\mathcal{C}$ is a distribution
over constraints of the instance
(representing a weighted instance of CSP-$\mathcal{P}$), and for each $c\in\mathcal{C}$, $\mathcal{V}(c)$ is the set of
variables appearing in $c$. For each
variable $i\in\mathcal{V}$ and
an alphabet symbol $a\in\Sigma$
the program has a vector
$\V b_{i,a}$ and additionally
there is a global vector $\V b_0$.
We think of these vectors as
describing a distribution over
good assignments to the instance,
and write down conditions corresponding to that.

The SDP-solution is then rounded, via a non-trivial rounding procedure, to an assignment to the variables. More precisely,
given a solution to the SDP program that lives in dimension $m$, Raghavendra samples Gaussians $g^{(1)},\ldots,g^{(R)}\in \mathbb{R}^m$ in which each coordinate is an independent standard Gaussian random variable, and produces jointly
distributed Gaussians $z_{\ell,(i,a)} = \langle g^{(\ell)}, \V b_{i,a}\rangle$. It is easy to
see that the $z$'s pairwise
correlations
match the inner products
of the SDP solution vectors.
Using these samples (as inputs),
Raghavendra shows that any
quasirandom function that
performs well in the dictatorship
tests yields a rounding function
that satisfies (in expectation)
at least $s$ fraction of the
constraints, where $s$ is
the soundness of the dictatorship
test. We stress that in
this context, quasirandom functions refer to low-degree
functions in which all variables
have small influence.

It is, therefore, natural to ask if
semi-definite programming relaxation also gives the best
approximation algorithm in
the setting of Theorem~\ref{thm:main}. This turns out to be false as can be seen from the following CSP. Let $(G, \cdot)$ be a non-Abelian group and consider the problem $3$-Lin$_{G}$. In this problem,
we have variables $x_1,\ldots,x_n$ that are supposed to be assigned values from $G$, and the constraints are of the form $x_i\cdot x_j \cdot x_k = c$ where $c\in G$ are constants. In~\cite{BenabbasGMT12} it is shown that the basic SDP for $3$-LIN$_{G}$ has an integrality gap of $1$ vs.\ $1/|G|$, meaning that
an algorithm that only uses SDP rounding cannot achieve a better ratio than $1/|G|$ on satisfiable instances. However, there is an algorithm that achieves $1/|[G,G]|$ factor approximation on satisfiable instances, where $[G,G]$ is a commutator subgroup of $G$; see~\cite{BKnonabelian} for example.\footnote{For many non-Abelian groups $G$, such as for the group $G=S_m$ for example, the size of the commutator subgroup $[G,G]$ is strictly smaller than the size of $G$. In particular, a $1/|[G,G]|$-approximation is strictly better than a $1/|G|$-approximation.}

\subsubsection{Gaussian Elimination}
Another important algorithmic technique is {\bf Gaussian elimination}, i.e., solving a system of linear equations over an Abelian group~\cite{EngebretsenHR}.
Gaussian elimination can at times be more powerful than
semi-definite programming
relaxations: using it one
can decide whether a given
system of linear equations
over $\mathbb{F}_p$ is perfectly
satisfiable or not, but SDPs fail
to do so~\cite{Grigoriev,Schoenebeck}.
Still, by itself it is rather weak, and it is not clear how
to use it to obtain non-trivial
approximation algorithms for
problems such as Max-Cut.
We remark that Gaussian elimination is not enough to check the satisfiability of bounded width $P$-CSPs~\cite{BartoK09}, which otherwise are tractable using local-propagation algorithms~\cite{BartoK09}.

\subsubsection{Our Hybrid Algorithm}
Our hybrid algorithm blends the two aforementioned algorithmic tools in a nontrivial way. We note that these two techniques have recently been used together to produce nontrivial algorithms in the area of {\em promise CSPs}. Indeed, therein a combination of semidefinite program/linear program (SDP/LP) and affine integer program (AIP) was used~\cite{BrakensiekG19, BrakensiekGWZ20, CiardoZ_CLAP23} to solve a few tractable cases. Similarly, SDP+AIP was shown~\cite{CiardoZ24} not to be enough to solve the approximate graph coloring problem. However, these prior works consider solving the SDP and the system of linear equations once and using these solutions to output the final decision. In contrast, our hybrid algorithm iteratively modifies the SDP program and the system of linear equations before coming up with the final decision.

\subsection{A New Approximation Algorithm for Satisfiable CSPs}\label{sec:new_approx_alg}
	In this section we present our hybrid algorithm that will be used in the analysis of our dictatorship test.

Let $\mathcal{P}$ be a collection of $3$-ary predicates, all of which are embeddable in a finite Abelian group $G$ via the map $\sigma\colon \Sigma\to G$.\footnote{For simplicity, we assume here that this is essentially the only group in which the predicates are embeddable and all the embedding maps are identical, i.e., $\sigma_1 = \sigma_2 = \sigma_3 = \sigma$.}
Fix an instance $\inst = (\calV, \calC)$ of Max-$\mathcal{P}$-CSP, where $\calV$ is identified with the set $\{1, 2, \ldots, n\}$, and $\calC$ is a set of constraints on $\calV$. The basic semidefinite programming relaxation of the instance is given on the left side of Table~\ref{fig:sdp+gaussianE}. It consists of vectors $\{\V b_{i,a}\}_{i\in \calV, a\in \Sigma}$, distributions $\{\mu_c\}_{c\in \calC}$ over the local assignments (i.e., on $\Sigma^{\calV(c)}$, where $\calV(c)$ denotes the tuple of variables in $c$) and a unit vector $\V b_0$. The notation $\simplex(Z)$ refers to the collection of all probability distributions over the set $Z$. Observe that this is indeed a relaxation: given a satisfying assignment $\alpha: \calV \rightarrow \Sigma$, choose $\V b_0$ to be some unit vector, and consider the vector assignment
\begin{align}
    \V b_{i,a}= \left\{\begin{array}{cc}
          \V b_0 & \mbox{ if } \alpha(i) = a,\\
          \vec{0} & \mbox{ otherwise.}
    \end{array}\right. \label{eq:intor_sdp_integral}
\end{align}
For every $c\in \calC$, the distribution $\mu_c$ is set to be supported on the assignment $\alpha|_{\calV(c)}$, which by definition satisfies the constraint $c$. It can be easily observed that all the constraints (1)--(4) are satisfied, and furthermore, the objective value of this vector solution is $1$, as $c(x) = 1$ if $x\sim \mu_c$ for every $c\in \calC$.

\begin{table}[t]
	\begin{tabular}{|c|c|}
		\hline
		Semidefinite Program & System of linear equations over $G$ \\
		\hline
		$\begin{aligned}
			&\mbox{maximize}\quad \E_{c\in \calC} \E_{x\sim \mu_{c}} [c(x)] \nonumber\\
		&\mbox{subject to } \\
		&(1) \quad \langle \V b_{i,a}, \V b_{j,b}\rangle = \Pr_{x\sim \mu_{c}}[x_i = a, x_j = b] &\\
		& \quad\quad\quad\quad\quad\quad \quad  c\in \calC,\quad i,j\in \calV(c), \quad a,b\in \Sigma, \\
		&(2) \quad \langle \V b_{i,a}, \V b_0 \rangle = \|\V b_{i,a}\|_2^2,  \forall i\in \calV, a\in \Sigma,\\
		&(3) \quad \|\V b_0 \|_2^2 = 1,\\
		&(4) \quad \mu_{c} \in \simplex(\Sigma^{\calV(c)}),  c\in \calC.
		\end{aligned}$
		&
		$\begin{aligned}
		&\mbox{Find $\vartheta:\{y_1, y_2,\ldots, y_n\} \rightarrow (G, +)$ such that}\\
		& \forall c\in \calC \mbox{ with } \calV(c) = (i_1, i_2, i_3), \mbox{ we have } \\
        \\
		& \quad\quad \vartheta(y_{i_1}) + \vartheta(y_{i_2}) + \vartheta(y_{i_3}) = 0_G.
		\end{aligned}$ \\
		\hline
	\end{tabular}
	\caption{A semidefinite programming formulation and a system of linear equations for a given Max-$P$-CSP instance $\inst = (\calV, \calC)$.}
	\label{fig:sdp+gaussianE}
\end{table}

We also write down a system of linear equations over an Abelian group $G$, given on the right-side of Table~\ref{fig:sdp+gaussianE}. Note that if the instance $\inst$ is satisfiable, then the system of linear equations over $G$ also has solutions. This follows from the definition of Abelian embeddability (Definition~\ref{def:lin_embed1}): take any satisfying assignment $\alpha\in \Sigma^n$ and assign $\vartheta(y_i) = \sigma(\alpha_i)$ for all $i\in [n]$.

We modify the semi-definite program and the system of linear equations iteratively as follows. At every step, we work with an SDP solution where for every satisfying assignment $\alpha: \calV \rightarrow \Sigma$ to the instance $\inst$, $\mu_c(\alpha|_{\calV(c)}) >0$.\footnote{This can be achieved in polynomial time, see Lemma~\ref{lemma:includes_all_sat_assignments}.} In other words, every satisfying assignment to the instance `survives' in the SDP solution.
\begin{itemize}
    \item For a constraint $c\in \calC$ we say that $(g_1,g_2,g_3)\in G^3$ is SDP-unattainable if 
    \[(\sigma(a_1),\sigma(a_2),\sigma(a_3)) \neq (g_1,g_2,g_3)
    \]
    for all $(a_1,a_2,a_3)\in {\sf supp}(\mu_c)$. Consider an SDP-unattainable tuple $(g_1,g_2,g_3)$ not in a subgroup (of $G^3$) generated by the SDP-attainable tuples.  We modify the system of linear equations by eliminating $(g_1,g_2,g_3)$ from being a possible setting to the variables corresponding to $c$ while preserving all the SDP-attainable tuples in a solution.
    \item For every constraint $c\in \calC$ and an assignment $\V a = (a_1, a_2, a_3) \in \Sigma^3$ if no solution to system of linear equations assigns $(\sigma(a_1), \sigma(a_2), \sigma(a_3))$ to the variables of $c$, then add a constraint $\mu_c(\V a) = 0$ to the SDP formulation.
\end{itemize}

It is easy to see that the process ends in polynomially many steps. The algorithm accepts an instance if the SDP value is $1$ at the end of the above procedure. Since we ensured that every satisfying assignment to the instance `survives' in the SDP solution, for a satisfying assignment $\alpha$, the vector solution given in~\eqref{eq:intor_sdp_integral}) will be a feasible solution to the final SDP. Therefore, the algorithm always accepts satisfiable instances of Max-$\mathcal{P}$-CSP.

As for the soundness of the algorithm, let $\mathcal{S}$ be the collection of all
Max-$\mathcal{P}$-CSP instances
where the algorithm accepts,
and define $s = \inf_{\inst\in \mathcal{S}}\val(\inst)$.
To complete the proof of
Theorem~\ref{thm:main} we
show a $1$ vs.\ $s+o(1)$
dictatorship test. We explain
this in Section~\ref{sec:analysis_of_hybrid}.

\subsection{The Dictatorship Test}\label{sec:analysis_of_hybrid}
In this section we explain how to
use an integrality gap of the above hybrid algorithm to construct a dictatorship test.
Throughout, we have an integrality gap $\inst$
with $\val(\inst) = s+o(1)$.
To analyze the dictatorship test, it will be useful for us to
consider the following approximation algorithm (we stress that it is specific for the instance $\inst$):
\begin{enumerate}
    \item \emph{Solving the programs.} Solve the final SDP program such that the objective value of the solution is $1$. As there are $n|\Sigma|+1$ vectors, we can assume without loss of generality that these vectors live in $\mathbb{R}^m$ for $m\leq  n|\Sigma|+1$. Denote by $S\subseteq G^n$ be the set of all the satisfying assignments $\vartheta$ for this system of linear equations,
    and note that $S$ is a subgroup.
    \item \emph{Setting up rounding functions and sampling.} For a constant $R\geq 1$ to be chosen as large enough, we fix a rounding function $f: \mathbb{R}^{R|\Sigma|} \times G^R\rightarrow \simplex(\Sigma)$, such that $f(\V z, \V w)$ gives a probability distribution on $\Sigma$. We sample $R$ Gaussian vectors $\bg^{(1)}, {\bg}^{(2)}, \ldots, {\bg}^{(R)} \in \mathbb{R}^m$ where for each $j\in [m]$ and $\ell\in [R]$, $g^{(\ell)}_j$is distributed according to the standard normal variable. We also sample $R$ uniformly random satisfying assignments over $G$ from the set $S$, which we denote by ${\V \sigma}^{(1)}, {\V \sigma}^{(2)}, \ldots, {\V \sigma}^{(R)}$.
    \item {\em SDP rounding component.} We first take the inner product of the Gaussian vectors with the SDP vectors corresponding to the variable $i$. More formally, let $z_{i,(\ell,a)} = \langle \V b_{i,a}, \bg^{(\ell)}\rangle$, for every $i\in\mathcal{V}$, $\ell\in [R]$ and $a\in \Sigma$. This gives a vector $\V z_i = (z_{i,(\ell,a)})_{\ell\in [R], a\in\Sigma}\in \mathbb{R}^{R|\Sigma|}$
    for each variable $i\in \mathcal{V}$.

    \item {\em Gaussian elimination component.} Taking the assignments we sampled from $S$, we create a string from $\V w_i \in G^R$ for each variable $i\in\mathcal{V}$, where $w_i = (\V \sigma^{(\ell)}(y_i))_{\ell\in [R]}$.
    \item \emph{Outputting an assignment.} For each $i\in \mathcal{V}$, sample $A(i)$ according to the distribution $f(\V z_i, \V w_i)$.
\end{enumerate}

\subsubsection{A Dictatorship Test with Imperfect Completeness}

Towards the construction of the dictatorship test, it is helpful
to first analyze a variant of the above algorithm that only solves the
SDP program and applies a rounding
scheme (or alternatively, that applies a rounding scheme that
ignores its $\V w$-component).
This is the setting in Raghavendra's theorem, and he uses it to construct a dictatorship test
with completeness $1-\eps$ and soundness $s+o(1)$ that uses the collection $\mathcal{P}$, for all $\eps>0$.
Denote by $(\V b, \V \mu)$ the SDP solution with value $1$
and consider the following test to check if a given function $f: \Sigma^R \rightarrow \Sigma$
is a dictator function or far from a dictator function:
\begin{enumerate}
    \item Sample $(\V y_1, \V y_2, \ldots, \V y_k)$ as follows:
    \begin{enumerate}
        \item Sample a random constraint $c\in \calC$
        and let $P\in\mathcal{P}$
        be the predicate it uses.
        \item For each $i\in [R]$, sample $(y_{1,i},  y_{2, i}, \ldots, y_{k, i})$ according to $\mu_c$ independently.
         \item For each $i\in [R]$, with probability $\eps$ resample $(y_{1,i},  y_{2, i}, \ldots, y_{k, i})$ from $\Sigma^k$ uniformly and independently.
    \end{enumerate}
    \item Check if $P(f(\V y_1), f(\V y_2), \ldots, f(\V y_k)) = 1$.
\end{enumerate}

\noindent{\bf Completeness:} If $f$ is a dictator function, say $f(\V y) = y_j$ for some $j\in [R]$, then the probability that the test passes is as follows:
\begin{align*}
    \Pr[\mbox{Test passes}] &= \E_{(c,P)}
     \left[\E_{(\V y_1, \V y_2, \ldots, \V y_k)}\left[P(f(\V y_1), f(\V y_2), \ldots, f(\V y_k)) \right]\right]\\
     & \geq (1-\eps)\E_{(c,P)}
     \left[\E_{(y_{1,j},  y_{2, j}, \ldots, y_{k, j})\sim \mu_c}\left[P(y_{1,j},  y_{2, j}, \ldots, y_{k, j}) \right]\right]\\
     & = (1-\eps) \cdot 1,
\end{align*}
where the last equality follows from the fact that the SDP value is $1$ and hence $\mu_c$ is supported on $P^{-1}(1)$ for every $c\in \calC$.
\skipi

\noindent{\bf Soundness:} Let us now analyze the soundness of the test. Towards this, suppose $f\colon \Sigma^R\to\Sigma$ is a quasirandom function.\footnote{By that, we mean that for all $a\in \Sigma$, the function $1_{f(x) = a}$ is quasirandom.} Our goal is to show that the dictatorship test accepts with probability at most $s+o(1)$. First, we express the test passing probability as follows:
\begin{align*}
    \Pr[\mbox{Test passes}] &= \E_{(c,P)}
    \left[\E_{(\V y_1, \V y_2, \ldots, \V y_k)}\left[P(f(\V y_1), f(\V y_2), \ldots, f(\V y_k)) \right]\right].
\end{align*}
Fix $c$ and $P$ and focus on the
inner expectation.
By expressing $P$ in terms of its multi-linear extension, the expectation above can be written as a linear combination of expectations of the form
\begin{equation}
\label{eq:exp}
    \E\left[\prod_{i\in S} F_i(\V y_i)\right],
\end{equation}
where $S\subseteq [k]$ and $F_i: \Sigma^R\rightarrow \{0,1\}$ is an indicator function of the form $F_i(\V y) = 1_{f(\V y) = a}$ for some $a\in\Sigma$. An important point is that the analysis of the test boils down to analyzing the expectation of product of functions where the output of each function is bounded. From this point onwards, Raghavendra's
argument~\cite{Rag08} proceeds as follows
% finishes the optimality of the SDP based algorithm on a almost satisfiable instances as follows
(for the simplicity of presentation we are omitting from the description many important technical details, such as how to keep the functions we work with bounded):
\begin{enumerate}
    \item[I.]
    For $\eps>0$, as the distribution on $(y_{1,j},  y_{2, j}, \ldots, y_{k, j})$ has full support, a result of~\cite{Mossel} asserts the expectation~(\ref{eq:exp}) can be approximately computed by only considering the low-degree part $F_i^{\leq d}$ of the corresponding $F_i$s, that is,
    \begin{equation}\label{eq:key_ld_only_matters}
    \E \left[\prod_{i\in S} F_i(\V y_i)\right] \approx  \E \left[\prod_{i\in S} F_i^{\leq d}(\V y_i)\right]  ,
    \end{equation}
    for some $d= O_\eps(1)$.
    \item[II.] Using the fact that $F_i^{\leq d}$s are low-degree and far from dictator functions, the invariance principle of~\cite{MOO} states that the inputs $F_i^{\leq d}$ can be ``replaced'' by correlated Gaussian random variables that
    have matching pairwise
    correlation to the $y_i$'s,
    i.e.,
    \begin{equation}\label{eq:cor_of_inv}
    \E \left[\prod_{i\in S} F_i^{\leq d}(\V y_i)\right] \approx \E \left[\prod_{i\in S} F_i^{\leq d}(\V g_i)\right].
    \end{equation}

    \item[III.] The process of sampling the correlated Gaussian can be simulated by the SDP rounding component as stated above. Thus, using
    the above observations, an
    algorithm can generate Gaussian samples such that
    \begin{equation}\label{eq:end_of_soundness_intro}
    \left[\E_{(\V y_1, \V y_2, \ldots, \V y_k)}\left[P(f(\V y_1), f(\V y_2), \ldots, f(\V y_k)) \right]\right]
    \approx
    \left[\E_{(\V y_1, \V y_2, \ldots, \V y_k)}\left[P(\tilde{f}(\V g_1), \tilde{f}(\V g_2), \ldots, \tilde{f}(\V g_k)) \right]\right],
    \end{equation}
    for some function $\tilde{f}$.
    This means that in expectation, the randomized
    strategy above satisfies at least $\Pr[\mbox{Test passes}] - o(1)$.
    On the other hand, $\val(\inst)\leq s+o(1)$,
    so overall we get that
    $\Pr[\mbox{Test passes}]\leq  s+o(1)$ as required.
\end{enumerate}
Note that in the analysis above,
even though we started with
an integrality gap with perfect
completeness, the resulting dictatorship test has completeness $1-\eps$.
Tracing this back, the only place
in the proof that the fact that
$\eps>0$ was used is in~\eqref{eq:key_ld_only_matters}.
In words, that equality asserts
that for a $k$-ary distribution
$\mu$, if we want to measure the
correlation of $f_1(y_1),\ldots,f_k(y_k)$
where $(y_1,\ldots,y_k)\sim \mu^{R}$ and $f_i\colon \Sigma^R\to\mathbb{R}$, then this correlation only comes from the
low-degree parts of $f_1,\ldots,f_k$. This fact holds if the support of the distribution $\mu$ is $\Sigma^k$,
and more generally when the
distribution $\mu$ is connected
in the sense of~\cite{Mossel}.
Alas, this fact is false for general distributions.

\subsubsection{A Dictatorship Test with Perfect Completeness}
To design a dictatorship test with completeness $1$ via the above paradigm we are forced to set $\eps=0$, meaning that
the distributions arising in~\eqref{eq:key_ld_only_matters} are arbitrary. This means that
the equality in~\eqref{eq:key_ld_only_matters}
is no longer true, and it is
unclear how to proceed with
the analysis of the dictatorship
test.

This is the point where the works~\cite{BKMcsp1,BKMcsp2,BKMcsp3,BKMcsp4} enter the picture.
The goal in these works is
to understand what sort of
functions may contribute
to the left hand side of~\eqref{eq:key_ld_only_matters}
under only very mild assumptions
on the input distribution.
This goal was partially achieved
in~\cite{BKMcsp4}, and we now
explain how we use that result to
proceed with the analysis of
the dictatorship test. Throughout
the rest of the discussion, we fix
$k=3$ and $\eps=0$ in
the above dictatorship test.

\paragraph{Decoding correlations and fixing~\eqref{eq:key_ld_only_matters}:}
the result of~\cite{BKMcsp2}
asserts that if $\mu_c$ has
no Abelian embeddings, then~\eqref{eq:key_ld_only_matters} continues to hold, and so
the analysis above proceeds in
the same way. The result of~\cite{BKMcsp4} is a
strengthening of it, asserting that if
$\mu_c$ admits Abelian embeddings but no $(\mathbb{Z},+)$ embeddings, then functions $F_1,F_2,F_3$ for which the left hand side of~~\eqref{eq:key_ld_only_matters} is non-negligible must arise from characters of Abelian groups and low-degree functions. More
precisely, that result shows that
there is a finite Abelian group
$G$, a character $\chi\in\widehat{G}^{R}$,
a map $\sigma\colon \Sigma\to G$
and a low-degree function
$L\colon \Sigma^R\to\mathbb{R}$
with $2$-norm $1$ such that
\[
\Expect{(\V y_1, \V y_2, \V y_3)\sim \mu_c^{R}}{F_1(\V y_1)\chi(\sigma(\V y_1))L(\V y_1)}\geq \Omega(1).
\]
In this paper, we deduce a list-decoding version of this statement, roughly making the
following assertion.
We can find functions $G_1,G_2,G_3\colon \Sigma^R\to\mathbb{R}$
that are each a sparse combination
of functions of the form
$\chi\circ \sigma\cdot L$
for a character $\chi$,
$\sigma\colon \Sigma\to G$
and low-degree function $L$
such that a correct version
of
~\eqref{eq:key_ld_only_matters}
becomes
\begin{equation}\label{eq:key_ld_char_only_matters}
    \E \left[\prod_{i\in S} F_i(\V y_i)\right] \approx  \E \left[\prod_{i\in S} G_i(\V y_i)\right]  ,
    \end{equation}
for any $S\subseteq \{1,2,3\}$.
\footnote{We remark that strictly speaking, the description of $G_1,G_2,G_3$ is incorrect, and
they are only close in $L_2$-distance to functions of this form. Our argument requires that $G_1,G_2,G_3$
are $O(1)$-bounded, and this results in
significant technical complications.}

\paragraph{Generalization of the invariance principle and fixing~\eqref{eq:cor_of_inv}:}
with~\eqref{eq:key_ld_char_only_matters} in hand, the analysis
of the dictatorship test would
be done provided an algorithm
could generate samples that
``fool'' the functions $G_i$
on the right hand. The invariance
principle of~\cite{MOO} asserts
that low-degree functions with
small influences are ``fooled''
by Gaussian samples, which we
are able to produce using the
semi-definite programming relaxation. However, Gaussian samples fail to fool characters.
This is where the Gaussian
elimination part of the algorithm
enters the picture: we argue
that characters are fooled by
the solutions to the linear
system of equations in the
hybrid algorithm. Expressing
\[
G_1(\V y_1) =
\sum\limits_{j=1}^{W}\chi_j(\sigma(\V y_1)) L_j(\V y_1)
\]
leads us to consider the function
$\tilde{G}_1\colon \mathbb{R}^{\card{\Sigma} R}\times G^{R}\to\mathbb{R}$
defined by
\[
\tilde{G}_1(\V \sigma_1,\V g_1)
=\sum\limits_{j=1}^{W}\chi_j(\V \sigma_1) L_j(\V g_1)
\]
A priori, it is not clear what
is the relation between $G_1$
and $\tilde{G}_1$. We have
split the input $\V y_1$ of
$G_1$ into two independent samples
from Gaussian space and from the
set of solutions to the linear
system. We show, however, that
if all $L_j$ have small influences, then the functions $G_1$ and $\tilde{G}_1$ are close in a sense that suffices for our purposes (this requires
some features from the set of characters $\chi_j$ appearing
in $G_1$ that we are able to ensure). This is what we refer
to by the ``mixed invariance principle'', because the discrete
probability distribution of $\V y_1$ is replaced by a mix of
a Gaussian distribution and a
distribution arising from a solution for a system of linear
equations. We defer a more formal
statement of the mixed invariance principle to Section~\ref{section:mixed_inv_intro} below.

Algorithmically, we are able to generate inputs to $\tilde{G}_1,\tilde{G}_2,\tilde{G}_3$ by solving the hybrid algorithm, allowing us to replace~\eqref{eq:cor_of_inv} with
\begin{equation}\label{eq:cor_of_mixed_inv}
    \E_{(\V y_1,\V y_2, \V y_3)\sim \mu_c^{R}} \left[\prod_{i\in S} G_i(\V y_i)\right] \approx
    \E_{(\V \sigma,\V g)} \left[\prod_{i\in S} \tilde{G}_i(\V \sigma_i,\V g_i)\right]
\end{equation}
for all $S\subseteq \{1,2,3\}$,
and the rest of the arguments proceeds in the same way. In
the end we get an analog of~\eqref{eq:end_of_soundness_intro} wherein the inputs to the function $\tilde{f}$ are generated
by the rounding algorithm presented in the beginning of the
section. Hence, we showed that the
probability the test accepts is at most $\val(\inst)+o(1)\leq s+o(1)$.

\subsubsection{On the {\symm} Assumption}\label{sec:on_symm}
We finish this section by discussing the {\symm} assumption, and in particular where it is used in the
above analysis. To carry out
the above argument, our only requirement is that the distributions $\mu_c$ arising in
the solution of the SDP part of
the hybrid algorithm are distributions for which the result
of~\cite{BKMcsp4} applies.
Therein, an inverse theorem
holds for all $3$-ary distributions that
do not admit $(\mathbb{Z},+)$
embeddings, and it stands to
reason that the result should hold
for the more general class of pairwise connected distributions
(in the sense defined therein).

Therefore, if one is interested
in a particular family of predicates $\mathcal{P}$ which
is not {\symm}, then in principle
one could still solve the SDP
program, and if the resulting
local distributions $\mu_c$
do not admit any $(\mathbb{Z},+)$
embeddings, then our argument
still goes through and the
hybrid algorithm above works.

The conditions stated in {\symm}
are a relatively elegant way of ensuring that
the local distributions $\mu_c$
have no $(\mathbb{Z},+)$ embeddings, and hence Theorem~\ref{thm:main} is stated in this way. More precisely, we argue that if we have an SDP solution with value $1$ for a collection $\mathcal{P}$ that is {\symm}, then we can modify it to get an SDP solution with value $1$ in which the local distributions $\mu_c$ have no $(\mathbb{Z},+)$ embeddings.

 \subsection{The Mixed Invariance Principle}
 \label{section:mixed_inv_intro}

In this section we discuss the mixed invariance principle.
Let $\mu$ be a distribution over
$\Sigma^3$ that has no $(\mathbb{Z},+)$ embeddings, and let
$f_1\colon (\Sigma^n,\mu_1^{\otimes n})\to\mathbb{C}$,
$f_2\colon (\Sigma^n,\mu_2^{\otimes n})\to\mathbb{C}$,
$f_3\colon (\Sigma^n,\mu_3^{\otimes n})\to\mathbb{C}$
be $1$-bounded functions.
The goal of the mixed invariance
principle is to study expectations
of the form
\begin{equation}\label{eq:intro_mixedinv}
    \Expect{(x,y,z)\sim \mu^{\otimes n}}{f_1(x)f_2(y)f_3(z)}
\end{equation}
and relate them to expectations
of related functions over different domains. The above expression should
be compared to the left hand
side in~\eqref{eq:key_ld_char_only_matters}. As explained therein, in
this paper we show that we can find functions
$\tilde{f}_1,\tilde{f}_2,\tilde{f}_3$ that are $1$-bounded, and
additionally are close in $L_2$-distance to sparse sums of functions that are product of characters over some finite Abelian group $G$ and low-degree functions, such that
\begin{equation}
    \card{\Expect{(x,y,z)\sim \mu^{\otimes n}}{f_1(x)f_2(y)f_3(z)}
    -
    \Expect{(x,y,z)\sim \mu^{\otimes n}}{\tilde{f}_1(x)\tilde{f}_2(y)
    \tilde{f}_3(z)}
    }=o(1).
\end{equation}
For the sake of simplicity assume that $\tilde{f}_1,\tilde{f}_2,\tilde{f}_3$ are themselves this sparse combinations (as opposed to just close to them; see Section~\ref{section:techniques} for a more formal discussion), and write
\begin{equation}\label{eq:decoupled_intro0}
\tilde{f}_1(x) = \sum\limits_{P\in\mathcal{P}_1}P(\sigma(x)) L_P(x)
\end{equation}
and similarly for $\tilde{f}_2$
and $\tilde{f}_3$.
Here, $\mathcal{P}_1$ is a set
of characters and $\sigma\colon \Sigma\to G$ is some map.\footnote{We remark that this decomposition is not unique, and often, for our applications, we work with a decomposition with certain extra properties.} Note that
in the case that $\card{\mathcal{P}_1} = 1$ and the
only element in $\mathcal{P}_1$
is the constant $1$ function,
the function $\tilde{f}_{1}$ is precisely
the low-degree part of $f_1$.
Towards the statement of the
invariance principle, we decouple
the inputs to the characters-part of $\tilde{f}_1$ and the low-degree parts of $\tilde{f}_1$
and define
\begin{equation}\label{eq:decoupled_intro}
	\tilde{f}^{1}_{{\sf decoupled}}(x,x') = \sum\limits_{P\in\mathcal{P}_1} P(\sigma(x)) \cdot L_P(x'),
\end{equation}
and we similarly define
$\tilde{f}^{2}_{{\sf decoupled}}$ and $\tilde{f}^{3}_{{\sf decoupled}}$.

\begin{definition}
		We say that a function $\tilde{f}_{{\sf decoupled}}$
		of the~\eqref{eq:decoupled_intro} has $\tau$-small
		shifted low-degree influences if
		for every $j\in [n]$ and every $P\in \mathcal{P}_1$ it holds that the $j^{th}$ influence of $L_P$ is at most $\tau$.
	\end{definition}

 We are now ready to state our mixed invariance principle.

 \begin{theorem}\label{thm:plain_mixed_inv} (Informal)
     Suppose $\mu$ is a distribution over $\Sigma^3$ that has no $\mathbb{Z}$-embedding. There exists an Abelian group $G$ and embeddings $\sigma_i: \Sigma\rightarrow G$ such that for $1$-bounded functions $f_1, f_2, f_3:\Sigma^n \rightarrow \mathbb{C}$, if the corresponding decoupled functions $\tilde{f}^{i}_{{\sf decoupled}}(x,x')$ have $\tau$-small shifted low-degree influences, then by letting $\tilde{F}_i$ be $\tilde{f}^{i}_{{\sf decoupled}}$ except we replace $L_P$ with their corresponding multi-linear expansions, we have
     \begin{equation}
			\card{
				\Expect{(x,y,z)\sim \mu^{\otimes n}}{f_1(x)f_2(y)f_3(z)}
				-
				\Expect{\substack{(x,y,z)\sim \mu^{\otimes n}, \\(g_x,g_y,g_z)\sim \mathcal{G}^{\otimes n}}}{\tilde{F}_1(x,g_x)\tilde{F}_2(y,g_y)\tilde{F}_3(z,g_z)}
			}\leq \xi(\tau),
		\end{equation}
  where $\xi(\tau)\rightarrow 0$ as $\tau\rightarrow 0$. Here $\mathcal{G}$ is a distribution over Gaussians with the same pairwise correlations as the distribution $\mu$.
 \end{theorem}

\begin{remark}\label{remark:reduction}
   In the above theorem, we use a non-standard notion of the functions that are `far from the dictator functions.' At this point, we do not know how to use this notion of dictatorship test in the actual (conditional) NP-hardness result (starting with the Rich $2$-to-$1$ Games Conjecture~\cite{BravermanKM21}). We believe that new techniques need to be designed for this goal, and we leave it as an open problem for future research. A related interesting problem, which may be easier, is deriving integrality gap instances from our dictatorship test. This problem is motivated by prior works, such as~\cite{Tulsiani,KV}, who were able to transform reductions (sometimes conditional on the Unique-Games Conjecture) to unconditional integrality gap for certain hierarchies of algorithms.
\end{remark}

\subsection{Further Remarks and Other Related Work}
\label{section:rel_work}
As discussed above, we view
the condition that the distribution $\mu$ has no
$(\mathbb{Z},+)$ embedding
as being of technical nature,
and expect the above argument
to have an analog as long as $\mu$
is pairwise connected. Here,
a $3$-ary distribution over
triples $(x,y,z)\in \Sigma^3$
is pairwise connected if the bipartite graph $G_{x,y} = ((\Sigma, \Sigma), E)$
where $E = \sett{(x,y)}{\exists z, (x,y,z)\in {\sf supp}(\mu)}$,
is connected, and similarly the
analogously defined
bipartite graphs $G_{x,z}$ and $G_{y,z}$
are also connected.

We expect that going beyond
pairwise connected distribution
would be a very significant challenge, in several ways. For predicates of arity $k=2$ it is not know how to handle SDP solutions that have disconnected local distributions. It is entirely possible that the best rounding algorithm needs to use richer rounding schemes.
For $k=3$, the class of distributions
that are not pairwise connected
is very rich, and a general inverse theorem for this entire class seems infeasible. It is likely that addressing it requires inverse theorems and accompanying analytical machinery for a special sub-class of distributions (such as pairwise connected distributions), as well as additional algorithmic machinery to handle disconnectedness. This class of distributions includes the density Hales-Jewett theorem for combinatorial lines of length $3$~\cite{DHJ1,DHJ2,DHJ3}, for which follow-up works have established ``reasonable bounds''~\cite{BKMcsp6,BKMcsp7,BKMcspDHJ}. For arity $k>3$,
the class of non pairwise connected distributions
captures difficult problems such as
multi-dimensional Szemer{\'e}di
theorems and their extensions~\cite{multidimensional}. While in some cases combinatorial proofs are known, in other cases, such as in the polynomial-version of Szemeredi's theorem, only ergodic theoretic proofs exist in the literature.

\paragraph{On Search vs Decision Algorithm.}
Our main result, Theorem~\ref{thm:main}, gives a polynomial time algorithm that distinguishes satisfiable instances of Max-$\mathcal{P}$-CSP from instances with value at most $\alpha_{\mathcal{P}}$. The approximation guarantee is matched with the soundness of the corresponding dictatorship test. At this point, we do not know how to convert this decision algorithm into a search algorithm. Namely, an algorithm that given a satisfiable instance as an input, {\em finds} an assignment with value at least $\alpha_{\mathcal{P}}$ in polynomial time. We believe that modifications of the techniques from Raghavendra-Steurer~\cite{RaghavendraS09} and Khot, Tulsiani, and Worah~\cite{KTW14} could provide the search algorithm, and we leave this as an open problem for the future.

\section{Techniques}
 \label{section:techniques}
	Is this section we outline the proof of Theorem~\ref{thm:plain_mixed_inv}. The formal proof
    appears in Sections~\ref{sec:reg_lemmas} and~\ref{sec:mixed_space}.

    \subsection{List Decoding}
    Consider an expectation
    of the form~\eqref{eq:intro_mixedinv} and suppose it is non-negligible in absolute value. We already
    know, using results of Mossel~\cite{Mossel}, that if the distribution $\mu$
    is connected, then each one of the functions $f_{i}$
    must be correlated with
    a low-degree function. But
    what happens in the case
    that $\mu$ is not connected? This question was asked in~\cite{BKMcsp1,BKMcsp2,BKMcsp4}, wherein under
    the condition that $\mu$ does not admit $(\mathbb{Z},+)$
    embedding, some conclusion regarding the structure of the functions $f_i$ was made. More precisely, the main result of~\cite{BKMcsp4}
    asserts that there is a
    constant size Abelian group $G_{{\sf master}}$ and
    embeddings $\sigma,\gamma,\phi\colon \Sigma\to G_{{\sf master}}$
    of $\mu$ into $G_{{\sf master}}$,
    such that if $f_1,f_2,f_3$
    are $O(1)$-bounded functions satisfying
    \[
    \card{\Expect{(x,y,z)\sim \mu^{\otimes n}}{f_1(x)f_2(y)f_3(z)}}\geq \eps,
    \]
    then function $f_1$ (and similarly $f_2,f_3$)
    must be correlated with
    a function of the form
    $\chi\circ \sigma \cdot L$
    where $\chi\in \widehat{G}_{{\sf master}}^n$
    and $L$ is a function of degree $O_{\eps}(1)$ and
    bounded $2$-norm. In other
    words, $f_1$ is correlated
    with a single function
    that appears on the right
    hand side of~\eqref{eq:decoupled_intro0}. Optimistically, one
    may hope that $\chi\circ \sigma \cdot L$
    ``explains'' all of the
    magnitude of $\Expect{(x,y,z)\sim \mu^{\otimes n}}{f_1(x)f_2(y)f_3(z)}$ in
    the sense that $f_1$ may be
    replaced by it without changing the value of the
    expectation by much. Unfortunately, this is incorrect; for once, $f_1$
    could be correlated with
    numerous functions of this
    type that are very far from each other. Besides, even if this was true, there is an additional issue: we do not wish to replace the function $f_1$ by a function that is not bounded (at least not at this point of the argument).

    A common strategy for addressing the first issue
    is by appealing to the notion of list decoding. As a first attempt, one may hope to collect all of the functions
    of the form $\chi\circ \sigma \cdot L$ that are
    correlated with $f_1$, and
    then replace $f_1$ with a
    weighted sum of them according to their correlation. This would have worked had that collection of functions been pairwise orthogonal, which is not true in our case. As a second attempt, one may find a single function $\chi\circ \sigma \cdot L$ that is correlated with $f_1$, and
    then iterate the argument
    with $f_1' = f_1 - \alpha \cdot \chi\circ \sigma \cdot L$
    where $\alpha = \inner{f_1}{\chi\circ \sigma \cdot L}$. This attempt also fails, this time due to the fact that $f_1'$ may not be $O(1)$-bounded,
    and hence the result of~\cite{BKMcsp4} no
    longer applies.

    To resolve this issue,
    we take inspiration from
    the invariance principle of~\cite{MOO,Mossel}
    and from the inverse Gowers' norms literature~\cite{green2007montreal}. In the former, instead of harsh truncations, one performs ``soft truncations'' by applying a noise operator
    that gives a bounded function that is close in $L_2$ distance to the low-degree part of the function. In the latter,
    one defines an averaging operator with respect to
    a sigma-algebra induced
    by the collection of functions that were already found to be correlated with
    $f_1$. We combine these two solutions, and the bulk of
    our effort is devoted to showing it works.

    \subsection{The Noise Operator, the Regularity Lemma, and the Approximating Formula}
    Let $\nu$ be the marginal distribution of $\mu$ on
    the first coordinate,
    and
    consider the function $f_1\colon(\Sigma^n,\nu^{\otimes n})\to \mathbb{C}$ above.
    Let $\mathcal{P} = \{P_1,\ldots,P_r\colon \Sigma^n\to H\}$ be a collection of functions into some discrete domain $H$, and let $\eps>0$.
    We define the function
    $\mathrm{T}_{\mathcal{P}, 1-\eps} f_1\colon \Sigma^n\to \mathbb{C}$ in the following way. First, for each $x\in \Sigma^n$, we define the distribution $\mathrm{T}_{\mathcal{P}, 1-\eps} x$, wherein a sample $y\sim \mathrm{T}_{\mathcal{P}, 1-\eps} x$
    is generated as follows: sample $I\subseteq [n]$ by including each element in it independently with probability $\eps$, then sample $y\sim \nu$ conditioned on $y_{\bar{I}} = x_{\bar{I}}$ and $P_i(y) = P_i(x)$ for all $i=1,\ldots,r$ (see Definition~\ref{def:noise_op_P}). The function $\mathrm{T}_{\mathcal{P}, 1-\eps} f_1$ is then defined by
    \[
    \mathrm{T}_{\mathcal{P}, 1-\eps} f_1(x)
    =\Expect{y\sim \mathrm{T}_{\mathcal{P}, 1-\eps}x}{f_1(y)}.
    \]
    To get some intuition to the
    operator $\mathrm{T}_{\mathcal{P}, 1-\eps}$ we recommend thinking of $\eps$ as very small. We begin by noting that the distribution $\nu$
    is a stationary distribution of $\mathrm{T}_{\mathcal{P}, 1-\eps}$, and that functions
    $f$ for which $f(x) = g(P_1(x),\ldots,P_r(x))$
    for some $g\colon H^r\to\mathbb{R}$ are eigenfunctions with eigenvalue $1$. We also
    note (though this is a bit
    more tricky to prove) that
    low-degree functions are
    nearly eigenfunctions of
    $\mathrm{T}_{\mathcal{P}, 1-\eps}$ of eigenvalue $1$,
    and more precisely that
    $\norm{(I-\mathrm{T}_{\mathcal{P}, 1-\eps})L}_2 = o(\norm{L}_2)$
    provided that $\eps$ is small compared to the degree of $L$. In other words, when constructed
    for an appropriate collection $\mathcal{P}$,
    the operator $\mathrm{T}_{\mathcal{P}, 1-\eps}$ almost does not affect
    functions of the form of
    the right-hand side of~\eqref{eq:decoupled_intro0}, and therefore it may be useful to detect such structures. We show that this is indeed the case.

    \subsubsection{The Regularity Lemma}
    With the new noise operator
    $\mathrm{T}_{\mathcal{P},1-\eps}$ in hand, we may attempt to execute the idea of the list decoding argument in a different way. More specifically, starting with $f=f_1$, $\mathcal{P} = \emptyset$ and $\tilde{f}_1 = 0$, once we find a function of the form $\chi\circ \sigma \cdot L$
    that $f$ is correlated with, we insert $P = \chi\circ \sigma$ into $\mathcal{P}$, take $\tilde{f}_1 = \mathrm{T}_{\mathcal{P},1-\eps} f$ and proceed the argument on $f = f_1 - \tilde{f}_1$.
    In the next step, we will
    find a new $\chi'\circ \sigma \cdot L'$ correlated with the updated $f$, insert $P' = \chi'\circ \sigma$ to $\mathcal{P}$, update $\tilde{f}_1 = \mathrm{T}_{\mathcal{P}, 1-\eps} f$ and proceed the argument on $f = f_1 - \tilde{f}_1$. When the argument
    terminates, we will have that
    \[
    \card{\Expect{(x,y,z)\sim \mu^{\otimes n}}{(f_1-\tilde{f}_1)(x)f_2(y)f_3(z)}} = o(1),
    \]
    meaning we can replace $f_1$
    with $\tilde{f}_1$.
    Our proof goes along these
    lines, but some care
    is required. The key issue is that if we inspect the previous
  inductive process closely,
    instead of updating $\tilde{f}_1$ at each step,
    we should have
    subtracted a new noised function. For instance, in
    the second step we would have subtracted $\mathrm{T}_{\{P'\}, 1-\eps}(f- \mathrm{T}_{\{P\}, 1-\eps} f)$ from $f- \mathrm{T}_{\{P\}, 1-\eps} f$ to get the function
    \[
    \tilde{f}_1 = f
    - \mathrm{T}_{\{P\}, 1-\eps} f
    -
    \mathrm{T}_{\{P'\}, 1-\eps}(f- \mathrm{T}_{\{P\}, 1-\eps} f),
    \]
    and at each iteration, the function we inspect gets even more complicated.
    After a large
    number of steps, the function we end up with is no longer $O(1)$-bounded, in which case the iterative process gets stuck. Fortunately, we show
    that instead of doing this,
    one may update the approximating function $\tilde{f}_1$ as explained
    above. We remark that in the formal argument we are
    required to modify the
    noise rate $\eps$ at each
    step, which contributes to
    several technical challenges. Once this process is done appropriately, we find a
    function $\tilde{f}_1 = \mathrm{T}_{\mathcal{P},1-\eps} f$ where $\card{\mathcal{P}}$ is constant and $\eps$ is bounded away from $0$, such that
    \[
    \card{
    \Expect{(x,y,z)\sim \mu^{\otimes n}}{f_1(x)f_2(y)f_3(z)}
    -
    \Expect{(x,y,z)\sim \mu^{\otimes n}}{\tilde{f}_1(x)f_2(y)f_3(z)}} = o(1).
    \]
    Similarly, we show that one may replace $f_2$ and $f_3$
    with noisy versions of them.

    \subsubsection{The Approximating Formula}
    The next step in
    the proof is to show that
    a function of the form
    $\mathrm{T}_{\mathcal{P},1-\eps} f_1$ may be approximated
     by a function of the form of the right-hand side of~\eqref{eq:decoupled_intro0} in $L_2$ distance.\footnote{We remark
    that in fact, such a formula is necessary
    for us even to make
    the previous regularity
    lemma go through.} We
    establish such a formula by
    fairly direct calculations, starting
    with
    \[
    \mathrm{T}_{\mathcal{P},1-\eps}f(x)
    =
    \Expect{I\subseteq_{\eps} [n]}
    {\frac{1}{A_I(x)}\cExpect{y\sim \nu}
    {y_I = x_I}{f(y)1_{P_i(y) = P_{i}(x)~\forall i}}},
    \]
    where
    $A_I(x) = \cProb{y\sim \nu}{y_{\bar{I}} = x_{\bar{I}}}{P_i(y) = P_i(x)~\forall i}$.
    Morally, we show that on most inputs, both $1/A_I(x)$
    and the expectation  above are close to
    functions of the form
    of the right-hand side
    of~\eqref{eq:decoupled_intro0}, except that the function
    $L_P$ is not a low-degree function but
    instead only depends on the coordinates from $\bar{I}$. The expectation over $I$ then turns these $L_P$ into low-degree functions (see Lemma~\ref{lem:avg_of_juntas_low_deg}).

    \subsection{Decoupling the Abelian Part and the Low Degree Part}
    With the previous
    steps done, we managed to replace
    each one of the functions
    $f_1,f_2,f_3$ in the
    expectation~\eqref{eq:intro_mixedinv} with the functions $\mathrm{T}_{\mathcal{P}, 1-\eps} f_1$,
    $\mathrm{T}_{\mathcal{Q}, 1-\eps} f_2$
    and
    $\mathrm{T}_{\mathcal{R}, 1-\eps} f_3$ respectively. Additionally,
    we have approximate
    formula for each one
    of these functions
    that seem to have the desired
    form. Namely,
    ignoring the
    aspect of this approximation for a
    moment, we have that
    \[
    \mathrm{T}_{\mathcal{P},1-\eps} f_1(x)=
    \sum\limits_{P\in \spn(\mathcal{P})} P(x) L_P(x)
    \]
    where $\spn(\mathcal{P})$ consists
    of all possible products
    of functions from $\mathcal{P}$, and
    each one of the functions $L_P$ is
    a low-degree function
    with bounded $2$-norm.
    We would now like to
    decouple the input $x$
    into two copies, one of
    which will be fed to
    the functions $P$ and
    the other one will be
    fed to the low-degree
    functions $L_P$ (indeed,
    this is the only difference between the formulas in~\eqref{eq:decoupled_intro0} and in~\eqref{eq:decoupled_intro}). It turns out
    that provided that the
    collection $\mathcal{P}$
    has a high rank, there is
    a coupling $X, (X',X'')$
    between $\nu^{\otimes n}$ and $\nu^{\otimes n}\times \nu^{\otimes n}$ such that
    \[
    \sum\limits_{P\in \spn(\mathcal{P})} P(X) L_P(X)
    =
    \sum\limits_{P\in \spn(\mathcal{P})} P(X') L_P(X'')
    +o(1)
    \]
    in $L_2$ distance.
    By high-rank,
    we mean that besides
    the all $1$
    function,
    every
    function in $\spn(\mathcal{P})$
    has a high Fourier analytic degree; in
    fact, this degree should be much higher than the degree of any of the
    functions $L_P$. We do not elaborate on this point further but remark that we show how to
    achieve this property via an appropriate clean-up process.

    The intuition for the coupling
    comes from the fact that if we look at the left-hand side above
    and expose $1-1/D$ of the
    coordinates of $X$, where $D$ is much higher than the degrees of the functions $L_P$ but much smaller than the rank of $\mathcal{P}$, then the values of
    the functions $L_P$
    are almost fully determined, whereas the
    values of the functions $P(X)$ still have the
    same distribution as
    the original one. This
    suggests that the behavior
    of the values of the functions $P(X)$ and the functions $L_P(X)$
    is almost independent of each other.

    \subsection{Deducing the Invariance Principle}
    In conclusion,
    we have shown that
    the values of the noised function $\mathrm{T}_{\mathcal{P},1-\eps} f_1$
    can be coupled with
    the values of
    decoupled function from~\eqref{eq:decoupled_intro} in a way
    that is close in $L_2$
    distance. Intuitively, that ought to imply that
    \[
    \card{
    \Expect{(x,y,z)\sim \mu^{\otimes n}}{f_1(x)f_2(y)f_3(z)}
    -
    \hspace{-3ex}\Expect{\substack{(x,y,z)\sim \mu^{\otimes n}\\ (x',y',z')\sim \mu^{\otimes n}}}{\tilde{f}_{1,{\sf decoupled}}(x,x')\tilde{f}_{2,{\sf decoupled}}(y,y')\tilde{f}_{3,{\sf decoupled}}(z,z')}}
    \]
    is $o(1)$.
    While this is morally the case, there are some issues one has to address. The first issue is of purely technical nature: the absolute value of the decoupled functions is not necessarily upper bounded by an absolute constant, so one should replace these values by some sort of truncation. The second issue is that above, we actually need a multi-variate version of our mixed invariance principle. While we show how to deduce such a result (in a relatively straightforward manner) from the univariate version of the mixed invariance, it requires the collections of product functions we have produced for each one of $f_1$, $f_2$ and $f_3$ to be aligned in a certain manner, and we defer the reader to Section~\ref{sec:mixed_invariance} for more details.

    With the above result, one
    may use the standard
    invariance principle
    to replace $x',y',z'$
    with Gaussian random
    variables. As for
    the $x,y,z$, one notes
    that the values of the
    functions only depend on
    $\sigma(x), \gamma(y)$
    and $\phi(z)$, and so
    they can be replaced with elements from $G_{{\sf master}}$ that add
    up to $0$. This finishes a proof sketch of Theorem~\ref{thm:plain_mixed_inv}.

	\section{Preliminaries}
	In this section, we collect a few basic notions and results that we need.

	\paragraph{Notations.}
   For a vector $x\in \Sigma^n$ and a subset $I\subseteq [n]$ of coordinates, we denote by
$x_I$ the vector in $\Sigma^{I}$ which results by dropping from $x$ all coordinates outside $I$. We denote by
$x_{\bar{I}}$ the vector in $\Sigma^{n-\card{I}}$ resulting from dropping from $x$ all coordinates from $I$.
For $i\in [n]$ we denote by $x_{-i}$ the vector in $\Sigma^{n-1}$ resulting from dropping the $i$th coordinate of $x$. For $I\subseteq[n]$, $a\in \Sigma^{I}$
and $b\in \Sigma^{n-\card{I}}$ we denote by $(x_I = a, x_{\bar{I}} = b)$ the point in $\Sigma^{n}$ whose $I$-coordinates
are filled according to $a$, and whose $\overline{I}$-coordinates are filled according to $b$. For two strings $x,y\in \Sigma^n$
we denote by $\Delta(x,y)$ the Hamming distance between $x$ and $y$, that is, the number of coordinates $i\in [n]$ such that $x_i\neq y_i$. We denote by ${\bm i}$ the complex root of $-1$, and by $\overline{a}$ the complex conjugate of the number $a\in\mathbb{C}$.
We denote by $I\subseteq_{\rho}[n]$ a
random subset of $[n]$ in which we include each $i\in[n]$ independently with probability $\rho$.

    We denote $A\lll B$ to refer to the fact that $A\leq C\cdot B$ for some absolute constant $C>0$,
and $A\ggg B$ to refer to the fact that $A\geq c \cdot B$ for some absolute constant $c>0$.
If this constant depends on some parameter, say $m$, the corresponding notation is $A\lll_m B$. We will also
use standard big-$O$ notations: we denote $A = O(B)$ if $A\lll B$, $A = \Omega(B)$ if $A\ggg B$; if
there is a dependency of the hidden constant on some auxiliary parameter, say $m$, we denote $A = O_m(B)$
and $A = \Omega_m(B)$.

We denote $0<A\ll B$ to refer
to the choice of $A$ and $B$ in
the way that $B$ is fixed, and
then $A$ is taken to be sufficiently small compared to $B$.

	\subsection{Discrete Fourier Analysis over Product Spaces}
	Let $(\Sigma, \nu)$ be a finite probability space, and consider the product space
	$(\Sigma^n, \nu^n)$. We will often consider the space $L_2(\Sigma^n,\nu^{n})$ of
	complex-valued functions $f\colon \Sigma^n\to\mathbb{C}$ equipped with the standard
	inner product
	\[
	\inner{f}{g}_{\nu} = \Expect{x\sim \nu^n}{f(x)\overline{g(x)}}.
	\]
	Often times, the underlying measure $\nu$ will be omitted from the notation, as it will be clear from context.
	The space $L_2(\Sigma^n,\nu^{n})$ admits several types of decompositions that we will use throughout this
	paper. The coarsest decomposition is the degree decomposition; the Efron-Stein decomposition is a refinement
	of the degree decomposition; finally, the Fourier decomposition is a refinement of the Efron-Stein decomposition.
	We next present each one of these decompositions.
	\subsubsection{The Degree Decomposition}
	We first define the notion of $d$-juntas and the space $V_{\leq d}(\Sigma^n)$.
	\begin{definition}
		For a subset $I\subseteq [n]$, a function $f\colon \Sigma^n\to\mathbb{C}$ is called a $I$-junta if there exists a function $f'\colon \Sigma^{I}\to\mathbb{C}$ such that $f(x) = f'(x_{I})$ for all $x\in\Sigma^n$. A function $f$ is called a $d$-junta
		for $d\in\mathbb{N}$ if it is an $I$-junta for $I$ of size at most $d$.
		We define the space $V_{\leq d}(\Sigma^n)$ to be the linear span of all $d$-juntas.
	\end{definition}
	The space $V_{\leq d}(\Sigma^n)$ is often referred to as the space of degree $d$ functions. Using our inner product,
	we may define the space of pure degree $d$ functions as follows.
	\begin{definition}
		Given a product space $(\Sigma^n,\nu^n)$, the space of pure degree $d$ functions is defined as
		\[
		V_{=d}(\Sigma^n,\nu^{n}) = V_{\leq d}(\Sigma^n)\cap V_{\leq d-1}(\Sigma^n)^{\perp}.
		\]
	\end{definition}
	It is easily seen that $L_2(\Sigma^n,\nu^{n})
	=V_{\leq n}(\Sigma^n)
	=\bigoplus_{d=0}^{n} V_{= d}(\Sigma^n,\nu^{n})$,
	and thus any function $f\colon \Sigma^n\to\mathbb{C}$ may be uniquely written as
	\[
	f(x) = \sum\limits_{d=0}^{n} f^{=d}(x),
	\qquad\qquad\text{where }f^{=d}\in V_{= d}(\Sigma^n,\nu^{n})~\forall d.
	\]
	This decomposition is called the degree decomposition of $f$, and the function $f^{=d}$ is
	referred to as the degree $d$ part of $f$.
	
	\subsubsection{The Efron-Stein Decomposition}
	The Efron-Stein decomposition is a refinement of the degree decomposition. For each $d\in\mathbb{N}$ and $S\subseteq [n]$
	of size $d$, one defines the space $V_{=S}(\Sigma^n,\nu^{n}) = V_{=d}(\Sigma^n,\nu^{n})\cap \set{\text{$S$-juntas}}$. It is a standard fact that $\bigoplus_{\card{S}=d} V_{=S} = V_{=d}$, and thus one may write $f^{=d}\in V_{=d}$ as
	\[
	f^{=d}(x) = \sum\limits_{\card{S} = d}{f^{=S}(x)}
	\qquad\qquad\text{where }f^{=d}\in V_{= S}(\Sigma^n,\nu^{n})~\forall S.
	\]
	Thus, one gets the decomposition of $f$ as
	\[
	f^{=d}(x) = \sum\limits_{S\subseteq[n]}{f^{=S}(x)},
	\]
	where $f^{=S}\in V_{=S}(\Sigma^n,\nu^{n})$ for all $S$. This decomposition
	is known as the Efron-Stein Decomposition.
        \begin{definition}
            The level $d$-weight of a function $f$ is defined as $W^{=d}[f] = \norm{f^{=d}}_2^2$, and by Parseval's equality we have that $W^{=d}[f] = \sum\limits_{\card{S}=d}\norm{f^{=S}}_2^2$.
            We also denote $W^{\geq d}[f] = \sum\limits_{i=d}^{n}W^{=i}[f]$.
        \end{definition}
       We will need the following lemma:
       \begin{lemma}\label{lem:avg_of_juntas_low_deg}
           Suppose that we have functions $f_I\colon \Sigma^n\to\mathbb{C}$, where each $f_I$ is
           an $\bar{I}$-junta, and define $f(x) = \Expect{I\subseteq_{\eps}[n]}{f_I(x)}$.
           \begin{enumerate}
               \item If each $f_I$ is $1$-bounded, then $f$ is $1$-bounded.
               \item If $\norm{f_I}_2\leq 1$ for each $I$, then for each $d\geq 1$ we have that $W^{\geq d}[f]\leq (1-\eps)^d$.
           \end{enumerate}
       \end{lemma}
       \begin{proof}
           The first item is clear by the triangle inequality.
           For the second item, note that for each $S\subseteq [n]$ we have that $f^{=S}(x) = \Expect{I\subseteq_{\eps}[n]}{f_I^{=S}(x)}$.
           As $f_I$ is an $\bar{I}$-junta, we have that
           $f_I^{=S}\equiv 0$ unless $S\subseteq \bar{I}$,
           and so $f^{=S}(x) = \Expect{I\subseteq_{\eps}[n]}{f_I^{=S}(x)1_{S\subseteq\bar{I}}}$. It follows from the Cauchy-Schwarz inequality that
           \begin{align*}
           \norm{f^{=S}}_2^2
           =\card{\Expect{x}{\Expect{I\subseteq_{\eps}[n]}{f_I^{=S}(x)1_{S\subseteq\bar{I}}}}}^2
           &\leq
           \Expect{x}{\Expect{I\subseteq_{\eps}[n]}
           {\card{f_I^{=S}(x)}^2}\Expect{I\subseteq_{\eps}[n]}{1_{S\subseteq\bar{I}}}}\\
           &=(1-\eps)^{\card{S}}\Expect{I\subseteq_{\eps}[n]}{\norm{f_I^{=S}}_2^2}.
           \end{align*}
           Thus,
           \[
           W^{\geq d}[f]
           =\sum\limits_{\card{S}\geq d}\norm{f^{=S}}_2^2
           \leq \sum\limits_{\card{S}\geq d}(1-\eps)^{\card{S}}\Expect{I\subseteq_{\eps}[n]}{\norm{f_I^{=S}}_2^2}
           \leq (1-\eps)^d
           \Expect{I\subseteq_{\eps}[n]}{\sum\limits_{\card{S}\geq d}\norm{f_I^{=S}}_2^2},
           \]
           which by Parseval is at most $(1-\eps)^{d}\Expect{I\subseteq_{\eps}[n]}{\norm{f_I}_2^2}\leq (1-\eps)^d$.
       \end{proof}
	\subsubsection{The Fourier Decomposition}
	The Fourier decomposition is the most refined and explicit decomposition among the decompositions discussed herein.
	One first looks at the base space, $L_2(\Sigma,\nu)$, and picks an orthonormal basis for it consisting of
	$m=\card{\Sigma}$ real-valued functions, $v_0,\ldots,v_{m-1}\colon \Sigma\to\mathbb{R}$. It is standard to take
	$v_0$ to be the all $1$ function, and we will do so here. With these notations, we may construct an
	orthonormal basis for $L_2(\Sigma^n,\nu^n)$ consisting of the functions $\{v_{\vec{\alpha}}\}_{\vec{\alpha}\in \{0,\ldots,m-1\}^{n}}$
	where
	\[
	v_{\vec{\alpha}}(x) = \prod\limits_{i=1}^{n} v_{\alpha_i}(x_i).
	\]
	Thus, any function $f\colon\Sigma^n\to\mathbb{C}$ admits a unique decomposition as $f(x) = \sum\limits_{\vec{\alpha}\in [m]^n}\widehat{f}(\vec{\alpha})v_{\vec{\alpha}}(x)$ where $\widehat{f}(\vec{\alpha}) = \inner{f}{v_{\vec{\alpha}}}$.
	We remark that it is easily seen that the space $V_{=d}$ is the span of all $v_{\vec{\alpha}}$ with $\card{{\sf supp}(\alpha)} = d$,
	and furthermore that $V_{=S}$ is the span of all $v_{\vec{\alpha}}$ with ${\sf supp}(\alpha) = S$.
	
	The Fourier definition presented herein can be somewhat arbitrary, as the functions $v_1,\ldots,v_{m-1}$ are not unique.
	Nevertheless, it will be useful in the presentation of the invariance principle~\cite{MOO} in Section~\ref{sec:invariance_principle}.
	
	\subsubsection{Hypercontractivity}
	\begin{definition}
		For $q\geq 2$ and $f\colon (\Sigma^n,\nu^{n})\to\mathbb{C}$, we define the $q$-norm of $f$ as $\norm{f}_q = \left(\Expect{x\sim\nu^n}{\card{f(x)}^q}\right)^{1/q}$.
	\end{definition}

    \begin{theorem}[{\cite[Chapter 10]{ODonnell}}]\label{thm:hypercontractivity}
		Suppose that $(\Sigma,\nu)$ is a probability space with $\card{\Sigma} = m$ such that $\nu(x)\geq \alpha$ for all
		$x\in\Sigma$, where $m\in\mathbb{N}$ and $\alpha>0$. For all $q\geq 2$ there exists $C = C(q,\alpha,m)>0$ such
		that if $f\colon (\Sigma^n,\nu^{n})\to\mathbb{C}$ is a degree $d$ function, then
		$\norm{f}_q\leq C^d\norm{f}_2$.
	\end{theorem}
	\subsubsection{The Noise Operator}
	We will need the standard noise operator on product spaces.
	\begin{definition}
		For a parameter $\rho\in [0,1]$ and an input $x\in\Sigma^n$, we define the distribution $\mathrm{T}_{\rho} x$
		over $\rho$-correlated inputs with $x$ in the following way. For each $i\in [n]$ independently, with probability
		$\rho$ pick $y_i = x_i$, and otherwise sample $y_i$ according to $\nu$.
	\end{definition}
	The process $\mathrm{T}_{\rho}$ is described as a Markov chain, and as is standard we may consider it as an averaging operator
	over functions. That is, we may consider it as a linear operator $\mathrm{T}_{\rho}\colon L_2(\Sigma^n,\nu^n)\to L_2(\Sigma^n,\nu^n)$
	defined as $\mathrm{T}_{\rho} f(x) = \Expect{y\sim \mathrm{T}_{\rho} x}{f(y)}$.
        We will also define an operator $\mathrm{T}_{I}$
        for each $I\subseteq[n]$, as follows:
        \begin{definition}
		 For a subset $I\subseteq [n]$, and for $x\in \Sigma^n$, we define the distribution of $\mathrm{T}_{I}x$ as: to sample $y\sim \mathrm{T}_I x$, set $y_i = x_i$ for $i\not\in I$, and $y_i\sim\nu$
         independently for each $i\in I$. Abusing notation, we will also think of $\mathrm{T}_I$ as an operator $\mathrm{T}_{I}\colon L_2(\Sigma^n,\nu^n)\to L_2(\Sigma^n,\nu^n)$ defined by
    $\mathrm{T}_{I}f(x) = \Expect{y\sim \mathrm{T}_{I}x}{f(y)}$.
	\end{definition}

	\subsubsection{Influences and Low-degree Influences}
	Next, we define the notions of influences and low-degree influences.
	\begin{definition}\label{def:influence}
		For a function $f\colon (\Sigma^n,\nu^n)\to\mathbb{C}$ and a coordinate $i\in [n]$, the influence of
		$i$ is defined as
		\[
		I_i[f] = \Expect{\substack{y\sim \nu^{n-1}\\a,b\sim \nu}}{\card{f(x_{-i} = y, x_i = a)-f(x_{-i} = y, x_i = b)}^2}.
		\]
	\end{definition}
	\begin{definition}
		For a function $f\colon (\Sigma^n,\nu^n)\to\mathbb{C}$, a parameter $d\in\mathbb{N}$ and a coordinate $i\in [n]$,
		the degree $d$ influence of $i$ is defined as $I_i^{\leq d}[f] = I_i[f^{\leq d}]$.
	\end{definition}
	\subsection{The Invariance Principle}\label{sec:invariance_principle}
	In this section, we present the invariance principle of~\cite{MOO}, and we begin with some set-up.
	Suppose that $\Sigma,\Gamma,\Phi$ are finite alphabets of sizes $m_1,m_2,m_3$ respectively, and
	$\mu$ is a probability measure over $\Sigma\times \Gamma\times \Phi$ in which the probability of
	each atom is at least $\alpha>0$. We set up Fourier bases for $(\Sigma,\mu_x)$,
	$(\Gamma,\mu_y)$ and $(\Phi,\mu_z)$ given by $v_0,\ldots,v_{m_1-1}$,
	$u_0,\ldots,u_{m_2-1}$ and $w_0,\ldots,w_{m_3-1}$. Consider the ensemble of random
	variables
	\[
	\mathcal{X} = \{v_1(x),\ldots,v_{m_1-1}(x),u_1(y),\ldots,u_{m_2-1}(y),w_1(z),\ldots,w_{m_3-1}(z)\}
	\]
	where $(x,y,z)\sim \mu$. We define the covariance matrix
	$P\in\mathbb{R}^{(m_1+m_2+m_3)\times(m_1+m_2+m_3)}$ whose rows and columns correspond
	to the functions in $\mathcal{X}$.
	For example, for $v_i$ and $u_j$,
	the corresponding entry is
	\[
	P(v_i,u_j) = \Expect{(x,y,z)\sim \mu}{v_i(x)\overline{u_j(y)}}.
	\]
	Let
	\[
	\mathcal{G} = \{G_{1,x},\ldots,G_{m_1-1, x},\ldots, G_{1,z},\ldots,G_{m_3-1,z}\}
	\]
	be an ensemble of centered Gaussian random variables with covariance matrix $P$. The invariance principle relates
	the behavior of low-influence, multi-linear polynomials over $\mathcal{X}$ and over $\mathcal{G}$. Below, we state
	the version that we need from~\cite{Mossel} specialized to our case of interest, but before that, we need a few definitions.
	
	Denote $q = m_1+m_2+m_3-3$, and let $M\colon \mathbb{C}^{q n}\to\mathbb{C}$ be a multi-linear polynomial given as
	\[
	M(a_{1,1},\ldots,a_{1,q},\ldots, a_{n,1},\ldots, a_{n,q})
	=\sum\limits_{T\subseteq [n]\times [q]}m_T \prod\limits_{(i,j)\in T} a_{i,j},
	\]
    where the coefficients $m_T$ are complex numbers and
    $a_{i,j}$ are thought of as complex-valued inputs.
	\begin{definition}
		The influence of variable $(i,j)$ on $M$ as above is defined as $I_{i,j}[M] = \sum\limits_{T\ni (i,j)}\card{a_{i,j}}^2$.
	\end{definition}
	We will also consider vector-valued multi-linear functions, which are functions
	$M\colon \mathbb{C}^{qn}\to\mathbb{C}^{k}$ wherein each $M_s$ is a multi-linear function.
	The influence of $(i,j)$ on $M$ is defined as $I_{i,j}[M] = \max_{s}I_{i,j}[M_s]$.

	Finally, define ${\sf trunc}\colon\mathbb{C}\to\mathbb{C}$ by ${\sf trunc}(a) = a$ if $\card{a}\leq 1$
	and ${\sf trunc}(a) = a/|a|$ otherwise.
	\begin{theorem}\label{thm:invariance_principle}
		For all $\alpha>0$, $k,m\in\mathbb{N}$, $d\in\mathbb{N}$, $C>0$ and $\eps>0$, there exists $\tau>0$
		such that the following holds. Suppose that $\card{\Sigma},\card{\Gamma},\card{\Phi}\leq m$,
		that $\mu$ is a distribution over $\Sigma\times\Gamma\times\Phi$ in which the probability
		of each atom is at least $\alpha$, and let $\mathcal{X}$ and $\mathcal{G}$ be the ensembles
		of random variables as above with the same covariance matrix. Let $M\colon\mathbb{C}^{qn}\to\mathbb{C}^k$
		be a multi-linear polynomial with $\max_{i,j}I_{i,j}[M]\leq \tau$ and $\norm{M}_2\leq C$.
		\begin{enumerate}
			\item If $\Psi\colon \mathbb{C}^{k}\to\mathbb{C}$ is differentiable three times and its third order derivatives
			are at most $C$ in absolute value, then
			\[
			\card{\Expect{}{\Psi(M(\mathcal{X}^n))}-\Expect{}{\Psi(M(\mathcal{G}^n))}}\leq \eps.
			\]
			\item Define $\truncerr\colon \mathbb{C}^k\to \mathbb{R}$ by $\truncerr(a_1,\ldots,a_k) = \sqrt{\sum\limits_{i=1}^{k}\card{{\sf trunc}(a_i) - a_i}^2}$.
			Then
			\[
			\card{\Expect{}{\truncerr(M(\mathcal{X}^n))}-\Expect{}{\truncerr(M(\mathcal{G}^n))}}\leq \eps.
			\]
		\end{enumerate}
	\end{theorem}
	
	Last, we need the following elementary fact.
	\begin{fact}\label{fact:trivial_lipshitz_pf}
		Consider the function $\truncerr\colon \mathbb{C}\to [0,\infty)$ defined as
		$\truncerr(a) = \card{{\sf trunc}(a) - a}$.
        \begin{enumerate}
            \item ${\sf trunc}$ is a $2$-Lipschitz function.
            \item $\truncerr$ is a $3$-Lipschitz function.
        \end{enumerate}
	\end{fact}
	\begin{proof}
		Let $a,b\in\mathbb{C}$; we show that $\card{\truncerr(a) - \truncerr(b)} \leq 2\card{a-b}$.
		If $a,b$ are both at most $1$ in absolute value, then the left-hand side is
		$0$ and the claim is trivial. If exactly one of $a$ and $b$ is at most $1$
		in absolute value, say $a$ then
		\[
		\card{\truncerr(a) - \truncerr(b)} =
		\card{\truncerr(b)}
		=\card{b}-1
		\leq \card{b}-\card{a}
		\leq \card{b-a},
		\]
		and we are done. It remains to consider the case that both $a$ and $b$ are at least
		$1$ in absolute value. In that case, first by the triangle inequality
        $\card{\truncerr(a)-\truncerr(b)}
        \leq \card{{\sf trunc}(a) - a - ({\sf trunc}(b) -b)}\leq \card{{\sf trunc}(a)-{\sf trunc}(b)}+\card{a-b}$, and we next analyze ${\sf trunc}$. We have
		\begin{align*}
			\card{{\sf trunc(a) - {\sf trunc}(b)}}
			=\frac{\card{a\card{b} - b\card{a}}}{\card{a}\card{b}}
			\leq
			\frac{\card{a\card{b} - a\card{a}}+\card{a\card{a} - b\card{a}}}{\card{a}\card{b}}
			=
			\frac{\card{\card{b} - \card{a}}}{\card{b}}
			+
			\frac{\card{\card{a-b}}}{\card{b}}
		\end{align*}
        which is at most $
			2\card{a-b}$, and the proof is concluded.
	\end{proof}
	\subsection{The \texorpdfstring{$\mu$}{mu}-norm and the CSP Stability Result}
	\begin{definition}
		For a distribution $\mu$ over $\Sigma\times \Gamma\times \Phi$
		and a function $f\colon (\Sigma^n,\mu_x^{\otimes n})\to\mathbb{C}$,
		we define the $\mu$ semi-norm of $f$ as
		\[
		\norm{f}_{\mu}
		=\sup_{\substack{g\colon \Gamma^n\to\mathbb{C}\\ h\colon \Phi^n\to\mathbb{C}\\ \text{$1$-bounded}}}\card{\Expect{(x,y,z)\sim \mu^{\otimes n}}{f(x)g(y)h(z)}}.
		\]
	\end{definition}
	For general distributions $\mu$, $\norm{\circ}_{\mu}$ is actually a semi-norm; for instance, if the distribution $\mu$ is uniform over $\Sigma\times \Gamma\times \Phi$,
	then $\norm{f}_{\mu} = \card{\Expect{}{f(x)}}$. For most distributions we will be concerned with, though, this semi-norm will actually be a norm.
	
	In our applications, we will need to work with several distributions $\mu$ over triplets that have the same marginal distribution over $x$. A special collection of distributions $\mu$ that we care about is as follows:
	\begin{definition}
		For alphabets $\Sigma$, $\Gamma$ and $\Phi$, a distribution $\nu$ over $\Sigma$ and a parameter $\alpha>0$, define the collections
		\[
		M_{\nu} = \sett{\mu}{\text{pairwise connected distribution over $\Sigma\times \Gamma\times \Phi$ with no $(\mathbb{Z}, +)$-embedding}, \mu_x = \nu},
		\]
		\[
		M_{\alpha} = \sett{\mu}{\mu(x,y,z)\geq \alpha~\forall(x,y,z)\in {\sf supp}(\mu)},
		\]
		and $M_{\nu,\alpha} = M_{\nu}\cap M_{\alpha}$.
	\end{definition}
	\begin{definition}\label{def:semi_norm_nu}
		Let $\nu$ be a distribution over $\Sigma$, let $\Gamma$, $\Phi$ be alphabets and let $M$ be a collection of distributions over $\Sigma\times\Gamma\times \Phi$ such that $\mu_x = \nu$ for all $\mu\in M$. For a function $f\colon (\Sigma^n,\mu_x^{\otimes n})\to\mathbb{C}$,
        we define
        \[
        \norm{f}_{M,\nu} = \sum_{\mu\in M} \norm{f}_\mu.
        \]
        In the special case that $M = M_{\nu,\alpha}$, we refer to the associated norm $\norm{f}_{M,\nu}$ as the $\nu$ semi-norm of $f$, and denote it by
		\[
		\norm{f}_{\nu,\alpha}
		=\sup_{\mu\in M_{\nu,\alpha}}\norm{f}_{\mu}.
		\]
	\end{definition}
    \begin{remark}\label{remark:mismatch1}
    The results below apply to the more general notion
    of semi-norm $\norm{\circ}_{M,\nu}$ with suitable adaptations. However,
    stating it in this generality would complicate the statements
    (that are already quantifier heavy) further, and hence we
    state all of the results for
    $\norm{\circ}_{\nu,\alpha}$.
    \end{remark}
	It will be important for us to understand the type of functions $f$ that have a significant $\mu$-norm, and towards that end, we use the following result from~\cite{BKMcsp4}:
	\begin{theorem}\label{thm:csp4}
		For all $m\in\mathbb{N}$
		there is an Abelian group $G$ such that for all
		$\eps,\alpha>0$ there exists $\delta>0$ and $d\in\mathbb{N}$ such that the following holds.
		Suppose that $\Sigma$, $\Gamma$, $\Phi$ are alphabets of size at most $m$,
		and $\mu$ is a pairwise connected distribution over $\Sigma\times\Gamma\times \Phi$ with no $\mathbb{Z}$ embeddings
		in which the probability of each atom is at least $\alpha$.
        Then there exists an embedding $(\sigma,\gamma,\phi)$ of ${\sf supp}(\mu)$ into $G$, such that
		if $f\colon \Sigma^n\to\mathbb{C}$ is a $1$-bounded function with $\norm{f}_{\mu}\geq \eps$, then
		there is $\chi\in \hat{G}^{\otimes n}$ and $L\colon \Sigma^n\to\mathbb{C}$ of degree at most
		$d$ and $\norm{L}_2\leq 1$ such that
		\[
		\card{\inner{f}{L\cdot \chi\circ\sigma^{\otimes n}}}\geq \delta.
		\]
	\end{theorem}
     Throughout, a function of the type $\chi\circ \sigma^{\otimes n}$ will be referred to as a character function or a character embedding function, and we will often simply write such a function as $\chi\circ \sigma$.
	
	\subsection{Product Functions}\label{sec:setup}
	In this section we introduce some useful notions about embedding functions arising in Theorem~\ref{thm:csp4}.\begin{definition}\label{def:prod_fn}
        Let $\Sigma$ be an alphabet, $G$ be a group and $\sigma\colon \Sigma\to G$ be a map. The class of product functions $\mathcal{P}(\Sigma, G, \sigma)$ is defined as
		the collection of functions of the form $P\colon \Sigma^n\to\mathbb{C}$ for which
		there are $u_1,\ldots,u_n\in\widehat{G}$ and a root of unity $\theta\in\mathbb{C}$ of order at most $\card{G}$, such that for all $x\in\Sigma^n$,
		$P(x) = \theta\prod\limits_{i=1}^{n} u_i(\sigma(x_i))$. In the case that $\theta=1$ we also refer to $P$ as a character product function.
	\end{definition}
    We will often refer to functions
    as in Definition~\ref{def:prod_fn} as product functions, often without specifying the group $G$ or the map $\sigma$. The product function we deal with will often in fact be character product functions.
	\begin{definition}
		We say that a class of product functions $\mathcal{P}(\Sigma, G, \sigma)$
		is $\tau$-separated if for any univariate functions $u,v\colon \Sigma\to\mathbb{C}$
		in it, it either holds that $\card{\inner{u}{v}} = 1$ or else $\card{\inner{u}{v}}\leq 1-\tau$.
	\end{definition}
	
	We need to define a metric on product functions, which we refer to as the symbolic metric.
	\begin{definition}
    \label{def:symbolic_dist}
		Let $P, P'\colon \Sigma^n\to\mathbb{C}$ be product functions in $\mathcal{P}(\Sigma, G, \sigma)$.
		The symbolic metric $\Delta_{{\sf symbolic}}(P,P')$ is defined
		as the minimum number $k$ such that there are
		$u_1,\ldots,u_n,v_1,\ldots,v_n\in\widehat{G}$,
		for which $u_i = v_i$ for all but $k$ of the indices $i\in [n]$, and $\theta,\theta'\in\mathbb{C}$ of absolute value $1$ such that
		\[
		P(x_1,\ldots,x_n) = \theta\prod\limits_{i=1}^{n} u_i(\sigma(x_i)),
		\qquad\qquad
		P'(x_1,\ldots,x_n) = \theta'\prod\limits_{i=1}^{n} v_i(\sigma(x_i)).
		\]
	\end{definition}
	
	\begin{definition}
		For a set $\mathcal{P} = \set{P_1,\ldots,P_r\colon \Sigma^n\to\mathbb{C}}\subseteq \mathcal{P}(\Sigma,G,\sigma)$, we define
		\[
		\spn(\mathcal{P}) = \sett{\prod\limits_{i=1}^{r} P_i^{\alpha_i}}{\alpha_i\in\mathbb{N}, 0\leq \alpha_i\leq |G|~\forall i}.
		\]
	\end{definition}

	\begin{definition}
		The rank of $\mathcal{P}\subseteq\mathcal{P}(\Sigma,G,\sigma)$ is defined as
		\[
        {\sf rk}(\mathcal{P}) = \min_{P\in \spn(\mathcal{P}), P\neq 1}{\Delta_{{\sf symbolic}}\left(P, 1\right)}.
        \]
	\end{definition}
	
	We have the following simple fact, asserting that product functions with large symbolic distance act like
	high-degree monomials.
	\begin{lemma}\label{lemma:symbolic_to_decay}
        Suppose that $\mathcal{P}(\Sigma,G,\sigma)$ is $\tau$-separated and that $P\colon \Sigma^n\to\mathbb{C}$ is a product function from $\mathcal{P}(\Sigma,G,\sigma)$
        such that $\Delta_{{\sf symbolic}}(P,1)\geq 2$. Let $\nu$ be a distribution whose support is $\Sigma$ and in which the probability of each atom is at least $\alpha$. Then for all $c>0$ and $\xi\in [c,1]$ we have that
		\[
        \norm{\mathrm{T}_{1-\xi} P}_2\leq 2^{-\Omega_{\alpha,\tau,c}(\Delta_{{\sf symbolic}}(P,1))},
        \]
        where the $2$-norm and the operator $\mathrm{T}_{1-\xi}$
        are with respect to $\nu^n$.
	\end{lemma}
    \begin{proof}
       By definition, we may write
       $P(x) = \prod\limits_{i=1}^{n}u_i(\sigma(x_i))$ where for at least $k :=\Delta_{{\sf symbolic}}(P,1)$ many $i$'s
       we have that $u_i\circ\sigma$ is not constantly $1$. By minimality, it follows that for at least $k-1$ of the $i$ we have that $u_i\circ \sigma$ is not constant, so by separatedness it has variance at least $\Omega_{\alpha,\tau}(1)$.
       For each such $i$ get that $\norm{\mathrm{T}_{1-\xi}u_i\circ \sigma}_2\leq 1-\Omega_{\alpha,\tau}(\xi)\leq 1-\Omega_{\alpha,\tau,\xi}(1)$, and the result follows as
       $\norm{\mathrm{T}_{1-\xi} P}_2=\prod\limits_{i=1}^{n}\norm{\mathrm{T}_{1-\xi} u_i\circ \sigma}_2$.
    \end{proof}

    Observe that if ${\sf rk}(\mathcal{P})$ is large, then the above lemma (when $\xi = 1$) implies that for every $1\neq P\in \spn(\mathcal{P})$, the average of $P$ is negligible.

	\section{The Noise Operator and Its Properties}
	In this section we introduce a variant of the standard noise operator over product spaces that will be crucial
	for our applications.
	\begin{definition}\label{def:noise_op_P}
		Let $\Sigma$ be a finite alphabet, 
		let $\mathcal{P} = \set{P_1,\ldots,P_r\colon \Sigma^n\to\mathbb{C}}$ be a collection of functions, let $\nu$ be a distribution over $\Sigma$,
        and let $I\subseteq [n]$.
		For each $x\in \Sigma^n$ we define the distribution $\mathrm{T}_{\nu, \mathcal{P}, I} x$
		as:
		\begin{enumerate}
			\item Sample $y\sim \nu^{\otimes n}$ conditioned on $y_{\overline{I}} = x_{\overline{I}}$ and $P_i(y) = P_i(x)$ for all $i$.
			\item Output $y$.
		\end{enumerate}
        For $\eps>0$, we define the distribution of $\mathrm{T}_{\nu,\mathcal{P},1-\eps} x$ by sampling $I\subseteq_{\eps}[n]$, and then outputting $y\sim \mathrm{T}_{\nu,\mathcal{P},I} x$.
	\end{definition}
	As usual, we will associate $\mathrm{T}_{\nu, \mathcal{P}, I},\mathrm{T}_{\nu, \mathcal{P}, 1-\eps}$ with an averaging operator over functions,
	which by abuse of notation we denote as $\mathrm{T}_{\nu, \mathcal{P}, I}, \mathrm{T}_{\nu, \mathcal{P}, 1-\eps} \colon L_2(\Sigma^n,\nu^{\otimes n})\to L_2(\Sigma^n,\nu^{\otimes n})$. We define
	\[
    \mathrm{T}_{\nu, \mathcal{P}, I}f(x)=\Expect{y\sim \mathrm{T}_{\nu, \mathcal{P}, I}x}{f(y)},
    \qquad
	\mathrm{T}_{\nu, \mathcal{P}, 1-\eps} f(x) =
    \Expect{I\subseteq_{\eps} [n]}{\Expect{y\sim \mathrm{T}_{\nu, \mathcal{P}, I} x}{f(y)}}
    =\Expect{I\subseteq_{\eps} [n]}{\mathrm{T}_{\nu, \mathcal{P}, I}f(x)}.
	\]
	When the distribution $\nu$ is clear from context, we will often drop $\nu$ from notation and simply write  $\mathrm{T}_{\mathcal{P},I}$
    and $\mathrm{T}_{\mathcal{P},1-\eps}$ in place of $\mathrm{T}_{\nu,\mathcal{P},I}$ and $\mathrm{T}_{\nu,\mathcal{P},1-\eps}$. We will use noise operators as a replacement for harsh truncations: as we shall see in $L_2$ it acts a truncation-like yet operator, which in addition preserves the boundedness of functions. The distribution $\nu$ will sometimes be omitted from the notation and will be clear from context.

    \subsection{Some Basic Properties of the Noise Operator}
    The goal of this section is to prove a few basic properties of the noise operator $\mathrm{T}_{\nu, \mathcal{P}, 1-\eps}$, and we fix a distribution $\nu$ henceforth.
	\subsubsection{A Stationary Distribution of \texorpdfstring{$\mathrm{T}_{\nu,\mathcal{P}, 1-\eps}$}{the noise operator}}

	We first prove that $\mathrm{T}_{\nu, \mathcal{P}, 1-\eps}$ is a Markov chain over $\Sigma^n$ and $\nu^{\otimes n}$ is a stationary distribution of it.\footnote{We remark that whenever the set $\mathcal{P}$ is non-empty, the Markov chain $\mathrm{T}_{\nu, \mathcal{P}, 1-\eps}$ is disconnected, so it has multiple stationary distributions. This will not be important for us.}
	\begin{fact}\label{fact:noise_op_basic_prop1}
        For all $I\subseteq[n]$ and $\mathcal{P}$, $\nu^{\otimes n}$
        is a stationary distribution of $\mathrm{T}_{\nu,\mathcal{P},I}$. Consequently $\nu^{\otimes n}$ is a stationary distribution of $\mathrm{T}_{\nu, \mathcal{P}, 1-\eps}$.	\end{fact}
	\begin{proof}
		Fix $I$, write $\mathcal{P} = \{P_1,\ldots,P_r\}$ and fix $w\in \Sigma^n$. We  calculate the probability that sampling $x\sim\nu^{\otimes n}$ and $y\sim \mathrm{T}_{\nu,\mathcal{P},I} x$ we have $y=w$. Denoting
		$a_i = P_i(w)$, we have
		\[
		\Prob{x,y}{y = w}
		=\Prob{x,y}{P_i(x) = a_i\forall i, x_I = w_I}\cdot
		\cProb{x, y}{P_i(x) = a_i\forall i, x_I = w_I}{y=w}.
		\]
		Let $A$ be the set of $x\in \Sigma^n$ such that $P_i(x) = a_i$ for all $i$ and $x_I = w_I$, so that the first term is equal to $\nu^{\otimes n}(A)$. By definition of $\mathrm{T}_{\nu,\mathcal{P},I}$, $y$ is distributed according to $y\sim \nu^{\otimes n}$ conditioned on $y\in A$, so the second term is equal to $\nu^{\otimes n}(w)/\nu^{\otimes n}(A)$. Together we get $\Prob{x,y}{y = w}=\nu^{\otimes n}(w)$.
	\end{proof}
	
	\subsubsection{Relating Different Noise Operators}
	We next show that if $P'$ is close to $\spn(\mathcal{P})$, then the operators
	$\mathrm{T}_{\nu, \mathcal{P}, 1-\eps}$ and $\mathrm{T}_{\nu, \mathcal{P}\cup \{P'\}, 1-\eps}$ are close.
	\begin{fact}\label{fact:noise_op_basic_prop2}
		Let $\mathcal{P}$ be a collection of product functions, let $\eps>0$, let $P'\colon \Sigma^n\to\mathbb{C}$ be a product function and suppose that
		$k = \min_{P\in \spn(\mathcal{P})}\Delta(P,P')$. Then
		\begin{enumerate}
			\item There is a coupling of $(x, y, y')$ such that $(x,y)$ is distributed according to $(x, \mathrm{T}_{\nu, \mathcal{P}, 1-\eps} x)$,
			$(x,y')$ is distributed according to $(x, \mathrm{T}_{\nu, \mathcal{P}\cup\{P'\}, 1-\eps}x)$ and $\Prob{}{y\neq y'}\leq k\eps$.
			\item For any $1$-bounded function $f\colon \Sigma^n\to\mathbb{C}$,
			$\norm{\mathrm{T}_{\nu, \mathcal{P}\cup \{P'\}, 1-\eps} f - \mathrm{T}_{\nu, \mathcal{P}, 1-\eps} f}_2\leq 2\sqrt{k\eps}$.
		\end{enumerate}
	\end{fact}
	\begin{proof}
		We begin with the first item. Let $P\in \spn(\mathcal{P})$ be the $P$ achieving the minimum, and let $K\subseteq [n]$ be the set of coordinates where $P$ and $P'$ differ. We note that if $K\subseteq\bar{I}$, then
        $\mathrm{T}_{\nu, \mathcal{P}\cup \{P'\}, I} = \mathrm{T}_{\nu, \mathcal{P}, I}$, and as $\Prob{I\subseteq_{\eps}[n]}{K\subseteq\bar{I}}\geq 1-k\eps$, we get that for each $x$
        the statistical distance between
        $\mathrm{T}_{\nu, \mathcal{P}\cup \{P'\}, 1-\eps}x$
        and
        $\mathrm{T}_{\nu, \mathcal{P}, 1-\eps}x$ is at most
        $k\eps$, and the first item follows.
		For the second item, fix the coupling $(x,y,y')$ so that we may write
		\[
		\norm{\mathrm{T}_{\nu, \mathcal{P}\cup \{P'\}, 1-\eps} f - \mathrm{T}_{\nu, \mathcal{P}, 1-\eps} f}_2^2
		=\Expect{x}{\card{\Expect{y, y'}{f(y) - f(y')}}^2}
		\leq 4\Expect{x}{1_{y\neq y'}}
		\leq 4k\eps.\qedhere
		\]
	\end{proof}

	\subsubsection{Nearly Preserving Low Degree Functions}
	We would like to argue that $\mathrm{T}_{\nu,\mathcal{P}, 1-\eps}$ roughly preserves low-degree functions.
	Toward this end, for each $I\subseteq[n]$ we introduce the function
	\[
	A_{\mathcal{P},I}(x) =
	\Prob{\substack{ x'\sim \mathrm{T}_Ix}}{P_i(x') = P_i(x)~\forall i}.
	\]
	We have the following fact:
	\begin{fact}\label{fact:P_val_stable}
		Suppose that $\Sigma$ is an alphabet of size at most $m$, $\mathcal{P}=\set{P_1,\ldots,P_r}\subseteq\mathcal{P}(\Sigma,G,\sigma)$ where $G$ has size at most $m$, and
		$\nu$ is a distribution over $\Sigma$. Then for all $I\subseteq[n]$ and $\tau>0$ we have that
		\[
		\Prob{x\sim \nu^{\otimes n}}{A_{\mathcal{P},I}(x)\leq \tau}\leq O_{r, m}(\tau).
		\]
	\end{fact}
	\begin{proof}
		For each $\vec{a}\in \prod\limits_{i=1}^{r} {\sf Image}(P_i)$ define
		\[
		B_{\vec{a}} = \sett{x\in \Sigma^n}{P_i(x) = a_i~\forall i=1,\ldots,r},
		\qquad
		B_{\vec{a}}' = \sett{x\in B_{\vec{a}}}{A_{\mathcal{P},I}(x)\leq \tau}.
		\]
        Note that for all $\vec{a}$,
		\begin{align*}
		\Prob{x\sim\nu^{\otimes n}, y\sim\mathrm{T}_{I}x}{x,y\in B_{\vec{a}}'}
		=\inner{1_{B_{\vec{a}}'}}{\mathrm{T}_{I}1_{B_{\vec{a}}'}}
		=\norm{\mathrm{T}_{I}1_{B_{\vec{a}}'}}_2^2
		&\geq
		\norm{\mathrm{T}_{I}1_{B_{\vec{a}}'}}_1^2
        =\nu^{\otimes n}(B_{\vec{a}}')^2.
		\end{align*}
        In the second transition we used the fact that $\mathrm{T}_{I}$ is self adjoint and $\mathrm{T}_{I}^2 = \mathrm{T}_{I}$.  On the other hand,
		\[
		\Prob{x\sim \nu^{\otimes n}, y\sim\mathrm{T}_{I}x}{x,y\in B_{\vec{a}}'}
		=\nu^{\otimes n}(B_{\vec{a}}')\cdot
        \cExpect{x\sim\nu^{\otimes n}}{x\in B_{\vec{a}}'}{A_{\mathcal{P},I}(x)}
		\leq
		\nu^{\otimes n}(B_{\vec{a}}')\tau.
		\]
		Combining the two inequalities we get that $\nu^{\otimes n}(B_{\vec{a}}')\leq \tau$. Defining $X_{{\sf bad}} = \bigcup_{\vec{a}\in \prod\limits_{i=1}^{r} {\sf Image}(P_i)} B_{\vec{a}}'$ we get that as the size of the image of each one of the $P_i$'s is $O_m(1)$, it holds that
		\[
		\Prob{x\sim \nu^{\otimes n}}{A_{\mathcal{P},I}(x)\leq \tau}\leq
		\nu^{\otimes n}(X_{{\sf bad}})\lll_{r,m} \tau.\qedhere
		\]
	\end{proof}
	
	\begin{claim}\label{claim:pres_low_deg_functions}
		Let $d,r,m\in\mathbb{N}$ and $\alpha,\eps>0$.
		Suppose that $\Sigma$ has size at most $m$, $\nu$ is a distribution over $\Sigma$
		in which the probability of each atom is at least $\alpha$, and
		$\mathcal{P} \subseteq\mathcal{P}(\Sigma, G, \sigma)$
		be a collection of size at most $r$, size of $G$ at most $m$, and $\sigma\colon \Sigma\to G$ is some map. Then for any $L\colon \Sigma^n\to\mathbb{C}$ of degree at most $d$
		we have that
		\[
		\norm{(I-\mathrm{T}_{\nu,\mathcal{P}, 1-\eps})L}_2^2\lll_{d,m,r,\alpha} \eps^{1/3}\norm{L}_2^2.
		\]
	\end{claim}
	\begin{proof}
		Normalizing, we assume that $\norm{L}_2 = 1$. The left-hand side is equal to
		\begin{align}\label{eq1}
			\Expect{x\sim\nu^{\otimes n}}{\card{\Expect{y\sim \mathrm{T}_{\mathcal{P}, 1-\eps} x}{L(x)-L(y)}}^2}
		%&\leq	\Expect{x}{\Expect{y\sim \mathrm{T}_{\mathcal{P}, 1-\eps} x}{\card{L(x)-L(y)}^2}}\notag\\
            &\leq\Expect{x\sim\nu^{\otimes n}, I\subseteq_{\eps}[n]}
            {\Expect{y\sim \mathrm{T}_{\mathcal{P},I} x}{\card{L(x)-L(y)}^2}}.
		\end{align}
		Let $\tau>0$ be a parameter to be chosen. We break right-hand side of~\eqref{eq1} to
		$(\rom{1})+(\rom{2})$ where:
		\[
		(\rom{1}) = \Expect{\substack{x\sim\nu^{\otimes n}\\I\subseteq_{\eps}[n]}}{\Expect{y\sim \mathrm{T}_{\mathcal{P}, I} x}{\card{L(x)-L(y)}^2}1_{A_{\mathcal{P},I}(x)<\tau}},
        \]
	\[
		(\rom{2})
        =\Expect{\substack{x\sim\nu^{\otimes n}\\ I\subseteq_{\eps}[n]}}{\Expect{y\sim \mathrm{T}_{\mathcal{P}, I} x}{\card{L(x)-L(y)}^2}1_{A_{\mathcal{P},I}(x)\geq\tau}}.
		\]
		For $(\rom{1})$, using Cauchy-Schwarz twice we have
		\begin{align*}
		(\rom{1})
            &\leq
            \sqrt{\Expect{\substack{x\sim\nu^{\otimes n}\\ I\subseteq_{\eps}[n]}}{
            \Expect{y\sim \mathrm{T}_{\mathcal{P},I} x}{\card{L(x)-L(y)}^4}}}
            \sqrt{\Expect{I\subseteq_{\eps}[n]}{\Prob{x\sim \nu^{\otimes n}}{A_{\mathcal{P},I}(x)\leq \tau}}}\\
            &\lll
            \sqrt{\Expect{x\sim\nu^{\otimes n}}{\card{L(x)}^4}
            +\Expect{\substack{x\sim\nu^{\otimes n}\\ y\sim \mathrm{T}_{\mathcal{P},1-\eps}x}}{\card{L(y)}^4}}
            \sqrt{
        \Expect{I\subseteq_{\eps}[n]}{\Prob{x\sim \nu^{\otimes n}}{A_{\mathcal{P},I}(x)\leq \tau}}}
            \\
		&\lll \norm{L}_4^2\sqrt{
        \Expect{I\subseteq_{\eps}[n]}{\Prob{x\sim \nu^{\otimes n}}{A_{\mathcal{P},I}(x)\leq \tau}}},
		\end{align*}
        where we used Fact~\ref{fact:noise_op_basic_prop1} in
        the last inequality.
		By hypercontractivity $\norm{L}_4^2\lll_{\alpha,d}\norm{L}_2^2 = 1$, and combining with Fact~\ref{fact:P_val_stable} we conclude that $(\rom{1})\lll_{\alpha,m,d}\sqrt{\tau}$.
		For $(\rom{2})$, we have that
		\begin{align*}
			(\rom{2})
			=
            \Expect{\substack{x\sim\nu^{\otimes n}\\I\subseteq_{\eps}[n]}}{
\frac{\Expect{y\sim\mathrm{T}_I x}{\card{L(y) - L(x)}^2 1_{P_i(y) = P_i(x)~\forall i}}}{\mathcal{A}_{\mathcal{P},I}(x)}1_{A_{\mathcal{P},I}(x)\geq \tau}},		
\end{align*}
		which is at most
		\[
		\frac{1}{\tau} \Expect{\substack{x\sim\nu^{\otimes n}\\I\subseteq_{\eps}[n]}}{\Expect{y\sim \mathrm{T}_{I} x}{\card{L(y) - L(x)}^2}}
		=\frac{1}{\tau}\left(2-2\inner{L}{\mathrm{T}_{1-\eps}L}\right)
		\leq
		\frac{2}{\tau}\left(1-(1-\eps)^d\right),
		\]
		which is at most $\frac{2d}{\tau} \eps$. Combining, we get 		$\eqref{eq1}\lll_{\alpha,m,d} \sqrt{\tau} + \frac{\eps}{\tau}$,
		and choosing $\tau = \eps^{2/3}$ finishes the proof.
	\end{proof}

	\subsection{An Approximate Formula for \texorpdfstring{$\mathrm{T}_{\nu,\mathcal{P}, 1-\eps}f$}{the noise operator}}
	Note that any function $P\in \spn(\mathcal{P})$ is an eigenvector of $T_{\mathcal{P}, 1-\eps}$ of eigenvalue $1$. Thus, by Claim~\ref{claim:pres_low_deg_functions}
	functions of the form $P\cdot L$ are nearly preserved under the operator $T_{\mathcal{P}, 1-\eps}$ (provided that the degree of $L$ is sufficiently small compared to $1/\eps$).
	In this section, we show that for any function $f$, $\mathrm{T}_{\mathcal{P}, 1-\eps}f$ may be approximated by a linear combination of such functions. This statement, which is Lemma~\ref{lem:approx_formula_noise_op_fancy} below, will be used in two contexts:
	\begin{enumerate}
		\item First, in Section~\ref{sec:reg_lemmas} we will use this approximate formula in the proof of our $\mu$-regularity lemma.
		\item Second, in Section~\ref{sec:mixed_invariance} we will use it to state and prove our mixed invariance principle.
	\end{enumerate}
	
	Towards getting such an approximation, we need more basic properties of the function $A_{\mathcal{P},I}(x)$. In Claim~\ref{fact:P_val_stable} we proved that its values are almost
	always positive numbers bounded away from $0$. In the next fact, we show that $A_{\mathcal{P},I}(x)$ can be written as a linear combination of functions of the form $P\cdot L$ for
	bounded functions $L$ which are $\bar{I}$-juntas.
	\begin{fact}\label{fact:A_fn_close_to_pl}
		For all $r,m\in\mathbb{N}$,
        $I\subseteq[n]$ and any collection of $r$
        product functions $\mathcal{P}\subseteq \mathcal{P}(\Sigma, G, \sigma)$ with $|\Sigma|=m$ and $|G|= O_m(1)$,
		we may write
		\[
		A_{\mathcal{P},I}(x) = \sum\limits_{P\in \spn(\mathcal{P})} P\cdot \tilde{L}_{I,P},
		\]
		where
		for every $P$,  $\tilde{L}_{I,P}$ is
        a $\bar{I}$-junta, $O_{m,r}(1)$ bounded function.
	\end{fact}
	\begin{proof}
		Let $\mathcal{P} = \{ P_1, P_2, \ldots, P_r\}$. Note that as all of the functions $P_i$ get the values of characters over an Abelian group of size at most $O_m(1)$, we have that in
		the set $\mathcal{I} = \bigcup_{i}{\sf Image}(P_i)$ the distance between any two distinct points is $\Omega_{m}(1)$. Thus, we may find a bi-variate
		polynomial $Q(z_1,z_2)$ of degree at most $O_{m}(1)$ and coefficients that are at most $O_{m}(1)$ in absolute value
		such that $Q(z_1,z_2) = 1$ if $z_1 = z_2\in \mathcal{I}$, $Q(z_1,z_2) = 0$ if $z_1, z_2\in \mathcal{I}$ are distinct.\footnote{A polynomial $Q$ satisfying these properties
			may be constructed via Lagrange interpolation, for example.}Using $Q$, we may write
		\[
		A_{\mathcal{P},I}(x) = \Expect{y\sim \mathrm{T}_{I}x}{\prod\limits_{i=1}^{r}Q(P_i(y), P_i(x))}.
		\]
		Write $Q(z_1,z_2) = \sum\limits_{j,k=0}^{d} a_{j,k}z_1^jz_2^k$ for $d = O_{m}(1)$ and plug it in above to get that
		\begin{align*}
			A_{\mathcal{P},I}(x)
			&=
			\sum\limits_{j_1,\ldots,j_r,k_1,\ldots,k_r=0}^{d} a_{j_1,k_1}\cdots a_{j_r,k_r} P_1(x)^{k_1}\cdots P_r(x)^{k_r}
			\Expect{y\sim \mathrm{T}_{I}x}{P_1(y)^{j_1}\cdots P_r(y)^{j_r}}\\
			&=
			\sum\limits_{k_1,\ldots,k_r=0}^{d}  P_1(x)^{k_1}\cdots P_r(x)^{k_r}
			\mathrm{T}_{I}\left(
            \sum\limits_{j_1,\ldots,j_r=1}^d a_{j_1,k_1}\cdots a_{j_r,k_r}P_1^{j_1}\cdots P_r^{j_r}\right)(x).
		\end{align*}
	\end{proof}
	
	We will also want to show that the function $1/A_{\mathcal{P},I}(x)$ is close to being a linear combination of functions of the type $P\cdot L_I$, but
    some care is required for $x$ such that $A_{\mathcal{P},I}(x)$ is close to $0$.
	The following fact asserts that such approximation holds when we ignore this type of inputs.
	
	\begin{fact}\label{fact:approx_1_over}
		Consider the setup from Fact~\ref{fact:A_fn_close_to_pl}. For all $m, r\in\mathbb{N}$, $\tau,\xi>0$ and $I\subseteq [n]$, there is $A_I''\colon \Sigma^n\to\mathbb{C}$ with the following properties:
		\begin{enumerate}
			\item $A_I''$ may be written as $A_I''(x) = \sum\limits_{P\in \spn(\mathcal{P})} P\cdot \tilde{L}_{I,P}$
			where for all $P$, $\tilde{L}_{I,P}$ is a $\bar{I}$-junta which is $O_{m,\xi,\tau, r}(1)$-bounded.
			\item $\Expect{x}{\card{\frac{1}{A_{\mathcal{P},I}(x)} - A_I''(x)}^2 1_{A_{\mathcal{P},I}(x)\geq \tau}}\leq \xi$.
		\end{enumerate}
	\end{fact}
	\begin{proof}
		Fix $x$ such that $A_{\mathcal{P},I}(x)\geq \tau$, and note that clearly $A_{\mathcal{P},I}(x)\leq 1$. Thus, we may write
		\[
		\frac{1}{A_{\mathcal{P},I}(x)} = \frac{1}{1-(1-A_{\mathcal{P},I}(x))}
		= \sum\limits_{k=0}^{\infty}(1-A_{\mathcal{P},I}(x))^k
		=\sum\limits_{k=0}^{T}(1-A_{\mathcal{P},I}(x))^k + {\sf err}(x),
		\]
        where $\card{{\sf err}(x)}\leq e^{-\tau T}$.
	We take $T = \log(2/\xi)/ \tau$ and define $A_I''(x) = \sum\limits_{k=0}^{T}(1-A_{\mathcal{P},I}(x))^k$; then the second item is clear. For the first item, we use the formula for $A_{\mathcal{P},I}$ from Fact~\ref{fact:A_fn_close_to_pl}, expand $A_I''$ and re-group terms to get that $A_I''$ may be written
		as a linear combination of functions of the form $P\cdot \tilde{L}_{I,P,1}\cdots \tilde{L}_{I,P,T+1}$ where $P\in\spn(\mathcal{P})$ and each
		$\tilde{L}_{I,P,j}$ is a $\bar{I}$-junta which is $O_{m,r}(1)$-bounded. In particular, we may write
		\[
		A_I''(x) = \sum\limits_{P\in \spn(\mathcal{P})}P\tilde{L}_{I,P}
		\]
		where $\tilde{L}_{I,P}$ is a $\bar{I}$-junta which is $O_{m,T,r}(1) = O_{m,\xi,\tau,r}(1)$-bounded.
	\end{proof}

       We next give a similar approximate
       formula for $\mathrm{T}_{\nu,\mathcal{P}, I}f$.
       \begin{lemma}\label{lem:approx_formula_noise_op_fancy_I}
		Fix $m,r\in\mathbb{N}$, $\alpha, \xi>0$, suppose that $\Sigma$ has size at most $m$, that $\nu$ is a distribution over $\Sigma$
		in which the probability of each atom is at least $\alpha$, and
		$\mathcal{P} = \set{P_1,\ldots,P_r}$ is a collection of product functions from $\mathcal{P}(\Sigma, G, \sigma)$
		where $G$ is a group of size at most $m$.
		Then for any $1$-bounded function $f\colon \Sigma^n\to\mathbb{C}$ and $I\subseteq [n]$, there exists a function
		$f_I'\colon\Sigma^n\to\mathbb{C}$ such that:
		\begin{enumerate}
			\item The function $f'_I$ approximates $\mathrm{T}_{\nu,\mathcal{P}, I}f$: $\norm{\mathrm{T}_{\nu,\mathcal{P}, I}f - f_I'}_2\leq \xi$.
			\item The function $f'_I$ can be written as
			$f_I'(x) = \sum\limits_{P\in \spn(\mathcal{P})}P(x)\cdot L_{I,P}(x)$,
			where for all $P$, $L_{I,P}$ is a $\bar{I}$-junta and is $O_{m,r,\alpha,\xi}(1)$-bounded.
		\end{enumerate}
	\end{lemma}
	\begin{proof}
		Take the parameters
		\[
		0 < \zeta \ll\tau \ll r^{-1}, m^{-1}, \alpha, \xi\leq 1.
		\]
		Let $\tau'\in [\tau,\tau+\sqrt{\zeta}]$ be a parameter to be determined, define $F_1(x) = \mathrm{T}_{\mathcal{P}, I}f(x) 1_{A_{\mathcal{P},I}(x)\geq \tau'}$ and note that
		\begin{equation}\label{eq4}
			\norm{\mathrm{T}_{\mathcal{P}, I}f - F_1}_2
			= \norm{\mathrm{T}_{\mathcal{P}, I}f(x) 1_{A_{\mathcal{P},I}(x)\leq \tau}}_2
			\leq \norm{1_{A_{\mathcal{P},I}(x)\leq \tau'}}_2
			\lll_{m,r}\sqrt{\tau'}
			\leq \tau^{1/8},
		\end{equation}
		where we used Fact~\ref{fact:P_val_stable}. Next, expanding the definition we get that
		\begin{align*}
			\mathrm{T}_{\mathcal{P}, I} f(x)
			&=
            \frac{1}{A_{\mathcal{P},I}(x)}
			\underbrace{\Expect{\substack{x'\sim\mathrm{T}_Ix}}{f(x') 1_{P_i(x') = P_i(x)~\forall i}}}_{(\rom{1})}.
		\end{align*}
		Take $A_I''(x)$ from Fact~\ref{fact:approx_1_over} with the parameter $\xi$ therein being $\tau$ in the current setting. Then
		\[
        \norm{A_I''1_{A_{\mathcal{P},I}\geq \tau'} - \frac{1}{A_{\mathcal{P},I}}1_{A_{\mathcal{P},I}\geq \tau'}}_2\leq \tau.
        \]
        Thus, as $(\rom{1})(x)$ is
		a $1$-bounded function of $x$ we conclude that for $F_2(x) =  (\rom{1})(x) A_I''(x) 1_{A_{\mathcal{P},I}(x)\geq \tau'}$
		\begin{equation}\label{eq5}
			\norm{F_1 - F_2}_2\leq \tau.
		\end{equation}
		Next, note that $\Expect{x\sim\nu^{\otimes n}}{\card{(\rom{1})(x)}^2\card{A''_I(x)}^2}\leq O_{\tau}(1)$, therefore there exists $1\leq j\leq 1/\sqrt{\zeta}$ such that
		\[
		\Expect{x}{(\rom{1})(x)^2 A_I''(x)^2 1_{A_{\mathcal{P},I}(x)\in [\tau + j\zeta, \tau + (j+1)\zeta]}}\leq O_{\tau}(\sqrt{\zeta})\leq \zeta^{1/4},
		\]
		and we fix such $j$. Define the continuous function $h\colon [0,1]\to [0,1]$ so that $h(t) = 0$ for $t\leq \tau + j\zeta$,
		$h(t) = 1$ for $t\geq \tau + (j+1)\zeta$ and we linearly interpolate between the two ranges. Also, take
		$\tau' = \tau + j\zeta$. Defining $F_3(x) = (\rom{1})(x) A_I''(x) h(A_{\mathcal{P},I}(x))$,
		we get that
		\begin{equation}\label{eq6}
			\norm{F_2 - F_3}_2
			\leq
			\sqrt{\Expect{x}{(\rom{1})(x)^2 A_I''(x)^2 1_{A_{\mathcal{P},I}(x)\in [\tau + j\zeta, \tau + (j+1)\zeta]}}}
			\leq \zeta^{1/4}.
		\end{equation}
		
		At this point, our approximating function $F_3$ almost fits the form as needed for $f'_I$, but we still need to study the structure of $(\rom{1})(x)$ and of $h(A_{\mathcal{P},I}(x))$ further.		
		For $h(A_{\mathcal{P},I}(x))$, by Weierstrass approximation theorem we may find a polynomial $W\colon [0,1]\to \mathbb{R}$
		such that $\card{W(t) - h(t)}\leq \zeta$ for all $t\in [0,1]$. We define $F_4(x) = (\rom{1})(x) A_I''(x) W(A_{\mathcal{P},I}(x))$ and get
		that
		\begin{equation}\label{eq7}
			\norm{F_3 - F_4}_2
			\leq
			\zeta \sqrt{\Expect{x}{(\rom{1})(x)^2 A_I''(x)^2}}
			\leq
			\zeta \sqrt{\Expect{x}{A_I''(x)^2}}
			\leq O_{\tau}(\zeta)
			\leq \sqrt{\zeta}.
		\end{equation}
		
		\paragraph{The Structure of $(\rom{1})$.}
		As in the proof of Fact~\ref{fact:A_fn_close_to_pl}, take a bi-variate polynomial $Q(z_1,z_2)$ of degree $O_{m,r}(1)$ whose coefficients are all bounded by $O_{m,r}(1)$
		in absolute value such that $Q(z_1,z_2) = 1_{z_1 = z_2}$ for every $z_1,z_2\in \bigcup_{i=1}^{r}{\sf Image}(P_i)$. Then
		\[
		(\rom{1})(x)
		=\Expect{x'\sim \mathrm{T}_{I} x}{f(x') \prod\limits_{i=1}^{r}Q(P_i(x'), P_i(x))}.
		\]
		Writing $Q(z_1,z_2) = \sum\limits_{i,j=0}^{d} a_{i,j} z_1^iz_2^j$ where $\card{a_{i,j}} = O_{m,r}(1)$, we get that
		\begin{align*}
			(\rom{1})(x)
			&=
			\sum\limits_{j_1,\ldots,j_r=0}^{d}P_1(x)^{j_1}\cdots P_r(x)^{j_r}
			\Expect{x'\sim\mathrm{T}_{I}x}{\sum\limits_{i_1,\ldots,i_r=0}^{d} a_{i_1,j_1}\cdots a_{i_r,j_r}
				f(x') P_1(x')^{i_1}\cdots P_r(x')^{i_r}}\\
			&=\sum\limits_{j_1,\ldots,j_r=0}^{d}P_1(x)^{j_1}\cdots P_r(x)^{j_r}  F_{j_1,\ldots,j_r}(x),
		\end{align*}
		where
        \begin{align*}
        F_{j_1,\ldots,j_r}(x)
        &= \mathrm{T}_{I}\left(\sum\limits_{i_1,\ldots,i_r=0}^{d} a_{i_1,j_1}\cdots a_{i_r,j_r} fP_1^{i_1}\cdots P_r^{i_r}\right)(x)\\
        &=\Expect{x'\sim \mathrm{T}_I x}{
        \sum\limits_{i_1,\ldots,i_r=0}^{d} a_{i_1,j_1}\cdots a_{i_r,j_r} f(x') P_1^{i_1}(x')\cdots P_r^{i_r}(x')
        }.
        \end{align*}
        Note each $F_{j_1,\ldots,j_r}$ is $O_{m,r}(1)$-bounded by the triangle inequality, and is also a $\bar{I}$-junta.
		
		\paragraph{Combining the Structures.}
		Write $W(t) = \sum\limits_{i=0}^{q} a_i t^i$, where $q, \card{a_i}\leq O_{\zeta}(1)$ for all $i$, so that
		\[
		W(A_{\mathcal{P},I}(x)) = \sum\limits_{i=0}^{q} a_i A_{\mathcal{P},I}(x)^i.
		\]
		Plugging in the formula for $A_{\mathcal{P},I}(x)$ from Fact~\ref{fact:A_fn_close_to_pl}, we get that
		$W(A_{\mathcal{P},I}(x))$ may be written as a linear combination of at most $O_q(1)$ may terms each
		of the form $P\cdot \tilde{L}_{I,P,1}(x)\cdots \tilde{L}_{I,P,q}(x)$, where
		for all $i$, $\tilde{L}_{I,P,i}$ is a $\bar{I}$ junta which is $O_{m,r}(1)$-bounded. Re-arranging gives
		\[
		W(A_{\mathcal{P}}(x)) = \sum\limits_{P\in \spn(\mathcal{P})}P \cdot \tilde{L}_{I,P},
		\]
		where $\tilde{L}_{I,P}$ is a $\bar{I}$-junta which is $O_{\zeta}(1)$-bounded. By choice of $A_I''$ from Fact~\ref{fact:approx_1_over}, we may write
		\[
		A_I''(x) = \sum\limits_{P\in \spn(\mathcal{P})} P\cdot \tilde{L}'_{I,P}
		\]
		where $\tilde{L}'_{I,P}$ is a $\bar{I}$-junta which is $O_{\tau}(1)$-bounded.
		
		Combining the formulas for $(\rom{1})$, $W(A_{\mathcal{P}}(x))$ and $A''(x)$ and multiplying out, we get that
		\[
		F_4(x) = \sum\limits_{P\in \spn(\mathcal{P})} P\cdot \tilde{L}_{I,P}'',
		\]
		where each $\tilde{L}_{P}''$ is a linear combination of at most $O_{m,r}(1)$ products of at most $3$ of the functions
		$F_{j_1,\ldots,j_r}$, $\tilde{L}_{I,Q'}$ and $\tilde{L}'_{I,Q''}$. Thus, $\tilde{L}_{P}''$ is a $\bar{I}$-junta which is $O_{\zeta}(1)$-bounded.
        Combining~\eqref{eq4},~\eqref{eq5},~\eqref{eq6} and~\eqref{eq7} gives the second item of the lemma.
	\end{proof}

	\begin{lemma}\label{lem:approx_formula_noise_op_fancy}
		For all $m,r\in\mathbb{N}$, $\alpha, \eps>0$ and $\xi>0$ there exist $C, D\in\mathbb{N}$ such that the following holds.
		Suppose that $\Sigma$ has size at most $m$, $\nu$ is a distribution over $\Sigma$
		in which the probability of each atom is at least $\alpha$, and
		$\mathcal{P} = \set{P_1,\ldots,P_r}$ is a collection of product functions from $\mathcal{P}(\Sigma, G, \sigma)$
		where $G$ is a group of size at most $m$.
		Then for any $1$-bounded function $f\colon \Sigma^n\to\mathbb{C}$, there exists a function
		$f'\colon\Sigma^n\to\mathbb{C}$ such that:
		\begin{enumerate}
			\item The function $f'$ approximates $\mathrm{T}_{\nu,\mathcal{P}, 1-\eps}f$: $\norm{\mathrm{T}_{\nu,\mathcal{P}, 1-\eps}f - f'}_2\leq \xi$.
			\item The function $f'$ can be written as
			\[
			f'(x) = \sum\limits_{P\in \spn(\mathcal{P})}P(x)\cdot L_{P}(x)
			\]
			where for all $P$, ${\sf deg}(L_P)\leq D$ and $\norm{L_P}_2\leq C$.
		\end{enumerate}
	\end{lemma}
    \begin{proof}
        For each $I\subseteq [n]$, apply Lemma~\ref{lem:approx_formula_noise_op_fancy_I} to
        get $f'_I$ as there such that
        $\norm{\mathrm{T}_{\mathcal{P},I} f-f'_I}\leq \xi/2$.
        Define $f'(x) = \Expect{I\subseteq_{\eps}[n]}{f'_I(x)}$.
        First, by the triangle inequality we have that
        \begin{equation}\label{eq:approx_formula_first}
        \norm{\mathrm{T}_{\mathcal{P}, 1-\eps}f - f'}_2
        \leq
        \Expect{I\subseteq_{\eps}[n]}
        {\norm{\mathrm{T}_{\mathcal{P}, I}f - f'_I}_2}
        \leq \frac{\xi}{2}.
        \end{equation}
        Next, write $f'_I(x) = \sum\limits_{P\in \spn(\mathcal{P})}P(x)\cdot L_{I,P}(x)$
        as in Lemma~\ref{lem:approx_formula_noise_op_fancy_I}.
        We get that
        \[
        f'(x)
        =\sum\limits_{P\in \spn(\mathcal{P})}P(x)
        \Expect{I\subseteq_{\eps}[n]}{L_{I,P}(x)}.
        \]
        Defining $L_P(x) = \Expect{I\subseteq_{\eps}[n]}{L_{I,P}(x)}$, we have by Lemma~\ref{lem:avg_of_juntas_low_deg}
        that $W^{\geq D}[L_P]\leq (1-\eps)^D\cdot C$,
        where $C = O_{m,r,\alpha,\xi}(1)$ is from Lemma~\ref{lem:approx_formula_noise_op_fancy_I}.
        Taking $\tilde{L}_{P} = L_P^{\leq D}$
        and $f''(x) = \sum\limits_{P\in \spn(\mathcal{P})}P(x) \tilde{L}_P(x)$, we get that
        \[
        \norm{f'-f''}_2
        \leq \sum\limits_{P\in \spn(\mathcal{P})}
        \norm{L_P-\tilde{L}_P}_2
        \leq O_{m,r}((1-\eps)^{D/2}\sqrt{C})
        \leq \frac{\xi}{2}
        \]
        for sufficiently large $D$. Combining this with~\eqref{eq:approx_formula_first} we
        get that $\norm{\mathrm{T}_{\mathcal{P},1-\eps}f-f''}_2\leq \xi$. As $\norm{\tilde{L}_P}_2\leq \norm{L_P}_2\leq C$, the statement follows.
    \end{proof}

	\section{The Regularity Lemma}\label{sec:reg_lemmas}
	In this section we prove a regularity lemma for the semi-norm $\norm{\circ }_{\nu,\alpha}$.
	\subsection{A Basic Version of the Regularity Lemma}
	We begin with a basic version of our regularity lemma that already contains all of the essential ideas.
	\begin{lemma}\label{lem:regularity_basic}
		For all $\alpha>0$, $m\in\mathbb{N}$ there is an Abelian group $G$ of size $O_m(1)$ such that for all $\xi>0$ there exist $\eps>0$ and $r\in\mathbb{N}$ such that the following holds.
		Let $\Sigma$ be an alphabet of size at most $m$, let $\nu$ be a distribution over $\Sigma$ in which the probability of
		each atom is at least $\alpha$ and let $f\colon \Sigma^n\to\mathbb{C}$ be a $1$-bounded function. Then there exist
        $\sigma\colon \Sigma\to G$ and a collection $\mathcal{P}\subseteq\mathcal{P}(\Sigma, G, \sigma)$ of size at most $r$ such that
		\[
		\norm{f - \mathrm{T}_{\mathcal{P}, 1-\eps} f}_{\nu,\alpha}\leq\xi.
		\]
	\end{lemma}
	The rest of this section is devoted to the proof of Lemma~\ref{lem:regularity_basic}.
	
	We look at all $\mu\in M_{\nu,\alpha}$ and apply Theorem~\ref{thm:csp4} on them with the parameters $\alpha,m$ as here
	and $\eps$ there being $\xi/100$. As $M_{\nu,\alpha}$ has size $O_{m}(1)$, we may
	take $\delta$ to be the minimum over all the $\delta$'s in Theorem~\ref{thm:csp4} and $d$ to be the maximum of all the $d$'s,
	take the Abelian group $G$ to be the product
	of all the Abelian group in the applications in Theorem~\ref{thm:csp4} and $\sigma$ to be the concatenation of all of $\sigma$'s
	from the applications of Theorem~\ref{thm:csp4}. We fix these parameters, set $R = \lceil\frac{100}{\delta^2}\rceil$ and further use the following parameters:
	\begin{equation}\label{eq9}
		0<d_R^{-1}\ll\zeta_R\ll \eps_R\ll\ldots\ll d_1^{-1}\ll\zeta_1\ll \eps_1\ll d_0^{-1}, \delta\ll\alpha,\xi,m^{-1}\leq 1.
	\end{equation}
	We now proceed with the following iterative process. Starting with $g_0 = \E[f]$, $f_0 = f-g_0$, $\mathcal{P}_0=\emptyset$ and $i=0$, if $f_i$ has correlation at least
	$\delta$ with a function of the form $P_i\cdot L_i$ for $P_i = \chi\circ \sigma$ and $L_i\colon \Sigma^n\to\mathbb{C}$ of $2$-norm
	$1$, we do the following:
	\begin{enumerate}
		\item
		Define $\mathcal{P}_{i+1} = \mathcal{P}_{i}\cup \{P_{i}\}$, $g_{i+1} = \mathrm{T}_{\mathcal{P}_{i+1}, 1-\eps_{i+1}} f$
		and $f_{i+1} = f - g_{i+1}$.
		\item Increase $i$ by $1$ and repeat.
	\end{enumerate}
	The following claim, stating that $\norm{g_{i+1}}_2^2\geq i\delta^2/2$, is the key in the proof of Lemma~\ref{lem:regularity_basic}. Once
	it is established, the proof is concluded quickly.
	\begin{claim}\label{claim:key_of_regularity}
		For all $i\leq R$ we have that $\norm{g_{i+1}}_2^2\geq i\delta^2/10$.
	\end{claim}
	\begin{subproof}
		Fix $i$, and for ease of notation denote $\mathrm{T}_{i} = \mathrm{T}_{\mathcal{P}_i, 1-\eps_i}$. For $j=1,\ldots,i+1$ we have that $\card{\inner{f_j}{P_j L_j}}\geq \delta$, and we  pick a complex number
		$\theta_j$ of absolute value $1$ such that $\inner{f_j}{\theta_j P_j L_j}\geq \delta$. Noting that
		$f_i = f-g_i = (I-\mathrm{T}_j) f$ and using the fact that $I-\mathrm{T}_j$ is self-adjoint, we conclude that
		\begin{equation}\label{eq10}
			\inner{f}{\theta_j (I-\mathrm{T}_{j}) P_j L_j}\geq \delta
			\qquad\qquad\forall j=1,\ldots,i+1.
		\end{equation}
		We next inspect that quantity $(\rom{1}) = \inner{\mathrm{T}_{i+2} f}{\sum\limits_{j=1}^{i+1}\theta_j (I-\mathrm{T}_{j}) P_j L_j}$,
		and prove an upper bound as well as a lower bound on it.
		\paragraph{The Lower Bound:} as $\mathrm{T}_{i+2}$ is self adjoint, we conclude that
		$(\rom{1}) = \inner{f}{\sum\limits_{j=1}^{i+1}\theta_j \mathrm{T}_{i+2}(I-\mathrm{T}_{j}) P_j L_j}$, and we next
		argue that $\mathrm{T}_{i+2}(I-\mathrm{T}_{j}) P_j L_j$ is very close to $(I-\mathrm{T}_{j}) P_j L_j$ for each $j$.
		Clearly,
		\[
		\mathrm{T}_{i+2}(I-\mathrm{T}_{j}) P_j L_j
		=\mathrm{T}_{i+2}P_j L_j - \mathrm{T}_{i+2}\mathrm{T}_{j} P_j L_j.
		\]
		For the first term on the right-hand side, note that $P_j\in \spn(\mathcal{P}_{i+2})$ and so
		$\mathrm{T}_{i+2}P_j L_j = P_j\mathrm{T}_{i+2}L_j$.
        Indeed, the value of
        $(\mathrm{T}_{i+2}P_j L_j)(x)$ is the average of the value of $P_j(y)L_j(y)$ over $y\sim \mathrm{T}_{i+2}x$, and by definition of the noise operator, for all such $y$'s we have that $P_j(y) = P_j(x)$.
        It follows from Claim~\ref{claim:pres_low_deg_functions} that
		\[
		\norm{\mathrm{T}_{i+2}P_j L_j - P_j L_j}_2^2
        =\norm{\mathrm{T}_{i+2} L_j - L_j}_2^2 =
		\norm{(I-\mathrm{T}_{i+2})L_j}_2^2\lll_{d,m,i,\alpha} \eps_{i+2}^{1/3},
		\]
		and so $\norm{\mathrm{T}_{i+2}P_j L_j - P_j L_j}_2\lll_{d,m,i,\alpha} \eps_{i+2}^{1/6}$.
		For the second term, namely for $\mathrm{T}_{i+2}\mathrm{T}_{j} P_j L_j$, first apply Lemma~\ref{lem:approx_formula_noise_op_fancy}
		to get that $\norm{\mathrm{T}_{j} P_j L_j - h_j}_2\leq \zeta_j$ where $h_j$ is a function of the form
		$\sum\limits_{P'\in {\sf spn}_{\mathbb{N}}(\mathcal{P}_j)} P' L_{P'}$ where $L_{P'}$ has degree at most $d_j$ and $\norm{L_{P'}}_2 = O_{m,\zeta_j,\alpha,j,\eps_j}(1)$.
		Combining with Claim~\ref{claim:pres_low_deg_functions} again it follows that
		\[
		\norm{\mathrm{T}_{i+2}\mathrm{T}_{j} P_j L_j - \mathrm{T}_{j} P_j L_j}_2
		\leq 2\zeta_j + \norm{\mathrm{T}_{i+2}h_j - h_j}_2
		=2\zeta_j + \norm{(I-\mathrm{T}_{i+2})h_j}_2
		\leq 2\zeta_j + O_{m,\zeta_j,\alpha,j,\eps_j,d_j}(\eps_{i+2}^{1/6}).
		\]

		Concluding, we get that $\norm{\mathrm{T}_{i+2}(I-\mathrm{T}_{j}) P_j L_j - (I-\mathrm{T}_{j}) P_j L_j}\leq 2\zeta_j+\eps_{i+2}^{1/7}$,
		and plugging into $(\rom{1})$ gives
		\[
		\card{(\rom{1})}
		\geq
		\inner{f}{\sum\limits_{j=1}^{i+1}\theta_j (I-\mathrm{T}_{j}) P_j L_j} -2\sum\limits_{j=1}^{i+1}\zeta_j- 2(i+1)\eps_{i+2}^{1/6}
		\geq (i+1)\delta         -2\sum\limits_{j=1}^{i+1}\zeta_j-2(i+1)\eps_{i+2}^{1/7},
		\]
		where in the last inequality we used~\eqref{eq10}. Thus, $(\rom{1})\geq 0.99(i+1)\delta$.
		
		\paragraph{The Upper Bound:} by Cauchy-Schwarz we have that
		\[
		\card{(\rom{1})}\leq \norm{\mathrm{T}_{i+2} f}_2\norm{\sum\limits_{j=1}^{i+1}\theta_j (I-\mathrm{T}_{j}) P_j L_j}_2,
		\]
		and we upper bound the second norm. Taking a square and expanding, we have that
		\begin{equation}\label{eq11}
			\norm{\sum\limits_{j=1}^{i+1}\theta_j (I-\mathrm{T}_{j}) P_j L_j}_2^2
			\leq\sum\limits_{j}\norm{(I-\mathrm{T}_{j}) P_j L_j}_2^2
			+2\sum\limits_{j'<j}\card{\inner{(I-\mathrm{T}_{j}) P_j L_j}{(I-\mathrm{T}_{j'}) P_{j'} L_{j'}}}.
		\end{equation}
		We bound the first sum on the right-hand side by the trivial bound of $2(i+1)$ (as each one of the norms individually is at most $2$),
		and next we show that the off-diagonal terms are negligible. Fix $j'<j$ and inspect the corresponding summand. Then by self-adjointness
		and Cauchy-Schwarz
		\[
		\card{\inner{(I-\mathrm{T}_{j}) P_j L_j}{(I-\mathrm{T}_{j'}) P_{j'} L_{j'}}}
		=\card{\inner{P_j L_j}{(I-\mathrm{T}_{j})(I-\mathrm{T}_{j'}) P_{j'} L_{j'}}}
		\leq
		\norm{(I-\mathrm{T}_{j})(I-\mathrm{T}_{j'}) P_{j'} L_{j'}}_2.
		\]
		Using the same argument as in the upper bound section, we have that
		$\norm{\mathrm{T}_{j}(I-\mathrm{T}_{j'}) P_{j'} L_{j'} - (I-\mathrm{T}_{j'}) P_{j'} L_{j'}}\leq \eps_{j}^{1/6}$,
		implying that $\norm{(I-\mathrm{T}_{j})(I-\mathrm{T}_{j'}) P_{j'} L_{j'}}_2\leq \eps_j^{1/6}$. Plugging this into~\eqref{eq11}
		gives that
		\[
		\eqref{eq11}\leq 2(i+1) + (i+1)^2\eps_1^{1/6}
		\leq 2(i+2),
		\]
		and so $(\rom{1})\leq \norm{\mathrm{T}_{i+2} f}_2\sqrt{2(i+2)}$.
		
		\paragraph{Combining the Upper and Lower Bounds:} combining the upper and lower bounds for $(\rom{1})$, we conclude that
		$\norm{\mathrm{T}_{i+2} f}^2_2\cdot 2(i+2)\geq 0.99^2(i+1)^2\delta^2$, and simplifying we get that
		$\norm{\mathrm{T}_{i+2} f}^2_2\geq i\delta^2/10$.
	\end{subproof}
	
	Using Claim~\ref{claim:key_of_regularity}, the process we designed terminates within at most $R$ steps at step $i\leq R$, at which
	point we have that $f_i$ has correlation at most $\delta$ with any function of the form $P\cdot L$. Applying Theorem~\ref{thm:csp4}
	and the fact that $f_i$ is $2$-bounded we conclude that $\norm{f_i}_{\nu,\alpha}\leq 2\cdot \frac{\xi}{100}\leq \xi$, concluding the
	proof.
	\qed
	
	\subsection{A Version of the Regularity Lemma Allowing Noise Modification}
	In this section, we state and prove a variant of Lemma~\ref{lem:regularity_basic} in which we have flexibility
	in picking the noise parameter $\eps$. For the standard noise operator $\mathrm{T}_{1-\eps}$,
    if we have that $\mathrm{T}_{1-\eps} f$ is close to $f$ in $2$-norm, then $\mathrm{T}_{1-\eps'} f$ would be close to $f$ for any $0\leq \eps'\leq \eps$. With this analogy,
	one expects that the first item in Lemma~\ref{lem:regularity_basic} to hold not only for $\eps$ but rather for any $0<\eps'<\eps$. While we do not know
	how to make such an argument go through, we show a different argument that essentially achieves this extra flexibility.
	
	In the formulation of the lemma below, we will use a decay function $w\colon [0,1]\to[0,1]$, meaning a function
	$w$ satisfying that $w(\eps)\leq \eps$. The reader should have in mind $w$ of the form
	$w(\eps) = 2^{-1/\eps}$ or a more rapidly decaying function.
	\begin{lemma}\label{lem:regularity_fancy}
       For all $\alpha>0$, $m\in\mathbb{N}$ there is an Abelian group $G$ of size $O_m(1)$ such that for all $\xi>0$ there exist $\eps>0$ and $r\in\mathbb{N}$ such that for any decay function $w\colon[0,1]\to[0,1]$ there is $\eps_0$ such that the following holds.
		Let $\Sigma$ be an alphabet of size at most $m$, let $\nu$ be a distribution over $\Sigma$ in which the probability of
		each atom is at least $\alpha$ and let $f\colon \Sigma^n\to\mathbb{C}$ be a $1$-bounded function. Then there exist
        $\sigma\colon \Sigma\to G$, a collection $\mathcal{P}\subseteq\mathcal{P}(\Sigma, G, \sigma)$ of size at most $r$ and $\eps\geq \eps_0$ such that
		\[
		\norm{f - \mathrm{T}_{\mathcal{P}, 1-\eps'} f}_{\nu,\alpha}\leq\xi
		\]
		for all $w(\eps)\leq \eps'\leq \eps$.
	\end{lemma}
	\begin{proof}
		We run the same process as in the proof of Lemma~\ref{lem:regularity_basic}. Once the process terminates,
		say at step $i$, we know that $(I-\mathrm{T}_{\mathcal{P}_i, 1-\eps})f$ has correlations at most $\delta$
		with any function of the form $P\cdot L$ for $\eps = \eps_i$. If this holds for all $\eps\in (w(\eps_i), \eps_i)$,
		we are done as then by Theorem~\ref{thm:csp4} it follows that $\norm{(I-\mathrm{T}_{\mathcal{P}_i, 1-\eps})f}_{\nu,\alpha}\leq \xi$
		for all $\eps\in (w(\eps_i), \eps_i)$. Otherwise, we may find $\eps_i'\in (w(\eps_i), \eps_i)$ such that
		$(I-\mathrm{T}_{\mathcal{P}_i, 1-\eps_i'})f$ has correlation at least $\delta$ with some function $P_{i+1} L_{i+1}$, and
		we may continue the argument therein by modifying $\eps_i$ to be $\eps_i'$ ($\eps_i'$ still satisfies the same requirements
		from the parameters as $\eps_i$). Carrying out the same analysis as in Lemma~\ref{lem:regularity_basic}, we conclude
		that this modified process also terminates within $R$ steps, and the proof is concluded.
	\end{proof}

 The analysis of our dictatorship test in Section~\ref{section:dict_test_actual} requires an analogous statement in the setting of multiple functions. More precisely, in this setting we have a collection of functions $f_{j, \ell}$ for $j=1,\ldots, L$ and $\ell= 1, \ldots, q$, where for each fixed $j$ the domain of $f_{j,\ell}$ is the same. We show that we can find collection of product functions and noise rates for all functions simultaneously, and furthermore that for each $j$, the collections for $f_{j,1}, f_{j,2}, \ldots, f_{j,q}$ are identical.
    \begin{lemma}\label{lem:regularity_fancy_mult}
       For all $\alpha>0$, $L,q, m\in\mathbb{N}$ there is an Abelian group $G$ of size $O_m(1)$ such that for all $\xi>0$ there exist $\eps>0$ and $r\in\mathbb{N}$ such that for any decay function $w\colon[0,1]\to[0,1]$ there is $\eps_0$ such that the following holds.
		Let $\Sigma_1,\ldots,\Sigma_L$ be alphabets of size at most $m$, let $\nu_j$ be a distribution over $\Sigma_j$ in which the probability of
		each atom is at least $\alpha$ and let $f_{j, \ell}\colon (\Sigma_j^n,\nu_j^n)\to\mathbb{C}$ be a $1$-bounded function. Then there exist
        $\sigma_j\colon \Sigma_j\to G$, a collection $\mathcal{P}_j\subseteq\mathcal{P}(\Sigma_j, G, \sigma_j)$ of size at most $r$ and $\eps\geq \eps_0$ such that
		\[
		\norm{f_{j, \ell} - \mathrm{T}_{\mathcal{P}_j, 1-\eps'} f_{j, \ell}}_{\nu_j,\alpha}\leq\xi
		\]
		for all $j=1,\ldots,L$, $\ell = 1, \ldots, q$ and $w(\eps)\leq \eps'\leq \eps$.
	\end{lemma}
	\begin{proof}
    We run the process from Lemma~\ref{lem:regularity_basic} on all of the functions $f_{j,\ell}$ together. The main difference is that for each fixed $j$, the collection of product functions $\mathcal{P}_j$ is the same for all $\ell$, and that the parameter $\eps'$ works for all of the $f_{j,\ell}$ simultaneously.

    More precisely, start with $g_{j,\ell}^{(0)} = \E[f_{j,\ell}]$, $f_{j,\ell}^{(0)} = f_{j,\ell}-g_{j,\ell}^{(0)}$, $\mathcal{P}_j^{(0)}=\emptyset$ for every $j$ and $\ell$, and then $i_1,\ldots,i_L=i_{{\sf tot}}=0$. If $f_{j,\ell}$ has correlation at least
	$\delta$ with a function of the form $P^{(j)}_{i_j}\cdot L^{(j)}_{i_j}$ for $P^{(j)}_{i_j} = \chi\circ \sigma_j$ and $L^{(j)}_{i_j}\colon \Sigma_j^n\to\mathbb{C}$ of $2$-norm
	$1$, we do the following:
	\begin{enumerate}
		\item
		Define $\mathcal{P}_j^{(i_j+1)} = \mathcal{P}_j^{(i_j)}\cup \{P^{(j)}_{i_j}\}$ and
        $i_{\sf tot} = 1+\sum\limits_{j'}i_{j'}$.
        \item Define $g^{(i_j+1)}_{j, \ell} = \mathrm{T}_{\mathcal{P}_j^{(i_j+1)}, 1-\eps_{i_{\sf tot}}} f_{j,\ell}$
		and $f_{j,\ell}^{(i_j+1)} = f_{j,\ell} - g^{(i_j+1)}_{j, \ell}$ for every $\ell$.
        \item For $j'\neq j$ redefine $g^{(i_{j'})}_{j', \ell} = \mathrm{T}_{\mathcal{P}_{j'}^{(i_{j'})}, 1-\eps_{i_{\sf tot}}} f_{j',\ell}$
		and $f_{j',\ell}^{(i_{j'})} = f_{j',\ell} - g^{(i_{j'})}_{j', \ell}$ for every $\ell$.
		\item Increase $i_j$ by $1$ and repeat.
	\end{enumerate}
    Similar to the proof of Claim~\ref{claim:key_of_regularity}, we can show that $\sum_{\ell} \norm{g^{i_j+1}_{j, \ell}}_2^2\geq i_j\delta^2/2$, and we continue the argument as in the proof of Lemma~\ref{lem:regularity_fancy} to establish that the process stops after performing at most $qR$ steps for each $j$. When the process stops, we have $\norm{f_{j,\ell}^{(i_j)}}_{\nu_j, \alpha} \leq \xi$ for all $j$ and $\ell$.
	\end{proof}

	\subsection{A Multi-Variate Regularity Lemma with High Rank}
	 We finish this section off by presenting a strengthening of Lemma~\ref{lem:regularity_fancy} in the setting of multi-variate distributions and multiple functions. We will use the following convenient notation:
    \begin{definition}
        For a product function $P = P_1\cdots P_n\colon\Sigma^n\to\mathbb{C}$ and a subset $I\subseteq [n]$, we denote by $P|_{I}\colon \Sigma^n\to\mathbb{C}$ the function $P|_{I}(x) = \prod\limits_{i\in I}P_i(x_i)$.
    \end{definition}

     \begin{lemma}\label{lem:regularity_fancy_highrank}
         For all $\alpha>0$, $m\in\mathbb{N}$ and $\xi>0$ there exists $r\in\mathbb{N}$ such that
		for any decay function $w\colon [0,1]\to [0,1]$ there is $\eps_0>0$ for which the following holds.
        Let $\Sigma,\Gamma,\Phi$ be finite alphabets of size at most $m$ and let $\mu$ be a distribution over $\Sigma\times\Gamma\times \Phi$ in which the probability of each atom is at least $\alpha$. Then there exists a group $G$ whose size depends only on $m$, such that for all $1$-bounded functions $f\colon \Sigma^n\to\mathbb{C}$,
        $g\colon \Gamma^n\to\mathbb{C}$
        and $h\colon \Phi^n\to\mathbb{C}$
        there are collections
        $\mathcal{P}$,
        $\mathcal{Q}$,
        $\mathcal{R}$
       of embedding functions with $\card{\mathcal{P}}, \card{\mathcal{Q}}, \card{\mathcal{R}}\leq r$  and $\eps\geq \eps_0$ such that:
        \begin{enumerate}
            \item
            $\norm{f-\mathrm{T}_{\mathcal{P},1-\eps'}f}_{\mu_x,\alpha}\leq\xi$,
             $\norm{g-\mathrm{T}_{\mathcal{Q},1-\eps'}g}_{\mu_y,\alpha}\leq\xi$,
             and  $\norm{h-\mathrm{T}_{\mathcal{R},1-\eps'}h}_{\mu_z,\alpha}\leq\xi$
            for all $\eps'\in (w(\eps),\eps)$.
            \item Considering the collection $\mathcal{D} = \sett{D_P,D_Q,D_R}{P\in \mathcal{P},Q\in\mathcal{Q},R\in\mathcal{R}}$ of product functions over $\supp(\mu)$ where $D_P(x,y,z) = P(x)$,
            $D_Q(x,y,z) = Q(y)$
            and $D_{R}(x,y,z) = R(z)$, we have that ${\sf rk}(\mathcal{D})\geq 1/w(\eps)$.
        \end{enumerate}
     \end{lemma}
     \begin{proof}
     Fix $\alpha,m,\xi$, take $r$ from Lemma~\ref{lem:regularity_fancy} and
     let $w'\colon [0,1]\to[0,1]$ be a decay function to be determined (depending only on $\alpha,m,\xi,r$ and the decay function $w$). Applying Lemma~\ref{lem:regularity_fancy} on $f,g,h$ gives collections $\mathcal{P},\mathcal{Q},\mathcal{R}$ and $\eps_0$ satisfying the first item with $w'$ in place of $w$, and we next explain how to get the other two items. Let $\eps\geq \eps_0$.

    Note that the size of $S = \spn(\mathcal{D})$  is $A = O_{m,\alpha,r,\xi}(1)$.
    For each $i=1,\ldots,6A-1$
    define $w_{i+1}(a) = \frac{1}{2A}w_i(w(a))$
    and $w_1(a) = \frac{1}{2A}w(a)$.
    Then we pick the decay function $w' = w_{6A}$.
    For $i=1,\ldots,2A-1$ consider the interval $I_i=[w_{3i}(\eps)^{-1},w_{3i+3}(\eps)^{-1})$. Then
    we may find $i$ such that
    $\Delta_{{\sf symbolic}}(D,1)\not\in I_i$
    for all $D\in S$. Define
    $S_{{\sf bad}}\subseteq S$ as the collection of $D\in S$ such that
    $\Delta_{{\sf symbolic}}(D,1)\leq w_{3i}(\eps)^{-1}$, and let
    \[
    J = \sett{i\in [n]}{\exists D = D_1\cdots D_n\in S_{{\sf bad}}, D_i\not\equiv 1}.
    \]
    Then $\card{J}\leq A\cdot w_{3i}(\eps)^{-1}$. Define
    the collections
    \[
    \mathcal{P}' = \sett{P|_{\bar{J}}}{P\in\mathcal{P}},
    ~~~
    \mathcal{Q}' = \sett{Q|_{\bar{J}}}{Q\in\mathcal{Q}},
    ~~~
    \mathcal{R}' = \sett{R|_{\bar{J}}}{R\in\mathcal{R}},
    ~~~
    \mathcal{D}'= \sett{D|_{\bar{J}}}{D\in\mathcal{D}}.
    \]
    We note that each one of $\mathcal{P}',\mathcal{Q}',\mathcal{R}',\mathcal{D}'$
    has rank
    \[
    w_{3i+3}(\eps)^{-1} - \card{J}\geq \frac{1}{2}w_{3i+3}(\eps)^{-1}\geq \frac{1}{w(w_{3i+2}(\eps))},
    \]
    so the second item holds. The next claim establishes the first item, thereby finishing the proof:
    \begin{claim}\label{claim:remains_small_norm}
        For $\eps'\in (w_{3i+2}(\eps),w_{3i+1}(\eps))$ we have
    $\norm{f - \mathrm{T}_{\mathcal{P}', 1-\eps'} f}_{\mu_x, \alpha}\leq 2\eta$,
    $\norm{g - \mathrm{T}_{\mathcal{Q}', 1-\eps'} g}_{\mu_y, \alpha}\leq 2\eta$ and
    $\norm{h - \mathrm{T}_{\mathcal{R}', 1-\eps'} h}_{\mu_z, \alpha}\leq 2\eta$.
    \end{claim}
    \begin{subproof}
    We prove the inequality for $f$, and the other two follow in the same way.
    Recall that by choice of $\mathcal{P}$ we have
    $\norm{f - \mathrm{T}_{\mathcal{P}, 1-\eps'} f}_{\mu_x, \alpha}\leq \eta$, and we next argue that
    $\norm{\mathrm{T}_{\mathcal{P}, 1-\eps'} f-
    \mathrm{T}_{\mathcal{P}', 1-\eps'} f}_2\leq 2\card{\mathcal{P}}
    \sqrt{\card{J}\eps'}$.
    Indeed, write $\mathcal{P} = \{P_1,\ldots,P_r\}$, and
    for $0\leq i\leq r$ define
    $\mathcal{P}(i) = \{P_1|_{\bar{J}},\ldots,P_{i}|_{\bar{J}}, P_{i+1},\ldots, P_r\}$.
    Then by the triangle inequality for each $i$, we have that
    \begin{align}\label{eq:get_virt_rank}
    \norm{\mathrm{T}_{\mathcal{P}(i), 1-\eps'} f-
    \mathrm{T}_{\mathcal{P}(i+1), 1-\eps'} f}_2
    &\leq
    \norm{\mathrm{T}_{\mathcal{P}(i), 1-\eps'} f-
    \mathrm{T}_{\mathcal{P}(i)\cup \{P_{i+1}|_{\bar{J}}\}, 1-\eps'} f}_2\notag\\
    &+\norm{\mathrm{T}_{\mathcal{P}(i)\cup\{P_{i+1}|_{\bar{J}}\}, 1-\eps'} f-
    \mathrm{T}_{\mathcal{P}(i+1), 1-\eps'} f}_2.
    \end{align}
    Note that the first expression on the right-hand side is at most $\sqrt{\eps'\card{J}}$
    by Fact~\ref{fact:noise_op_basic_prop2}. As for the second expression, we note that
    $\mathcal{P}(i)\cup \{P_{i+1}|_{\bar{J}}\} = \mathcal{P}(i+1)\cup \{P_{i+1}\}$, so it is also at most $\sqrt{\eps'\card{J}}$
    by Fact~\ref{fact:noise_op_basic_prop2}. Summing up~\eqref{eq:get_virt_rank} over $i$ and using the triangle inequality gives that
    $\norm{\mathrm{T}_{\mathcal{P}, 1-\eps'} f-
    \mathrm{T}_{\mathcal{P}', 1-\eps'} f}_2\leq 2\card{\mathcal{P}}
    \sqrt{\card{J}\eps'}\leq \eta$.
    Thus, $\norm{\mathrm{T}_{\mathcal{P}, 1-\eps'} f-
    \mathrm{T}_{\mathcal{P}', 1-\eps'} f}_{\mu_x,\alpha}\leq \eta$, and the result follows from the triangle inequality.
    \end{subproof}
     \end{proof}

	\section{Moving to the Mixed Space}\label{sec:mixed_space}
	In this section we prove Theorem~\ref{thm:plain_mixed_inv},
    and we first show how to associate with a given product space $(\Sigma^n,\nu^{n})$ a mixed space $((\mathbb{R}^{\card{\Sigma}-1})^{n}\times \Sigma^{n}, \gamma^{n}\times\nu^n)$, where $\gamma$ is a jointed distribution over $\card{\Sigma}-1$ correlated Gaussian random variables.
	We then show a multi-variate
    version of this association.

	\subsection{Univariate Decoupling}\label{sec:decoupled_fn}
	Suppose we have a $1$-bounded function $f\colon (\Sigma^n,\nu^n)\to\mathbb{C}$
    and we found a collection $\mathcal{P}$ as in Lemma~\ref{lem:regularity_fancy_highrank}.
    In particular:
	\begin{enumerate}
		\item $\card{\mathcal{P}} = O_{m,\alpha,\eta}(1)$.
		\item $\norm{f - \mathrm{T}_{\mathcal{P}, 1-\eps} f}_{\nu, \alpha}\leq \eta$ where $\eps$ is arbitrarily small compared to
		$\eta$.
		\item ${\sf rk}(\mathcal{P})\geq M$ where $M\gg 1/\eta,1/\eps, 1/\alpha, m$.
	\end{enumerate}
	For simplicity of notation we denote $f' = \mathrm{T}_{\mathcal{P}, 1-\eps} f$.
	By Lemma~\ref{lem:approx_formula_noise_op_fancy} we may find a function $\tilde{f}$ of the form
	\[
	\tilde{f}(x) = \sum\limits_{P\in\spn(\mathcal{P})} P(x) \cdot L_P(x)
	\]
	such that
    $\norm{f' - \tilde{f}}_2\leq \zeta$, ${\sf deg}(L_P)\leq d = O_{m,\alpha,\eta,\zeta,\eps}(1)$ and $\norm{L_P}_2\leq O_{m,\alpha,\eta,\zeta,\eps}(1)$. We define $\tilde{f}_{{\sf decoupled}}\colon (\Sigma^n\times \Sigma^n, \nu^{n}\times \nu^{n})\to\mathbb{C}$ by
	\[
	\tilde{f}_{{\sf decoupled}}(x,x') = \sum\limits_{P\in\spn(\mathcal{P})} P(x) \cdot L_P(x').
	\]
	In words, we take two independent copies of the input of $\tilde{f}$, plug one copy of it to the embedding
	functions and another copy to the low-degree functions.
	We wish to prove a relation between the functions $\tilde{f}_{{\sf decoupled}}$ and $\tilde{f}$.
	Intuitively, the distribution of
	$\tilde{f}_{{\sf decoupled}}(x,x')$ when $x,x'$ are sampled from $\nu^{n}$ independently
	is close to the distribution of $\tilde{f}(X)$ when $X\sim \nu^{n}$. To see that, inspect the formula of $\tilde{f}(X)$, and note that the values of the low-degree functions $L_P$ are
	mostly determined after we expose $1-\delta/d$ of the coordinates of $X$. On the other hand, since $\mathcal{P}$ has a high
	rank, conditioning on exposing these coordinates of $X$, the distribution of $(P(X))_{P\in \spn(\mathcal{P})}$ is still mostly the same. Thus, the values of $(L_P(X))_{P\in \spn(\mathcal{P})}$ and $(P(X))_{P\in \spn(\mathcal{P})}$
	are almost independent of each other, hence it is close to what happens in $\tilde{f}_{{\sf decoupled}}$.
	
	To prove this formally we start with a few auxiliary facts. The first of which asserts that given a high rank collection of product functions $\mathcal{P} = \{P_1,\ldots,P_r\}$, we have that for each $a_i\in {\sf Image}(P_i)$, the probability over $x\sim\nu^{\otimes n}$ that $P_i(x) = a_i$ for all $i$ is either essentially $0$, or else it is bounded away from $0$.
	\begin{fact}\label{fact:high_rank_large_sets}
		For all $m,r\in\mathbb{N}$, for sufficiently large $M\in\mathbb{N}$, if
		$\mathcal{P} = \{P_1,\ldots,P_r\colon \Sigma^n\to\mathbb{C}\} \subseteq \mathcal{P}(\Sigma,G, \sigma)$ has $\card{G},\card{\Sigma} \leq m$ and ${\sf rk}(\mathcal{P})\geq M$,
        and $\nu$ is a distribution whose support is $\Sigma$ and in which the probability of each atom is at least $\alpha$, then for all $a_i\in {\sf Image}(P_i)$, one of the following holds:
        \begin{enumerate}
            \item $\Prob{x\sim\nu^{n}}{P_i(x) = a_i~\forall i} \leq 2^{-\Omega_{m,r,\alpha}(M)}$.
            \item
            $\Prob{x\sim\nu^{n}}{P_i(x) = a_i~\forall i} \geq \Omega_{m,r}(1)$.
        \end{enumerate}
	\end{fact}
	\begin{proof}
		We write
		$1_{P_i(x) = a_i} = \prod\limits_{b_i\in {\sf Image}(P_i), b_i\neq a_i}\frac{b_i-P_i(x)}{b_i-a_i}$, so that
		the left-hand side is equal to
		\[
		C\Expect{x}{\prod\limits_{i}\prod\limits_{b_i\in {\sf Image}(P_i)\setminus\set{a_i}}(b_i-P_i(x))},
		\]
		where $C=\prod\limits_{i}\prod\limits_{b_i\in {\sf Image}(P_i)\setminus\set{a_i}}\frac{1}{b_i-a_i}$.
		To compute the expectation, we expand it out and get that
		\[
		\Expect{x}{\prod\limits_{i}\prod\limits_{b_i\in {\sf Image}(P_i)\setminus\set{a_i}}(b_i-P_i(x))}
		=\sum\limits_{\alpha_1,\ldots,\alpha_r}C(\alpha_1,\ldots,\alpha_r)\Expect{x}{\prod\limits_{i=1}^{r}P_i(x)^{\alpha_i}},
		\]
		where each one of the coefficients $C(\alpha_1,\ldots,\alpha_r)$ is a signed sum of products of elements in the images of the functions $P_i$. Defining
		\[
		\mathcal{A} = \sett{(\alpha_1,\ldots,\alpha_r)}{\prod\limits_{i=1}^{r}P_i(x)^{\alpha_i}\not\equiv 1},
        \qquad
        C' = \sum\limits_{(\alpha_1,\ldots,\alpha_r)\not\in \mathcal{A}}C(\alpha_1,\ldots,\alpha_r),
		\]
        we get that
        $\Prob{x\sim\nu^{n}}{P_i(x) = a_i~\forall i} = C\cdot C' + C\sum\limits_{(\alpha_1,\ldots,\alpha_r)\not\in \mathcal{A}}C(\alpha_1,\ldots,\alpha_r)\Expect{x}{\prod\limits_{i=1}^{r}P_i(x)^{\alpha_i}}$. Note that by the triangle inequality,
        \[
        \card{C\sum\limits_{(\alpha_1,\ldots,\alpha_r)\not\in \mathcal{A}}C(\alpha_1,\ldots,\alpha_r)\Expect{x}{\prod\limits_{i=1}^{r}P_i(x)^{\alpha_i}}}
        \lll_{m,r}
        \max_{(\alpha_1,\ldots,\alpha_r)\not\in \mathcal{A}}\card{\Expect{x}{\prod\limits_{i=1}^{r}P_i(x)^{\alpha_i}}}
        \lll_{m,r}
        2^{-\Omega_{m,r,\alpha}(M)},
        \]
        where the last inequality is by Lemma~\ref{lemma:symbolic_to_decay} applied with $\xi=1$. As $M$ is sufficiently large, we conclude that
        \[
        \card{C\cdot C'}-2^{-\Omega_{m,r}(M)}
        \leq
        \Prob{x\sim\nu^{n}}{P_i(x) = a_i~\forall i}\leq \card{C\cdot C'}+2^{-\Omega_{m,r}(M)}.
        \]
        Because the image of each $P_i$ is a root of unity of order at most $\card{G}\leq m$, we conclude that $\card{C}\geq \Omega_{m,r}(1)$. As $C'$ is a signed sum of at most $r^r$ products, each product consisting of elements from images of the functions $P_i$, we get that either $C' = 0$ or else $\card{C'}\geq \Omega_{m,r}(1)$. In conclusion, it is either the case that $\card{C\cdot C'} = 0$, in which the first item holds, and otherwise $\card{C\cdot C'}\geq \Omega_{m,r}(1)$, in which case the second item holds.
	\end{proof}
	
	Next, suppose that $\mathcal{P} = \{P_1,\ldots,P_r\}$ product functions, and
	for each $i$ take $a_i\in {\sf Image}(P_i)$ and consider the set $S = \sett{x\in\Sigma^n}{P_i(x) = a_i~\forall i=1,\ldots,r}$. The next result shows
	that the function $\mathrm{T}_{1-\kappa} 1_S$ is close to a constant function.
	\begin{fact}\label{fact:high_rank_mixes}
		For all $r,m\in\mathbb{N}$, $\xi>0$ the following holds for sufficiently large $M\in\mathbb{N}$.
		Suppose $\mathcal{P} = \{P_1,\ldots,P_r\colon \Sigma^n\to\mathbb{C}\} \subseteq \mathcal{P}(\Sigma,G, \sigma)$ where $\card{G},\card{\Sigma} \leq m$ has ${\sf rk}(\mathcal{P})\geq M$ and
		let $a_i\in {\sf Image}(P_i)$ for $i=1,\ldots,r$. Then for
		\[
		S = \sett{x\in\Sigma^n}{P_i(x) = a_i~\forall i=1,\ldots,r}
		\]
		we have $\left\|\mathrm{T}_{1-\kappa} 1_S - \nu^n(S)\right\|_2\leq 2^{-\Omega_{m,\kappa,r}(M)}$.
	\end{fact}
	\begin{proof}
		Expanding $1_S$ as in the proof of Fact~\ref{fact:high_rank_large_sets}, we see that
		$1_S = C + \sum\limits_{P\in \spn(\mathcal{P}), P\neq 1} C(P) P$ where $C(P)$ are constants satisfying $\card{C(P)} = O_{m,r}(1)$.
		Taking expectation of both sides as in Fact~\ref{fact:high_rank_large_sets} gives $\nu^n(S) = C + 2^{-\Omega_{m,r}(M)}$,
        so $\card{\nu^n(S)-C}\leq 2^{-\Omega_{m,r}(M)}$.
		We get from the triangle inequality that
		\begin{align*}
			\left\|\mathrm{T}_{1-\kappa} 1_S - \nu^n(S)\right\|_2
            &\leq
            2^{-\Omega_{m,r}(M)}+
            \left\|\mathrm{T}_{1-\kappa} 1_S - C\right\|_2\\
			&\leq
			2^{-\Omega_{m,r}(M)}
			+O_{m,r}\left(\max_{P\in \spn(\mathcal{P}), P\neq 1} \norm{\mathrm{T}_{1-\kappa} P}_2\right)\\
			&\leq
			2^{-\Omega_{m,r}(M)}
			+O_{m,r}\left(\max_{P\in \spn(\mathcal{P}), P\neq 1} 2^{-\Omega_{m,r,\kappa}(M)}\right)\\
			&\leq
			2^{-\Omega_{m,r,\kappa}(M)},
		\end{align*}
		where we used Lemma~\ref{lemma:symbolic_to_decay}.
	\end{proof}
	
	We are now ready to prove the main technical ingredient in the relation between $\tilde{f}$ and $\tilde{f}_{{\sf decoupled}}$.
	\begin{lemma}\label{lem:coupling_decoupled}
		Let $r,m\in\mathbb{N}$, $\alpha,\kappa,\xi>0$, let  $M\in\mathbb{N}$ be sufficiently large
        and let 
        \[
        \mathcal{P} = \{P_1,\ldots,P_r\colon \Sigma^n\to\mathbb{C}\} \subseteq \mathcal{P}(\Sigma,G, \sigma)
        \]
        be a collection of product functions satisfying $\card{G},\card{\Sigma} \leq m$ and ${\sf rk}(\mathcal{P})\geq M$.
        Let $\nu$ be a distribution whose support is $\Sigma$ in which the probability of each atom is at least $\alpha$, and consider the joint distribution of $(x,x')$ and $X$
		sampled as:
        \begin{enumerate}
            \item Sample $X\sim \nu^{n}$.
            \item Independently sample $x\sim \mathrm{T}_{\mathcal{P},0} X$ and $x'\sim \mathrm{T}_{1-\kappa}X$.
        \end{enumerate}
        Then
		then distribution of $(x,x')$ is $2^{-\Omega_{r,m,\kappa}(M)}$ close to $\nu^{n}\times \nu^{n}$ in statistical distance.
	\end{lemma}
	\begin{proof}
		Fix $w,w'\in\Sigma^n$, and set $q_{w,w'} = \Prob{X, x, x'}{x = w, x' = w'}$,
		$q_{w'|w} = \cProb{X, x, x'}{x = w}{x' = w'}$. By Fact~\ref{fact:noise_op_basic_prop1}
		the marginal distribution of $x$ is $\nu^{n}$ and so $q_{w,w'} = \nu^{n}(w)q_{w'|w}$.
        For $a_i\in {\sf Image}(P_i)$
        define
		\[
		S_{a_1,\ldots,a_r} = \sett{u\in\Sigma^n}{P_i(u) = a_i~\forall i},
		\]
        and let $S_w = S_{P_1(w),\ldots,P_r(w)}$.
        By Fact~\ref{fact:high_rank_large_sets}, for each $\vec{a} = (a_1,\ldots,a_r)$ we either have that $\nu^n(S_{\vec{a}})\leq 2^{-\Omega_{m,r,\alpha}(M)}$ and otherwise $\nu^n(S_{\vec{a}})\geq \Omega_{m,r}(1)$, and we say $w$ is bad if
        $\nu^n(S_{w})\leq 2^{-\Omega_{m,r,\alpha}(M)}$.
        Note that
        \begin{equation}\label{eq:decoupling1}
        \sum\limits_{\substack{w,w'\\ w\text{ is bad}}}q_{w,w'}
        =\sum\limits_{w\text{ is bad}}\nu^n(w)
        =\sum\limits_{\substack{\vec{a}\\ \nu^{n}(S_{\vec{a}})\leq 2^{-\Omega_{m,r,\alpha}(M)}}}\nu^n(S_{\vec{a}})
        \leq 2^{-\Omega_{m,r,\alpha}(M)},
        \end{equation}
        where in the last inequality we used the fact that the total number of $\vec{a}$ is $O_{m,r}(1)$. Similarly,
        \begin{equation}\label{eq:decoupling2}
        \sum\limits_{\substack{w,w'\\ w\text{ is bad}}}\nu^{n}(w)\nu^{n}(w')
        =\sum\limits_{w\text{ is bad}}\nu^n(w)
        \leq 2^{-\Omega_{m,r,\alpha}(M)},
        \end{equation}
        We next analyze good $w$.
		Note that conditioned on $x = w$, the distribution of $X$ is $\nu$ conditioned on $S_w$:
		\[
		\cProb{X,x,x'}{x=w}{X=u} = \nu^n(u) \frac{1_{S_w}(u)}{\nu^n(S_w)}.
		\]
		Thus,
		\[
		q_{w'|w}
		= \sum\limits_{u} \nu^n(u) \frac{1_{S_w}(u)}{\nu^n(S_w)} \Prob{u'\sim \mathrm{T}_{1-\kappa} u}{u' = w'}
		= \sum\limits_{u} \nu^n(u) \frac{1_{S_w}(u)}{\nu^n(S_w)} \mathrm{T}_{1-\kappa} 1_{w'}(u),
		\]
		where $1_{w'}(v) = 1_{w' = v}$. Writing this last expression as an inner product, we get that
		\[
		q_{w'|w}
		= \inner{\frac{1_{S_w}}{\nu^n(S_w)}}{\mathrm{T}_{1-\kappa} 1_{w'}}
		= \inner{\mathrm{T}_{1-\kappa}\frac{1_{S_w}}{\nu^n(S_w)}}{1_{w'}}
		= \nu^n(w') \mathrm{T}_{1-\kappa}\frac{1_{S_w}}{\nu(S_w)}(w'),
		\]
		and so $q_{w,w'} = \nu^n(w)\nu^n(w')\mathrm{T}_{1-\kappa}\frac{1_{S_w}}{\nu^n(S_w)}(w')$. Computing, we get that
		\begin{align*}
			\sum\limits_{\substack{w,w'\\ w\text{ not bad}}}\card{q_{w,w'} - \nu^n(w)\nu^n(w')}
			&=
			\sum\limits_{\substack{w,w'\\ w\text{ not bad}}}\nu^n(w)\nu^n(w')\card{\mathrm{T}_{1-\kappa}\frac{1_{s_w}}{\nu^n(S_w)}(w') - 1}\\
            &\lll_{m,r}
            \sum\limits_{w,w'}\nu^n(w)\nu^n(w')\card{\mathrm{T}_{1-\kappa}1_{S_w}(w') - \nu^n(S_w)}\\
			&=
			\Expect{w\sim\nu^{n}}{\left\|\mathrm{T}_{1-\kappa}1_{S_w}-\nu^n(S_w)\right\|_1},
		\end{align*}
        which is at most $2^{-\Omega_{m,r,\alpha}(M)}$ by Fact~\ref{fact:high_rank_mixes}. Combining this with~\eqref{eq:decoupling1} and~\eqref{eq:decoupling2} gives the lemma.
	\end{proof}
	
	We can now state and prove the relation between the functions
	$\tilde{f}$ and $\tilde{f}_{{\sf decoupled}}$.
	\begin{lemma}\label{lem:decoupled_1}
		Let $m,r,d\in\mathbb{N}$, and $M\in\mathbb{N}$.
		Let $\mathcal{P} = \{P_1,\ldots,P_r\colon \Sigma^n\to\mathbb{C}\}$ be a collection product functions, and consider the joint distribution of $(x,x')$ and $X$
		from Lemma~\ref{lem:coupling_decoupled}. Then
		\[
		\Expect{x,x',X}{\card{\tilde{f}_{{\sf decoupled}}(x,x') - \tilde{f}(X)}^2}\lll_{d,m,r} \kappa.
		\]
	\end{lemma}
	\begin{proof}
		Writing $\tilde{f}(x) = \sum\limits_{P\in \spn(\mathcal{P})}{P(x) L_P(x)}$ where ${\sf deg}(L_P)\leq d$, we get that
		the left-hand side in the lemma is equal to
        (as $P(x)=P(X)$ for $x\sim \mathrm{T}_{\mathcal{P},0} X$):
		\begin{align*}
			\Expect{x,x',X}{\card{\sum\limits_{P\in \spn(\mathcal{P})}{P(x) L_P(x') - P(X) L_P(X)}}^2}
			=
			\Expect{x,x',X}{\card{\sum\limits_{P\in \spn(\mathcal{P})}P(X) (L_P(x') - L_P(X))}^2}.
		\end{align*}
		By Cauchy-Schwarz, we may upper bound the last quantity by
		\[
		\lll_{m,r}
		\sum\limits_{P\in \spn(\mathcal{P})}\Expect{x,x',X}{\card{L_P(x') - L_P(X)}^2}
		=\sum\limits_{P\in \spn(\mathcal{P})}\norm{L_P - \mathrm{T}_{1-\kappa}L_P}_2^2
		\lll_{m,r} d\kappa.
		\qedhere
		\]
	\end{proof}

    \subsection{Multi-Variate Decoupling}\label{sec:multi_variate_dec}

    In this section we generalize the results of the previous section to the multi-variate setting. In the multi-variate setting we have alphabets $\Sigma,\Gamma,\Phi$ each of size $m$ and a distribution $\mu$ over $\Sigma\times \Gamma\times \Phi$. The distribution $\mu$ will not be fully supported over $\Sigma\times \Gamma\times \Phi$, and when we discuss product functions over $\Sigma\times \Gamma\times \Phi$ we always mean that their domain is only ${\sf supp}(\mu)$. The following result is the multi-variate analog of Lemma~\ref{lem:coupling_decoupled}:
    \begin{lemma}\label{lem:coupling_decoupled_multi}
       Let $\Sigma,\Gamma,\Phi$ be alphabets of size at most $m$, let $\mu$ be a distribution over $\Sigma\times\Gamma\times\Phi$ in which the probability of each atom is at least $\alpha$, and let $\mathcal{D}\subseteq \mathcal{P}({\sf supp}(\mu), G, \sigma)$ be a collection of product functions with $\card{G}\leq m$ and ${\sf rk}(\mathcal{D})\geq M$.
       Consider the coupling
        $((x,y,z), (x',y',z'),(X,Y,Z))$ defined as follows:
        \begin{enumerate}
            \item Sample $(X,Y,Z)\sim \mu^{n}$.
            \item Sample
            $(x,y,z)\sim \mathrm{T}_{\mathcal{D},0} (X,Y,Z)$.
            \item Sample $(x',y',z')\sim \mathrm{T}_{1-\kappa}(X,Y,Z)$.
        \end{enumerate}
        Then
		then distribution of $((x',y',z'),(x,y,z))$ is $2^{-\Omega_{r,m,\kappa}(M)}$ close to $\mu^{n}\times \mu^{n}$ in statistical distance.
	\end{lemma}
    \begin{proof}
        Define the alphabet
        $\widetilde{\Sigma} = \supp(\mu)$, and consider the distribution $\mu$ over it and the collection of product functions $\mathcal{D}\subseteq \mathcal{P}(\widetilde{\Sigma},G,\sigma)$. In these notations, the distribution in the lemma can be viewed as sampling $W\sim \mu^{n}$, $w\sim \mathrm{T}_{\mathcal{D},0} W$ and $W'\sim \mathrm{T}_{1-\kappa} W$. This is precisely the distribution in Lemma~\ref{lem:coupling_decoupled}, so as ${\sf rk}(\mathcal{D})\geq M$, the statement follows.
    \end{proof}

    We next state a version of Lemma~\ref{lem:decoupled_1}
    for the coupling in Lemma~\ref{lem:coupling_decoupled_multi}. For that, given a functions $1$-bounded functions $f\colon(\Sigma^n,\mu_x^n)\to\mathbb{C}$,
    $g\colon(\Gamma^n,\mu_y^n)\to\mathbb{C}$
    and $h\colon(\Phi^n,\mu_z^n)\to\mathbb{C}$ we find collections of product functions $\mathcal{P},\mathcal{Q},\mathcal{R}$ as in Lemma~\ref{lem:regularity_fancy_highrank}, and then set up the functions $\tilde{f},\tilde{f}_{{\sf decoupled}}$ and analogously
    $\tilde{g},\tilde{g}_{{\sf decoupled}}$ and
    $\tilde{h},\tilde{h}_{{\sf decoupled}}$
    as in Section~\ref{sec:decoupled_fn}.
    \begin{lemma}\label{lem:decoupled_2}
		In the above setting, consider the joint distribution
		from Lemma~\ref{lem:coupling_decoupled_multi}. Then each one of the expectations
		\[
		\Expect{}{\card{\tilde{f}_{{\sf decoupled}}(x,x') - \tilde{f}(X)}^2},
        \Expect{}{\card{\tilde{g}_{{\sf decoupled}}(y,y') - \tilde{g}(Y)}^2},\Expect{}{\card{\tilde{h}_{{\sf decoupled}}(z,z') - \tilde{h}(Z)}^2}
		\]
        is $\lll_{d,m,r} \kappa$.
	\end{lemma}
	\begin{proof}
		Identical to the proof of Lemma~\ref{lem:decoupled_1}, as the marginal distribution on $(x, x', X)$ in the coupling from Lemma~\ref{lem:coupling_decoupled_multi} is identical to the distribution on $(x, x', X)$ in the coupling from Lemma~\ref{lem:coupling_decoupled}.
	\end{proof}

    One technical issue with the coupling is that it only guarantees closeness in statistical distance. To remedy that, we have the following result:

    \begin{lemma}\label{lem:decoupled_3}
		In the setting of Lemma~\ref{lem:coupling_decoupled_multi} there is a coupling $((X',Y',Z'),(X'',Y'',Z''))$, $(X,Y,Z)$
        between $\mu^n\times \mu^n$
        and $\mu^n$ such that
		\[
		\Expect{X,X',X''}{\card{\tilde{f}_{{\sf decoupled}}(X',X'') - \tilde{f}(X)}^2}
		\leq O_{m,r,d}(\kappa) + 2^{-\Omega_{m,r,d,\alpha,\eta,\kappa}(M)},
		\]
        and similarly for $g$ and $h$.
	\end{lemma}
	\begin{proof}
		Let $((x,y,z),(x',y',z')$ and $(X,Y,Z)$ be a coupling as in Lemma~\ref{lem:coupling_decoupled_multi}. As the statistical distance between
		$(x,x')$ and $\nu^n\times \nu^n$ is at most $2^{-\Omega_{r,m,\kappa}(M)}$, we may couple $((x,y,z),(x',y',z'))$ with $((X',Y',Z'),(X'',Y'',Z''))$
        distributed according to $\mu^{n}\times \mu^n$
		such that \[
        \Prob{}{(x,y,z,x',y',z')\neq (X',Y',Z',X'',Y'',Z'')}\leq 2^{-\Omega_{r,m,\kappa}(M)}.\]
        We prove that
        the coupling $((X',Y',Z'),(X'',Y'',Z''))$ and $(X,Y,Z)$ satisfies the lemma. Note that the left-hand side of the lemma is at most
		\[
		\lll
		\Expect{}{\card{\tilde{f}_{{\sf decoupled}}(x,x') - \tilde{f}(X)}^2}
		+
		\Expect{}{\card{\tilde{f}_{{\sf decoupled}}(x,x') - \tilde{f}_{{\sf decoupled}}(X',X'')}^2}.
		\]
		The first expectation above is at most $O_{r,m,d}(\kappa)$ by Lemma~\ref{lem:decoupled_2}. For the second expectation, let $E$ be the event that
		$(x,y,z,x',y',z')\neq (X',Y',Z',X'',Y'',Z'')$. Then by Cauchy-Schwarz
		\begin{align*}
			\Expect{}{\card{\tilde{f}_{{\sf decoupled}}(x,x') - \tilde{f}_{{\sf decoupled}}(X',X'')}^2 1_E}^2
			&\leq
			\Expect{}{\card{\tilde{f}_{{\sf decoupled}}(x,x') - \tilde{f}_{{\sf decoupled}}(X',X'')}^4}\Prob{}{E}\\
			&\hspace{-1ex}\lll
			\Expect{}{\card{\tilde{f}_{{\sf decoupled}}(x,x')}^4 +\card{\tilde{f}_{{\sf decoupled}}(X',X'')}^4}
			\Prob{}{E}.
		\end{align*}
		   We upper bound the first expectation above. By the triangle inequality,
		$\card{\tilde{f}_{{\sf decoupled}}(x,x')}\leq \sum\limits_{P}\card{L_P(x')}$, and hence by Holder's inequality
		and hypercontractivity, we get that
		\[
		\Expect{x,x'}{\card{\tilde{f}_{{\sf decoupled}}(x,x')}^4}
		\lll_{r}\sum\limits_{P}\norm{L_P}_4^4
		\lll_{r,d}\sum\limits_{P}\norm{L_P}_2^4
		\lll_{r,m,d,\alpha,\eta} 1.
		\]
		The expectation of $\card{\tilde{f}_{{\sf decoupled}}(X',X'')}^4$ can be bounded in the same way by $O_{r,m,d,\alpha,\eta}(1)$.
		As $\Prob{}{E}\leq 2^{-\Omega_{r,m,\kappa}(M)}$ we conclude that
		\[
		\Expect{x,x',X,X',X''}{\card{\tilde{f}_{{\sf decoupled}}(x,x') - \tilde{f}_{{\sf decoupled}}(X',X'')}^2}
		\lll_{r,m,d,\alpha,\eta} 2^{-\Omega_{r,m,\kappa}(M)}.
		\qedhere
		\]
	\end{proof}

	\subsection{The Mixed Invariance Principle}\label{sec:mixed_invariance}
	In this section, we prove the mixed invariance principle. The notion of shifted low-degree influences will be important for us:
    \begin{definition}\label{def:low_influence_mixed_invariance}
		We say that a function $\tilde{f}$
		of the form above has $\tau$-small
		shifted low-degree influences if
		for every $i\in [n]$ and every $P\in \spn(\mathcal{P})$ it holds that $I_i[L_P]\leq \tau$.
	\end{definition}
	
	Let $\Psi\colon \mathbb{C}^3\to\mathbb{C}$ be any smooth function such that $\Psi(a,b,c) = abc$ for $a,b,c$
	that have absolute value at most $1$, which additionally satisfies that
	\begin{equation}\label{eq:smooth_def}
	\card{\Psi(a,b,c) - \Psi(a',b',c')}\lll\sqrt{|a-a'|^2 + |b-b'|^2 + |c-c'|^2}
	\end{equation}
	for all complex numbers $a,b,c\in\mathbb{C}$.
	Denote $m_1 = \card{\Sigma}$, $m_2 = \card{\Gamma}$ and $m_3 = \card{\Phi}$.
	Given $1$-bounded functions
	$f\colon (\Sigma^n,\nu_x^n)\to\mathbb{C}$,
	$g\colon (\Gamma^n,\nu_y^n)\to\mathbb{C}$
	and
	$h\colon (\Phi^n,\nu_z^n)\to\mathbb{C}$,
    collections $\mathcal{P},\mathcal{Q},\mathcal{R}$ of product functions, we define $f',\tilde{f},\tilde{f}_{{\sf decoupled}}$ as in Section~\ref{sec:decoupled_fn}
    and similarly
    $g',\tilde{g},\tilde{g}_{{\sf decoupled}}$ and
    $h',\tilde{h},\tilde{h}_{{\sf decoupled}}$. Define the functions
	$F\colon \Sigma^n\times \mathbb{R}^{(m_1-1)n}\to\mathbb{C}$,
	$G\colon \Gamma^n\times \mathbb{R}^{(m_2-1)n}\to\mathbb{C}$,
	$H\colon \Phi^n\times \mathbb{R}^{(m_3-1)n}\to\mathbb{C}$
	by
	\begin{align*}
		&F(x,G_x) = {\sf trunc}(\tilde{f}_{{\sf decoupled}}(x,G_x)),
		\qquad
		G(y,G_y) = {\sf trunc}(\tilde{g}_{{\sf decoupled}}(y,G_y)),\\
		&\qquad\qquad\qquad\qquad
		H(z,G_z) = {\sf trunc}(\tilde{h}_{{\sf decoupled}}(z,G_z)),
	\end{align*}
	where ${\sf trunc}\colon \mathbb{C}\to\mathbb{C}$ is as in Section~\ref{sec:invariance_principle}.
    We are now ready to prove Theorem~\ref{thm:plain_mixed_inv}, formally stated below:
	\begin{theorem}\label{thm:basic_mixed_invariance}
        Let $m\in\mathbb{N}$, $\alpha>0$, let $\Sigma,\Gamma,\Phi$ be alphabets of size at most $m$, and let $\mu$ be a distribution over $\Sigma\times\Gamma\times\Phi$ with no $\mathbb{Z}$-embeddings in which the probability of each atom is at least $\alpha$. Then for every
		$\xi>0$ there exist $M$ and $\tau>0$
		such that the following holds.
		Let $f\colon \Sigma^n\to\mathbb{C}$,
		$g\colon \Gamma^n\to\mathbb{C}$
		and $h\colon \Phi^n\to\mathbb{C}$
		be $1$-bounded functions,
		consider collections $\mathcal{P},\mathcal{Q},\mathcal{R}$ and the functions $F, G, H$
		defined as above, and suppose that
        \begin{enumerate}
            \item $\tilde{f},\tilde{g},\tilde{h}$
		have $\tau$-small shifted low-degree
		influences.
        \item The collection $\mathcal{D}$ as defined in Lemma~\ref{lem:regularity_fancy_highrank} satisfies that ${\sf rk}(\mathcal{D})\geq M$.
        \end{enumerate}
		Then for a smooth $\Psi\colon\mathbb{C}^3\to
        \mathbb{C}$ satisfying~\eqref{eq:smooth_def} with bounded third order derivatives we have
        \begin{equation}\label{eq:main_goal_invriance}
			\card{
				\Expect{(x,y,z)\sim \mu^n}{\Psi(f(x),g(y),h(z))}
				-
				\Expect{
    \substack{(x,y,z)\sim \mu^n\\ (G_x,G_y,G_z)\sim \mathcal{G}^n}}{\Psi(F(x,G_x),G(y,G_y),H(z,G_z))}
			}\leq \xi.
		\end{equation}
	\end{theorem}
	\begin{proof}
		We take the parameters:
		\begin{equation}\label{eq:mixed_inv_params}
			0<\tau,
			M^{-1}
			\ll
			\kappa
			\ll
			d^{-1}
			\ll
			\eta,
			r^{-1}
			\ll
			m^{-1}, \alpha ,\xi<1.
		\end{equation}
        Here, $m,\alpha,\xi$ are as in the statement, $r$ is an upper bound on the size of the collections $\mathcal{P},\mathcal{Q},\mathcal{R}$. The parameter $\eta$ is the $L_2$-closeness parameter between $f',\tilde{f}$ and similarly $g',\tilde{g}$ and the $\norm{}_{\mu_x,\alpha}$-closeness parameter between $f,f'$ and similarly $g,g'$ and $h,h'$. The parameter $d$ is the degree of the functions $L_P,L_Q,L_R$, $\kappa$ is the parameter for the coupling, $M$ is a lower bound on the rank of $\mathcal{D}$ and $\tau$ is a bound on the influences. It can be seen that the parameters can be taken as~\eqref{eq:mixed_inv_params} using a suitable decay function in Lemma~\ref{lem:regularity_fancy_highrank}.

		First, as $\norm{f - f'}_{\mu_x, \alpha}\leq \eta$, $\norm{g - g'}_{\mu_y, \alpha}\leq \eta$ and $\norm{h - h'}_{\mu_z, \alpha}\leq \eta$, we get that
		\begin{equation}\label{eq18}
			\card{
				\Expect{(x,y,z)\sim \mu^n}{\Psi(f(x),g(y),h(z))}
				-
				\Expect{(x,y,z)\sim \mu^n}{\Psi(f'(x),g'(y),h'(z))}
			}\lll \eta.
		\end{equation}
		By the smoothness of $\Psi$ it follows that
		\begin{align}\label{eq19}
			&\card{
				\Expect{(x,y,z)\sim \mu^n}{\Psi(f'(x),g'(y),h'(z))}
				-
				\Expect{(x,y,z)\sim \mu^n}{\Psi(\tilde{f}(x),\tilde{g}(y),\tilde{h}(z))}
			}\notag\\
			&\qquad\qquad\lll
			\Expect{(x,y,z)\sim \mu^n}
			{
				\sqrt{
					\card{f'(x) - \tilde{f}(x)}^2
					+
					\card{g'(y) - \tilde{g}(y)}^2
					+
					\card{h'(z) - \tilde{h}(z)}^2
			}}\notag\\
			&\qquad\qquad\lll \eta.
		\end{align}
        In the last transition, we first use Cauchy-Schwarz,
        and then the fact that $f'$ and $\tilde{f}$ are $\eta$-close in $\ell_2$ distance to get that $\Expect{(x,y,z)\sim\mu^{n}}{\card{f'(x) - \tilde{f}(x)}^2}\leq\eta$.
        Similarly $g',\tilde{g}$
		and $h',\tilde{h}$.
        Using the coupling from Lemma~\ref{lem:decoupled_3} and the smoothness of $\Psi$ it follows that:
		\begin{align}\label{eq15}
			\Bigg|
				&\Expect{(X,Y,Z)}{\Psi(\tilde{f}(X),\tilde{g}(Y),\tilde{h}(Z))}\notag\\
				&-
				\Expect{\substack{(X,Y,Z)\\ (X',Y',Z')\\ (Z'',Y'',Z'')}}
				{\Psi(\tilde{f}_{{\sf decoupled}}(X',X''),\tilde{g}_{{\sf decoupled}}(Y',Y''),\tilde{h}_{{\sf decoupled}}(Z',Z''))}
			\Bigg|\notag\\
			&~~\lll
			\Expect{\substack{(X,Y,Z)\\ (X',Y',Z')\\ (X'',Y'',Z'')}}
			{
				\sqrt{
					\card{\Delta_1(X,X',X'')}^2
					+
					\card{\Delta_2(Y,Y',Y'')}^2
					+
					\card{\Delta_3(Z,Z',Z'')}^2
			}}.
		\end{align}
		Here, $\Delta_1(X,X',X'') = \tilde{f}(X) - \tilde{f}_{{\sf decoupled}}(X',X'')$ and $\Delta_2,\Delta_3$ are defined
		analogously for $g$ and $h$.
        By Cauchy-Schwarz and Lemma~\ref{lem:decoupled_3}, we get that $\eqref{eq15}\lll_{d,m,r}\kappa$, and so $\eqref{eq15}\leq\eta$.
        Denote
        \[
        (\rom{1}) = \Expect{\substack{ (x,y,z)\sim\mu^{n}\\ (x',y',z')\sim\mu^{n}}}
				{\Psi(\tilde{f}_{{\sf decoupled}}(x',x),\tilde{g}_{{\sf decoupled}}(y',y),\tilde{h}_{{\sf decoupled}}(z',z))},
        \]
        so that the second expectation on the left hand side of~\eqref{eq15} is equal to $(\rom{1})$.
        Fix $x,y,z$ in $(\rom{1})$, and note that the $2$-norm of $\tilde{f}_{{\sf decoupled}}(x',x)$ over the choice
		of $x'$ is at most $O_{m,r,\eta,d}(1)$ and the influences are at most $O(\tau)$.
		We apply Theorem~\ref{thm:invariance_principle} and average over $(x,y,z)\sim\mu^{n}$, to get that, provided that $\tau$ is small
		enough
		\begin{align}\label{eq16}
			\Big|
			(\rom{1})
			-
			\Expect{\substack{(x',y',z')\sim \mu^n\\ (G_x,G_y,G_z)\sim \mathcal{G}^n}}
			{\Psi(\tilde{f}_{{\sf decoupled}}(x',G_x),\tilde{g}_{{\sf decoupled}}(y',G_y),\tilde{h}_{{\sf decoupled}}(z',G_z))
			}\Big|
			\leq \eta.
		\end{align}
		Next, by the smoothness of $\Psi$ it follows that
		\begin{align}\label{eq17}
			\Big|
			&\Expect{\substack{(x',y',z')\sim \mu^n\\ (G_x,G_y,G_z)\sim \mathcal{G}^n}}
			{\Psi(\tilde{f}_{{\sf decoupled}}(x',G_x),\tilde{g}_{{\sf decoupled}}(y',G_y),\tilde{h}_{{\sf decoupled}}(z',G_z))
			}\notag\\
			&-
			\Expect{\substack{(x',y',z')\sim \mu^n\\ (G_x,G_y,G_z)\sim \mathcal{G}^n}}
			{\Psi(F(x',G_x),G(y',G_y),H(z',G_z))
			}\Big|
			\leq E,
		\end{align}
		where
		\[
		E = \Expect{\substack{(x',y',z')\sim \mu^n\\ (G_x,G_y,G_z)\sim \mathcal{G}^n}}
		{\sqrt{\truncerr(\tilde{f}_{{\sf decoupled}}(x',G_x)) + \truncerr(\tilde{g}_{{\sf decoupled}}(y',G_y)) + \truncerr(\tilde{h}_{{\sf decoupled}}(z',G_z))}},
		\]
		and we recall the function $\truncerr(a_1,\ldots,a_s) = \sqrt{\sum\limits_{i}\card{{\sf trunc}(a_i) - a_i}^2}$.
		\begin{claim}\label{claim:finish_mixed_inv}
			$E\lll \sqrt{\eta}$.
		\end{claim}
		\begin{subproof}
			By Cauchy-Schwarz, $E\leq \sqrt{E_1+E_2+E_3}$ where
			\begin{align*}
				&E_1 = \Expect{\substack{(x',y',z')\sim \mu^n\\ (G_x,G_y,G_z)\sim \mathcal{G}^n}}{\truncerr(\tilde{f}_{{\sf decoupled}}(x',G_x))},
				\qquad E_2 = \Expect{\substack{(x',y',z')\sim \mu^n\\ (G_x,G_y,G_z)\sim \mathcal{G}^n}}{\truncerr(\tilde{g}_{{\sf decoupled}}(y',G_y))},\\
				&\qquad\qquad\qquad\qquad E_3 = \Expect{\substack{(x',y',z')\sim \mu^n\\ (G_x,G_y,G_z)\sim \mathcal{G}^n}}{\truncerr(\tilde{h}_{{\sf decoupled}}(z',G_z))},
			\end{align*}
			and we upper bound each one of $E_1,E_2$ and $E_3$ separately. As the arguments are identical, we show it only for $E_1$.
			By Theorem~\ref{thm:invariance_principle}, provided that $\tau$ is small enough
			\[
			E_1
			\leq
			\Expect{\substack{(X',Y',Z')\sim \mu^n\\ (X'',Y'',Z'')\sim \mu^n}}{\truncerr(\tilde{f}_{{\sf decoupled}}(X',X''))}
			+\eta.
			\]
			By Fact~\ref{fact:trivial_lipshitz_pf} the function $\truncerr$ is $O(1)$-Lipschitz,
			and so
			\[
			\Expect{X',X''}{\truncerr(\tilde{f}_{{\sf decoupled}}(X',X''))}
			\lll
			\Expect{X,X',X''}{\truncerr(f'(X)) + \card{\tilde{f}_{{\sf decoupled}}(X',X'') - f'(X)}}.
			\]
			Note that as $f'$ is $1$-bounded, $\truncerr(f'(X)) = 0$. Also, note that
			$f'$ is $\eta$-close in $\ell_2$ distance to $\tilde{f}$, and
			so we get that
			\[
			E_1
			\leq
			\Expect{X,X',X''}{ \card{\tilde{f}_{{\sf decoupled}}(X',X'') - \tilde{f}(X)}}
			+O(\sqrt{\eta}).
			\]
			Applying Cauchy-Schwarz and Lemma~\ref{lem:decoupled_2} gives that $E_1\lll\sqrt{\eta}$, as desired.
		\end{subproof}
		Combining all of the inequalities~\eqref{eq18},~\eqref{eq19},~\eqref{eq15},~\eqref{eq16} and~\eqref{eq17}
		and Claim~\ref{claim:finish_mixed_inv}
		proves~\eqref{eq:main_goal_invriance}.
	\end{proof}

	\section{Proof of Theorem~\ref{thm:main}}
	\label{sec:CSP}
Following Raghavendra~\cite{Rag08} (see also~\cite{BKMcsp1}), a typical way to design a dictatorship for an optimization problem $\mathcal{P}$ is based on
considering its SDP relaxation and solving it. On satisfiable instances, an SDP solution gives to each constraint a local distribution on its satisfying assignments, which then can be used in constructing the queries of the dictatorship test. We note that even if the SDP value is $1$, there is no guarantee that a given local distribution is fully supported on $P^{-1}(1)$. In particular, this local distribution may not be pairwise connectedness and could admit $\mathbb{Z}$-embeddings even if the predicate $P$ does not. In this case we cannot apply the analytical lemma from~\cite{BKMcsp4}, hence we cannot analyze the soundness of the dictatorship test.
	
	We circumvent this issue by demonstrating the use of our mixed invariance principle on {\symm} predicates. These predicates have sufficient symmetry to guarantee that any SDP solution with value $1$ can be converted into another SDP solution with value $1$ in which the support of all local distributions is pairwise connectedness and has no $\mathbb{Z}$-embeddings.

	\paragraph{Notations:} fix a collection $\mathcal{P}$ of {\symm} predicates. An instance of Max-$\mathcal{P}$-CSP, $\inst = (\calV, \calC)$, consists of a variable set $\calV$ which take values from $\Sigma$, and a  distribution $\calC$ on the constraint set. We associate $\calV$ with the set $[N] = \{1,\cdots, N \}$ for $N = |\calV|$. Each constraint $C\in {\sf supp}(\calC)$ is over a tuple of $3$ variables, denoted by $\calV(C) = (s_1, s_2, s_3)$, and consists of a predicate $P_C: \Sigma^3 \rightarrow \{0,1\}$ from $\mathcal{P}$. An assignment $(x, y, z)$ to the tuple $\calV(C)$ satisfies the constraint $C$ iff $(x, y, z) \in P_C^{-1}(1)$. For ease of notation, we often use $C$ to refer to both the constraint and the underlying predicate $P_C$, and simply write $C(x, y, z) =1$ or $(x, y, z)\in C^{-1}(1)$ if $P_C(x, y, z) = 1$. Given a set $T$, the set of all the distributions on $T$ is denoted by $\simplex(T)$.

\subsection{The SDP Program}

	The basic semidefinite programming relaxation of an instance $\inst = (\calV, \calC)$ is given in Figure~\ref{fig:basicsdp1}.   The SDP formulation consists of vectors $\{\V b_{i,a}\}_{i\in \calV, a\in \Sigma}$, distributions $\{\mu_{C}\}_{C\in \supp(\calC)}$ over local assignments (i.e., on $\Sigma^{\calV(C)}$) and a unit vector $\V b_0$. We denote by $\val(\V V, \V \mu)$ the objective value of the solution $(\V V, \V \mu)$. For every $\eta>0$, the SDP can be solved up to an additive accuracy of $\eta$ in time ${\sf poly}(n, \log(1/\eta))$ (see, for instance~\cite{GrotschelLS12}). We will ignore this issue of approximation and assume that the SDP can be solved optimally in polynomial time, and in Remark~\ref{remark:SDP_approximate_solution} we explain the modifications necessary to accommodate for this.
	
	\begin{figure}[t]
		\fbox{
			
			\parbox{558pt}{

                          				\begin{align}
					\mbox{maximize}\quad &\E_{C\sim \calC} \E_{x\in \mu_{C}} [C(x)] \nonumber\\
					\mbox{subject to } \quad &  \langle \V b_{i,a}, \V b_{j,b}\rangle = \Pr_{x\sim \mu_{C}}[x_i = a, x_j = b] &\forall C\in \supp(\calC),\quad i,j\in \calV(C), \quad a,b\in \Sigma \label{eq:inner_p_SDP1}\\
					&\langle \V b_{i,a}, \V b_0 \rangle = \|\V b_{i,a}\|_2^2& \forall i\in \calV, a\in \Sigma\\
					& \|\V b_0 \|_2^2 = 1\\
					&\mu_{C} \in \simplex(\Sigma^{\calV(C)}) & \forall C\in\supp(\calC)
                                                        \end{align}

		}}
		\caption{Basic SDP relaxation of a Max-CSP instance $\inst = (\calV, \calC)$.}
		\label{fig:basicsdp1}
	\end{figure}

During the execution of our algorithm, we will modify the SDP by imposing additional conditions on the local distributions. A key feature of these modification will be that they preserve valid {\em integral solutions}:
 \begin{definition}
     Fix any assignment $\V \alpha$ to $\inst$. The vector assignment
     \begin{align*}
         b_{i, a} = \begin{cases}
             \V b_0 & a = \alpha|_i,\\
             0 & \mbox{otherwise},
         \end{cases}
     \end{align*}
     along with $\V \mu$ where for every $C\in \supp(\calC)$, $\mu_C(\alpha|_{\calV(C)}) = 1$ and $\mu_C(\V d) = 0$  for every $\V d\neq \alpha|_{\calV(C)}$, is called an integral assignment to the SDP. Such an assignment is called a valid integral solution if it is a feasible solution to the SDP.
 \end{definition}
\subsection{Setting up a System of Linear Equations}
\newcommand{\gesystem}{{\sf GE\ System}}
\newcommand{\gesolver}{{\sf GE\ Solver}}
Once the SDP relaxation is solved, we construct an initial system of linear equations over a certain Abelian group. In this section, we describe an algorithm for formulating this system of linear equations. For convenience, we call this system of linear equations {\gesystem} (i.e., the Gaussian Elimination System).

Fix an arbitrary SDP solution $(\V V, \V \mu)$ with value $1$. The solution induces {\em local distributions} $\mu_C$ over $\Sigma^{\calV(C)}$ where $C\in {\sf supp}(\calC)$. Towards showing the optimality of our algorithm, we require that for every  $C\in {\sf supp}(\calC)$, the support of $\mu_C$ is pairwise connected and has no $\mathbb{Z}$-embedding. This condition is easy to achieve if the predicates in the instance are {\symm} and we will show this towards the end. For the following discussion, we assume that the SDP solution satisfies this condition.

\subsubsection{Setting up the Variables associated to \texorpdfstring{$v\in \calV$}{v} in \texorpdfstring{{\gesystem}}{GESystem} }
In our {\gesystem}, there will be many variables associated with a given  $v\in \calV$ from the CSP instance $\inst$. Here, we describe a polynomial-time procedure that first constructs a matrix $M_v$ with $|\Sigma|$ rows associated with the variable $v$. The columns of the matrix are all the embedding functions associated with all the constraints $v$ is involved in.
In order to be consistent across different embeddings, we will need to work with embeddings that assign the identity element to a special element from $\Sigma$. It will be convenient to treat $\Sigma$ as $[q] =\{1, 2, \ldots, q\}$ where $q=|\Sigma|$ and let $w^\star = 1$.

\begin{definition}
    An embedding $\sigma_1, \sigma_2, \sigma_3: \Sigma \rightarrow G$ of a subset $S\subseteq \Sigma\times \Sigma\times \Sigma$ is called a standard embedding if $\sigma_1(w^\star) = \sigma_2(w^\star) = \sigma_3(w^\star) = 0_G$ and there exists $g\in G$ such that for every $(x, y, z)\in S$, $\sigma_1(x)+\sigma_2(y)+\sigma_3(z) = g$. We will denote such embeddings by $((\sigma_1, \sigma_2, \sigma_3), g)$.
\end{definition}

An embedding $(\sigma_1, \sigma_2, \sigma_3)$ as in Definition~\ref{def:lin_embed1} can be converted into a standard embedding 
by replacing $\sigma_i$ by the map $\sigma_i-\sigma_i(w^{\star})$ for each $i$. 

\skipi

\noindent An informal description of the algorithm (presented in Algorithm~\ref{alg:generating_matrix_variables}) for computing the matrix $M_v$ is as follows. Fix a variable $v\in \calV$. The matrix $M_v$ is generated in the following three steps.
\begin{itemize} \setlength{\itemindent}{1.5em}
    \item [Step 1.] For every constraint $C$ such that $v\in \calV(C)$, and every embedding $(\sigma_1, \sigma_2, \sigma_3, g)$ of $\supp(\mu_{C})$ into an Abelian group $G$, we add columns corresponding to the evaluations of characters of $G$ on $(\sigma_j(x))_{x\in \Sigma}$ where $j$ is the index of $v$ in $C$.
    \item [Step 2.] Next, for every subset $S$ of columns in the matrix $M_v$ generated in Step 1, we add a column, which is a pointwise multiplication of columns $S$. If we treat each complex number in the polar form as $e^{{\bf i}\theta}$, then in this step, we add all the linear combinations of columns when viewed as vectors of exponents.
    \item [Step 3.] Finally, we add rows that are pointwise multiplication of the subsets of rows of $M_v$ generated in Step 2. Similarly to Step 2 we add all the linear combinations of the exponents of the rows.
\end{itemize}
This completes the informal description of the algorithm for computing $M_v$.

In the formal algorithm below, we also keep track of various other objects. First, we keep track of the group elements that are in the images of the embeddings. This is stored in the variable $\{r^v_\ell\}_{\ell\in\Sigma}$, where $r^v_{\ell}$ is the tuple of group elements that $\ell$ is mapped to under the various embeddings of $v$.
Second, we define the group $G^v_{\sf master}$, which is the product of all the groups arising in the step $1$ above.
In a sense, this is a large group which encompasses all of the Abelian embeddings $v$ may participate in. We denote by $H_v^\star$ a subgroup of $G^v_{{\sf master}}$ generated by $\{ r^v_\ell\}_{\ell\in \Sigma}$. In a sense $H_v^\star$ is the sub-group generated by feasible values of $v$ in the system of equations.

\begin{algorithm}[!ht]
\caption{Computing the matrix $M_v$}
\label{alg:generating_matrix_variables}

Start with $q \times 1$ matrix $M_v$ with $\vec{1}$ as a column\;
Instantiate $r^v_\ell$ to be an empty tuple for every $\ell\in [q]$. Set $H$ to be the trivial group\;
Suppose $v$ is involved in the constraints $C_1, C_2, \ldots, C_t$\;

\For{$i\gets 1$ \KwTo $t$}{
    Let $j\in \{1,2,3\}$ be the index of $v$ in $\calV(C_i)$\;
    
    \For{each embedding $((\sigma_1,\sigma_2,\sigma_3), g)$ of $\supp(\mu_{C_i})$ into an Abelian group $G$}{
        For every $\chi\in \hat{G}$ such that $\chi\not\equiv 1$, add the column
        $(\chi(\sigma_j(x)))_{x\in \Sigma}$ to $M_v$ if it is not already there\;
        
        $H \leftarrow H \times G$\;
        
        $r^v_\ell \leftarrow (r^v_\ell, \sigma_j(\ell))$ for every $\ell\in [q]$\tcp*{$r^v_\ell \in H$}
    }
}

Set $G^v_{\sf master} \leftarrow H$\;
Let $L'$ be the number of columns in $M_v$\;

\For(\tcp*[f]{Adding more columns}){every subset $S\subseteq [L']$}{
    \If{$\circ_{i\in S} M_v[.][i]$ is not present as a column in $M_v$}{
        Add a column $\circ_{i\in S} M_v[.][i]$ to $M_v$\;
    }
}

Let $H^\star_v$ be the group generated by $\{r^v_\ell\}_{\ell\in [q]}$\;
Let $i=q$\;

\For(\tcp*[f]{Adding more rows}){each $h\in H^\star_v$}{
    
    Suppose $h = \sum_{\ell\in S} r^v_\ell$ where $S$ is a multi-subset of $[q]$ and addition is a group operation in $G^v_{\sf master}$\;
    
    Add a row $\circ_{\ell\in S} M_v[\ell][.]$ to $M_v$ and set $i\leftarrow i+1$\;
    
    Set $r^v_i = h$\;
}

\end{algorithm}

In the algorithm below, in steps $4$ to $10$ we find all Abelian embeddings of local distributions $\mu_{C}$ that a variable $v\in\mathcal{V}$ participates in. Formally speaking, these are embeddings and characters over the group $G$, but it will be convenient to think of them as characters of the larger group $G^v_{{\sf master}}$.
This is done by identifying a character $\chi \in \hat{G}$ in step $7$ with the character $(1,\ldots, 1, \chi,1,\ldots, 1)\in \widehat{G^v_{\sf master}}$.

The following claim analyzes the run-time of Algorithm~\ref{alg:generating_matrix_variables}.
\begin{claim}
Algorithm~\ref{alg:generating_matrix_variables} runs in time ${\sf poly}(O_{m}(1),t)$, and the matrix $M_v$ output in the end has dimension $O_{m}(1)$.
\end{claim}
\begin{proof}
    Since there are $O_{m}(1)$
    many options for supports of $3$-ary distributions over alphabets of size $m$, and each one of which that has no $\mathbb{Z}$-embeddings has at most $O_m(1)$ Abelian embeddings, the number of different columns after steps 4-11 is at most $O_{m}(1)$, and the runtime of these lines is polynomial in $O_m(1)$ and $t$. It follows that the runtime of steps 14-18 is $O_{m}(1)$. Lastly, the group $H_v^{\star}$ in step 19 has size at most $O_{m}(1)$, so lines 20-24 also run in time $O_{m}(1)$.
\end{proof}
\begin{claim}
    \label{claim:Alg_Masterchar}
    The columns of the matrix $M_v$ contain evaluations of all the characters of $G^v_{\sf master}$ on a subgroup generated by $\{r^v_1, r^v_2, \ldots, r^v_q\}$.
\end{claim}
\begin{proof}
In steps 14-18, we add columns that correspond to every other character from $\widehat{G^v_{\sf master}}$ evaluated at $\{r^v_\ell\}_{\ell\in [q]}$. This follows as the character $\vec{\chi} :=(\chi_1, \chi_2, \ldots, \chi_t)\in \widehat{G^v_{\sf master}}$ is, by definition, $\vec{\chi} (\V a)= \prod_{i=1}^t \chi_i(a_i)$, and step 16 is precisely adding such columns.

Finally, in lines 20-24, we are adding rows that correspond to the evaluations of all the characters from $\widehat{G^v_{\sf master}}$ on every element of $G^v_{\sf master}$ generated by $\{r^v_1, r^v_2, \ldots, r^v_q\}$. This follows as $\vec{\chi} \in \widehat{G^v_{\sf master}}$ is a group homomorphism $\vec{\chi} : G^v_{\sf master}\rightarrow  \mathbb{C}$, i.e.,  $\vec{\chi}(\V a_1 + \V a_2) = \vec{\chi}(\V a_1)\cdot \vec{\chi}(\V a_2)$ for every $\V a_1, \V a_2 \in G^v_{\sf master}$.
\end{proof}

The variables in \gesystem~associated to $v\in \calV$ are $\{y_v^{\vec{\chi}}\}$ where $\vec{\chi}\in \widehat{G^v_{\sf master}}$. It would be convenient to think of the variable set $\{y_v^{\vec{\chi}}\}$ taking values that correspond to a row of the matrix $M_v$.

\subsubsection{Adding Equations to the \texorpdfstring{{\gesystem}}{GESystem}}\label{sec:add_eq}
Our next step is to set up a system of linear equations. Ideally, we would have liked to assign a group element from $\{r^v_1, r^v_2, \ldots, r^v_q\}\subseteq H_v^\star$ to a variable $v$ (since they directly correspond to elements from $\Sigma$). However, we do not know how to find such an assignment by an efficient algorithm, let alone by a set of linear equations. We thus relax this requirement and settle for an assignment from $H_v^\star$.
To do so we add to the linear system {\gesystem} the following two types of equations:
\begin{enumerate}
    \item {\bf Valid character constraints.} These equations enforce that the variables $\{y_v^{\vec{\chi}}\}_{\vec{\chi}\in \widehat{G^v_{\sf master}}}$ correspond to a row of the matrix $M_v$, and hence, the vector assignment corresponds to a group element from $H_v^\star$. Towards this, we add the following set of equations for every $v\in \calV$,
    \begin{itemize}
        \item $y_v^{\sf triv} = 1$, where ${\sf triv}$ corresponds to the first column of the matrix $M_v$.
        \item For every $\vec{\chi}, \vec{\chi'}, \vec{\chi''}$ such that $\vec{\chi''} = \vec{\chi}\cdot \vec{\chi'}$, we add the equation $y_v^{\vec{\chi''}} = y_v^{\vec{\chi}}\cdot y_v^{\vec{\chi'}}$.
        \item As $H^\star_v$ may be a proper subgroup of $G^v_{\sf master}$, there will be certain columns in $M_v$ that are constant. If the column corresponding to $\vec{\chi}$ is constant $c$, then add the equation $y_v^{\vec{\chi}} = c$.
    \end{itemize}
    \item {\bf Valid satisfying assignments constraints.} These equations ensure that any solution to {\gesystem} is consistent with the image of the support of local distributions under various embeddings into Abelian groups. Towards this, we add the following set of equations
    \begin{itemize}
        \item For every constraint $C\in {\sf supp}(\calC)$ such that $\calV(C) = (s_1, s_2, s_3)$, and for every embedding $((\sigma_1, \sigma_2, \sigma_3), g)$ of ${\sf supp}(\mu_{C})$ in $G$, we add the equation $y_{s_1}^{\vec{\chi}_1}\cdot y_{s_2}^{\vec{\chi}_2}\cdot y_{s_3}^{\vec{\chi}_3} = \chi(g)$ where $\chi\in \hat{G}$ and $\vec{\chi}_1, \vec{\chi}_2, \vec{\chi}_3$ are the respective columns added in step $7$ in  $M_{s_1}, M_{s_2}, M_{s_3}$, for the embedding $((\sigma_1, \sigma_2, \sigma_3), g)$ and the character $\chi$.
    \end{itemize}

\end{enumerate}

The equations are linear equations over the circle group $\mathbb{T} = \{ z\in \mathbb{C} \mid |z| = 1\}$ and hence a solution to the {\gesystem} can be found in polynomial time.
The justification for these constraints is given by the following observation, stating that a valid integral solution to the SDP gives a solution to the {\gesystem}.

    \begin{observation}
        \label{claim:justify_GESysytem}
        Suppose $\V x \in [q]^N$ is a satisfying assignment to the instance $\inst$ and suppose that this assignment survives in the SDP solution $(\V V, \V \mu)$, i.e., for every $C\in {\sf supp}(\calC)$, the local distribution $\mu_C$ assigns a non-zero probability mass to $\V x|_{\calV(C)}$. Then  assigning the $\V x_v^{th}$ row from $M_v$ to the variables $\{y_v^{\vec{\chi}}\}$ satisfies all the equations from the {\gesystem}.
    \end{observation}

    There is a natural way in which one can view a solution to the system of equations above as an element of $\prod_{v\in \calV} H_v^\star$. With this map, we denote the set of all the solutions $\mathcal{T}$ to the {\gesystem} as a subset of $\prod_{v\in \calV} H_v^\star$.

\subsection{Getting Consistent SDP and {\gesystem} Solutions}
\newcommand{\calT}{\mathcal{T}}
Jumping ahead, in the soundness analysis of our dictatorship test
we are interested in computing the expectation of the form
\[
\E_{C(s_1, s_2, s_3) \sim \calC}\left[\E_{(x, y, z)\sim \mu_C^R}[F_{s_1}(\sigma_1(x), x) \cdot F_{s_2}(\sigma_2(y), y) \cdot F_{s_3}(\sigma_3(z), z)]\right]
\]
for functions for functions $F_{s}: (H_s^\star)^R\times \Sigma^R\rightarrow \mathbb{C}$, where the first input goes to the product functions from the decomposition of $F_{s_i}$, and the second input goes to the low-degree functions. During the rounding procedure, we want to replace the local distribution $(\sigma_1(x),\sigma_2(y), \sigma_3(z))$ by some global distribution. In our analysis, the global distribution will be a random solution to our {\gesystem}.

One important technical condition that is required to show the optimality of the algorithm\footnote{Optimality with respect to the soundness of the dictatorship test.} is that when we make this switch, the embedding functions coming from the functions $F_{s_i}$ should satisfy the following property.
\begin{itemize}
    \item[$\blacklozenge$] For every constraint $C$ over $(s_1,s_2,s_3)$ and tuple of embedding functions $P$, $Q$ and $R$ coming from the decompositions of $F_{s_1}, F_{s_2}$ and $F_{s_3}$, respectively,  $PQR\equiv 1$ under the local distribution $\mu_C$ iff $PQR\equiv 1$ under the distribution from ${\gesystem}$.
\end{itemize}

To guarantee this condition we define below a condition called {\em consistent solutions}, which can be achieved via our hybrid algorithm and is sufficient to guarantee the above condition.

 For a variable $s\in \calV$, let $\sigma_{s}$ be the map $\sigma_{s} : \Sigma\rightarrow G^{s}_{\sf master}$ given by $\sigma_{s}(a) = r^{s}_a$, where $r^{s}_a$ is the group element from Algorithm~\ref{alg:generating_matrix_variables} associated with the variable $s$ and $a\in \Sigma$.
We now define the span of the local distribution $\mu_C$ under the embeddings $(\sigma_{s_1}, \sigma_{s_2}, \sigma_{s_3})$.
\begin{definition}
    For a constraint $C\in \calC$ with $\calV(C) = (s_1, s_2, s_3)$ and the embedding maps $\sigma_{s_i}$ to $G^{s_{i}}_{\sf master}$ for $1\leq i\leq 3$, define ${\sf span}({\sf supp}(\mu_C))$ to be the subgroup of $G^{s_{1}}_{\sf master}\times G^{s_{2}}_{\sf master}\times G^{s_{3}}_{\sf master}$ generated by $\{(\sigma_{s_1}(a), \sigma_{s_2}(b), \sigma_{s_3}(c)) \mid (a, b, c)\in {\sf supp}(\mu_C)\}$.
\end{definition}

We are now ready to define the notion of consistent solutions.
\begin{definition}[Consistent solutions]
    \label{def:tight_system}
   The solution to the semidefinite program $(\V V, \V \mu)$ is \emph{consistent} with a solution space $\mathcal{T}\subseteq \prod_{v\in \calV} H_v^\star$ to the {\gesystem} if the following condition holds:
  for every constraint $C\in \calC$ with $\calV(C) = (s_1, s_2, s_3)$, we have
  $${\sf span}({\sf supp}(\mu_C)) = \calT|_{C},$$
  where $\calT|_{C}:= \{ (a_{s_1}, a_{s_2}, a_{s_3}) \mid \V a\in \calT\}$.
\end{definition}

The utility of consistent solutions is given by the following lemma, which is crucial in the soundness analysis of the dictatorship test to achieve the condition $\blacklozenge$ mentioned above.
\begin{lemma}
\label{lemma:PQR_local_global}
    Fix a constraint $C\in {\sf supp}(\calC)$, let $\calV(C) = (s_1, s_2, s_3)$ and let $\sigma_{s_1}: \Sigma\rightarrow G^{s_1}_{\sf master}$, $\sigma_{s_2}: \Sigma\rightarrow G^{s_2}_{\sf master}$ and $\sigma_{s_3}: \Sigma\rightarrow G^{s_3}_{\sf master}$ be as above. For $\chi\in \widehat{G^{s_1}_{\sf master}}$, $\chi'\in \widehat{G^{s_2}_{\sf master}}$ and $\chi'' \in \widehat{G^{s_3}_{\sf master}}$, consider the functions $P(x) = \chi(\sigma_{s_1}(x))$, $Q(y) = \chi'(\sigma_{s_2}(y))$ and $R(z) = \chi''(\sigma_{s_3}(z))$ and the functions $P'(a) = \chi(a)$, $Q'(b) = \chi'(b)$ and $R'(c) = \chi''(c)$. If the SDP solution $(\V V, \V \mu)$ is consistent with the set of solutions $\calT$ to the {\gesystem}, then $PQR\equiv 1$ under ${\sf supp}(\mu_C)$ iff $P'Q'R'\equiv 1$ in the support of $\calT|_{C}$.
\end{lemma}
\begin{proof}
    To prove the lemma in one direction, suppose $PQR\equiv 1$ in the support of ${\sf supp}(\mu_C)$.
Take any $(\V a, \V b, \V c)\in \calT|_{C}$.
Since ${\sf span}({\sf supp}(\mu_C)) = \calT|_{C}$ there are $(x_p,y_p,z_p)\in {\sf supp}(\mu_C)$  such that
$(\V a, \V b, \V c) = \sum\limits_{p=1}^{L} (\sigma_{s_1}(x_p),\sigma_{s_2}(y_p), \sigma_{s_3}(z_p))$.
It follows from the multiplicativity of characters that
\[
    P'(\V a) Q'(\V b) R'(\V c) = \prod\limits_{p=1}^{L}P'(\sigma_{s_1}(x_p))\cdot
    \prod\limits_{p=1}^{L}Q'(\sigma_{s_2}(y_p))
    \prod\limits_{p=1}^{L}R'(\sigma_{s_3}(z_p))
     = \prod\limits_{p=1}^{L}P(x_p)\cdot
    Q(y_p)
    \cdot R(z_p),
\]
which is equal to $1$ as $PQR\equiv 1$
in the support of $\mu_C$ and $(x_p,y_p,z_p)\in{\sf supp}(\mu_C)$
for $p=1,\ldots,L$.

The other direction of the lemma is proved in the same way.
\end{proof}
We next explain how to achieve consistent solutions. Towards this end, we iteratively modify our SDP solution, as well as the {\gesystem}, as stated in Algorithm~\ref{alg:massage_sdp_ge}.
This algorithm consists of two sub-routines, which we present next:
\begin{enumerate}
    \item $\mathtt{ModifyGESystem}$: In this procedure, if ${\sf span}({\sf supp}(\mu_C)) \subsetneq \calT|_{C}$, we modify the {\gesystem} to reduce $\calT|_{C}$.
    \item $\mathtt{ModifySDP}$: In this procedure, if ${\sf span}({\sf supp}(\mu_C)) \supsetneq \calT|_{C}$, we modify the SDP formulation to reduce ${\sf span}({\sf supp}(\mu_C))$.
\end{enumerate}

\setcounter{AlgoLine}{0}
\begin{algorithm}[t]
\caption{\textsc{ModifyGESystem}}
\label{alg:modifyGEsystem}

\For{every $C\in {\sf supp}(\calC)$}{
    Let $\calV(C) = (s_1, s_2, s_3)$\;
    
    \If{${\sf span}({\sf supp}(\mu_C)) \subsetneq \calT|_{C}$}{
        For all $(\chi_1,\chi_2,\chi_3)\in 
        \widehat{G^{s_{1}}_{\sf master}}\times 
        \widehat{G^{s_{2}}_{\sf master}}\times 
        \widehat{G^{s_{3}}_{\sf master}}$ that evaluate to the constant $1$ on 
        ${\sf span}({\sf supp}(\mu_C))$, add the equation
        $y_{s_1}^{\chi_1}\cdot y_{s_2}^{\chi_2}\cdot y_{s_3}^{\chi_3} = 1$
        to the {\gesystem}\;
    }
}
\end{algorithm}

The following claim shows the correctness of the procedure $\mathtt{ModifyGESystem}$ given in Algorithm~\ref{alg:modifyGEsystem}.
\begin{claim}
\label{claim:pqr_eq1_modifyGE}
    Fix a constraint $C\in \calC$ with $\calV(C) = (s_1, s_2, s_3)$. If ${\sf span}({\sf supp}(\mu_C)) \subsetneq \calT|_{C}$, then the {\gesystem}, as modified in Algorithm~\ref{alg:modifyGEsystem}, satisfies $ {\sf span}({\sf supp}(\mu_C))\supseteq \calT|_{C} $.
\end{claim}
\begin{proof}
    Fix a constraint $C\in \supp(\calC)$ and let $\calV(C) = (s_1, s_2, s_3)$. Consider the collection of all the characters $\{(\chi^i_1, \chi^i_2, \chi^i_3)\}_i$ of the group $G^{s_{1}}_{\sf master}\times G^{s_{2}}_{\sf master}\times G^{s_{3}}_{\sf master}$ that evaluate to the constant $1$ on the subgroup ${\sf span}({\sf supp}(\mu_C))$. Using these characters, we are adding the following equations to the {\gesystem}:  $y_{s_1}^{\chi^i_1}\cdot y_{s_2}^{\chi^i_2}\cdot y_{s_3}^{\chi^i_3} = 1$ for all $i$. These equations enforce that $\{ (a_{s_1}, a_{s_2}, a_{s_3}) \mid \V a\in \calT\}$ is a subset of ${\sf span}({\sf supp}(\mu_C))$ as required.
\end{proof}

The above claim fixes the problem in one direction towards getting consistent solutions. We also require that ${\sf span}({\sf supp}(\mu_C)) \subseteq \calT|_{C}$. This is done in Algorithm~\ref{alg:modifySDP} below by simply excluding the assignments from a local distribution not implied by the {\gesystem}.

\setcounter{AlgoLine}{0}
\begin{algorithm}[t]
\caption{\textsc{ModifySDP}}
\label{alg:modifySDP}

\For{every $(a,b,c)\in \supp(\mu_C)$ such that 
$(r_a^{s_1}, r_b^{s_2}, r_c^{s_3})\notin \calT|_{C}$ where $\calV(C)=(s_1,s_2,s_3)$}{
    Augment SDP($\inst$) by adding the constraint $\mu_C(a,b,c)=0$\;
}

\end{algorithm}

The following claim shows the correctness of the procedure $\mathtt{ModifySDP}$ given in Algorithm~\ref{alg:modifySDP}.
\begin{claim}
\label{claim:pqr_eq1_modifySDP}
    Fix a constraint $C\in \calC$ with $\calV(C) = (s_1, s_2, s_3)$. Suppose  ${\sf span}({\sf supp}(\mu_C)) \supsetneq \calT|_{C}$. Then the SDP, as modified in Algorithm~\ref{alg:modifySDP}, satisfies $ {\sf span}({\sf supp}(\mu_C))\subseteq \calT|_{C}$.
\end{claim}
\begin{proof}
    Suppose  ${\sf span}({\sf supp}(\mu_C)) \supsetneq \calT|_{C}$. This implies that there exists $(a, b, c)\in \supp(\mu_C)$ such that $(r_a^{s_1}, r_b^{s_2}, r_c^{s_3})\notin \calT|_{C}$. Any solution to the modified SDP from Algorithm~\ref{alg:modifySDP} ensures that $\mu_{C}(a, b, c) = 0$ for such tuples. Therefore, any SDP solution to the modified SDP satisfies $ {\sf span}({\sf supp}(\mu_C))\subseteq \calT|_{C}$ for every constraint $C$.
\end{proof}

   Thus, we run the procedures $\mathtt{ModifyGESystem}$ and $\mathtt{ModifySDP}$ towards achieving consistent solutions as defined in Definition~\ref{def:tight_system}.
   Note that the procedures depend on the initial SDP solution  $(\V V, \V \mu)$, and as discussed, we will want to preserve all satisfying assignments in it.
   In more detail, we need to ensure that, after running the subroutines on a satisfiable instance, the final SDP and the {\gesystem} still have all the satisfying assignments preserved.
   We formally define this below for the SDP solution.

   \begin{definition} [Solution preserving all the integral solutions]
       \label{def:sdp_solution_preserves_all}
       We say an SDP solution $(\V V, \V \mu)$ preserves all the valid integral solutions, if for every satisfying assignment $\V \alpha \in \Sigma^n$ to the instance and any constraint $C\in \supp(\calC)$, $\alpha|_{\calV(C)}\in \supp(\mu_C)$.
   \end{definition}

    Algorithm~\ref{alg:massage_sdp_ge} stated below calls the two subroutines iteratively in a specific manner so that subsequent SDP formulations preserve all valid integral solutions.

\setcounter{AlgoLine}{0}
    \begin{algorithm}[!ht]
\DontPrintSemicolon
\KwIn{An instance $\inst(\calV, \calC)$ of Max-$\mathcal{P}$-CSP.}

Let SDP($\inst$) be the basic semidefinite program from Figure~\ref{fig:basicsdp1}\;
Let $(\V{V}, \V{\mu})$ be 
the SDP solution to SDP($\inst$) preserving all the integral solutions (we explain how to achieve this in Lemma~\ref{lemma:includes_all_sat_assignments})\;

\If{$\val(\V V, \V \mu) \neq 1$}{
    \textbf{Abort}\tcp*{$\inst$ is not satisfiable}
}

\If{there is $C\in {\sf supp}(\calC)$ such that ${\sf supp}(\mu_C)$ is either pairwise disconnected or has a $\mathbb{Z}$-embedding}{
    \textbf{Abort}\tcp*{$\inst$ is not satisfiable or $C$ is not {\symm}}
}
\Else{
    Set up a {\gesystem} using the SDP solution $(\V V, \V \mu)$ and solve it\;
    Let $\calT\subseteq \prod_{v=1}^N H_v^\star$ be the set of solutions to the {\gesystem}\;
    
    \If{$\exists\,(a,b,c)\in \supp(\mu_C)$ such that $(r_a^{s_1}, r_b^{s_2}, r_c^{s_3})\notin \calT|_{C}$ where $\calV(C)=(s_1,s_2,s_3)$}{
        Run $\mathtt{ModifySDP}$\;
    }
    \Else{
        Run $\mathtt{ModifyGESystem}$\;
    }
    
    \If{the SDP $\mathcal{S}$ or the {\gesystem} is modified}{
        Repeat from step 3 above\;
    }
    \Else{
        \Return{$\mathcal{S}$ and the {\gesystem}}\;
    }
}
\caption{Massaging the SDP solution and the {\gesystem} consistent}
\label{alg:massage_sdp_ge}
% \begin{algorithmic}[1]
% \State {\bf Input:} An instance $\inst(\calV, \calC)$ of Max-$\mathcal{P}$-CSP.

% \State Let SDP($\inst$) be the basic semidefinite program from Figure~\ref{fig:basicsdp1}.

% \State Let $(\V V, \V \mu)$ be the SDP solution to SDP($\inst$) preserving all the integral solutions (we explain how to achieve this in Lemma~\ref{lemma:includes_all_sat_assignments}).

% \If {$\val(\V V, \V \mu) \neq 1$}
% \State Abort.\Comment{$\inst$ is not satisfiable}
% \EndIf

% \If {there is $C\in {\sf supp}(\calC)$ such that ${\sf supp}(\mu_C)$ is either pairwise disconnected or has a $\mathbb{Z}$-embedding}
%     \State Abort. \Comment{$\inst$ is not satisfiable or $C$ is not {\symm}}
% \Else
%     \State Set up a {\gesystem} using the SDP solution $(\V V, \V \mu)$ and solve it.
%     \State Let $\calT\subseteq \prod_{v=1}^N H_v^\star$ be the set of solutions to the {\gesystem}.
%     \If {$\exists$ $(a, b, c)\in \supp(\mu_C)$ such that $(r_a^{s_1}, r_b^{s_2}, r_c^{s_3})\notin \calT|_{C}$ where $\calV(C) = (s_1, s_2, s_3)$}
%         \State Run $\mathtt{ModifySDP}$.
%     \Else
%         \State Run $\mathtt{ModifyGESystem}$.
%     \EndIf
%     \If {the SDP $\mathcal{S}$ or the {\gesystem} is modified}
%     \State {Repeat from step $3$ above}
%     \Else
%     \State Return $\mathcal{S}$ and the {\gesystem}
%     \EndIf
% \EndIf

% \end{algorithmic}
\end{algorithm}

The following section shows that step $3$ of the algorithm can be done in polynomial time. In Section~\ref{sec:getting_consistent_solutions}, we prove that the algorithm returns consistent solutions in polynomial time.

\subsubsection{Preserving All of the Integral Solutions}
To make sure that we do not exclude any satisfying assignment from subsequent SDP formulations during the execution of Algorithm~\ref{alg:massage_sdp_ge}, we will ensure that the SDP solution $ (\V V, \V \mu)$ in step $3$ preserves all integral solutions. Lemma~\ref{lemma:includes_all_sat_assignments} below states that such a solution $(\V V, \V \mu)$ can always be found in polynomial time. To state and prove the lemma, we need the following claim.

\begin{claim}
\label{claim:combineSDP}
Suppose $(\V V, \V \mu)$ and $(\V V', \V \mu')$ are two SDP solutions with value $1$. Then there is a SDP solution $(\V V'', \V \mu'')$ with value $1$ such that for every constraint $C\in \supp(\calC)$,  $\supp(\mu''_C) = \supp(\mu_C) \cup \supp(\mu'_C)$.
\end{claim}
\begin{proof}
The proof is based on a standard tensorization argument. Suppose that the SDP solution $(\V V, \V \mu)$ consists of vectors $\{\V b_{i,a}\}\cup \{ \V b_0\}$ and the SDP solution $(\V V', \V \mu')$ consists of vectors $\{\V b'_{i,a}\}\cup \{ \V b'_0\}$ from $\mathbb{R}^{qn+1}$. Let $\V e_1 = (1,0)$ and $\V e_2 = (0,1)$ be the vectors in $\mathbb{R}^2$. Consider the vectors
\[
\V b''_{i,a} = \frac{1}{\sqrt{2}} (\V e_1\otimes \V b_{i, a}) + \frac{1}{\sqrt{2}} (\V e_2 \otimes \V b'_{i,a}), \quad\quad\quad \V b''_0 = \frac{1}{\sqrt{2}} (\V e_1 \otimes \V b_0) + \frac{1}{\sqrt{2}} (\V e_2 \otimes \V b'_0).
\]
It can be easily verified that
\begin{align*}
    \langle \V b''_{i,a}, \V b''_0\rangle =\langle\frac{1}{\sqrt{2}} (\V e_1\otimes \V b_{i, a}) + \frac{1}{\sqrt{2}} (\V e_2\otimes \V b'_{i,a}), \V b''_0\rangle
    &= \frac{1}{\sqrt{2}} \langle  \V e_1\otimes \V b_{i, a}, \V b''_0\rangle + \frac{1}{\sqrt{2}}\langle  \V e_2\otimes \V b'_{i,a}, \V b''_0\rangle\\
    &= \frac{1}{2} \|\V b_{i, a}\|^2 + \frac{1}{2} \|\V b'_{i,a}\|^2\\
    & = \|\V b''_{i,a}\|^2,
\end{align*}
Similarly, $\| \V b''_0\|^2 = \frac{1}{2}\| \V b_0\|^2 + \frac{1}{2} \| \V b'_0\|^2 = 1$. Consider $\mu_C'' = \frac{1}{2} \mu_{C} + \frac{1}{2}\mu'_C$ for every $C\in \supp(\calC)$. We have
\begin{align*}
    \langle \V b''_{i,a}, \V b''_{j,b}\rangle &=\langle\frac{1}{\sqrt{2}} (\V e_1\otimes \V b_{i, a}) + \frac{1}{\sqrt{2}} (\V e_2\otimes \V b'_{i,a}), \frac{1}{\sqrt{2}} (\V e_1\otimes \V b_{j, b}) + \frac{1}{\sqrt{2}} (\V e_2\otimes \V b'_{j,b})\rangle  \\
    &= \frac{1}{2} \langle  \V b_{i, a}, \V b_{j, b}\rangle + \frac{1}{2}\langle  \V b'_{i,a}, \V b'_{j,b}\rangle\\
    & = \Pr_{x\sim \mu''_{C}}[x_i = a, x_j = b]
\end{align*}
Therefore, the vectors $\V V''= \{\V b''_{i,a}\}\cup \{ \V b''_0\}$ along with the local distributions $\V \mu''$ satisfy all the SDP constraints. Furthermore, we have $\supp(\mu''_C) = \supp(\mu_C) \cup \supp(\mu'_C)$ for every $C\in \supp(\calC)$. Finally, the SDP value of the solution $(\V V'', \V \mu'')$ is $1$ as $\mu_C''$ is supported on the set of satisfying assignments to $C$ for every $C\in \supp(\calC)$.
\end{proof}

\begin{lemma}
\label{lemma:includes_all_sat_assignments}
    Fix a satisfiable instance $\inst(\calV, \calC)$ of a Max-$\mathcal{P}$-CSP and a SDP formulation SDP($\inst$) such that every satisfiable assignment to $\inst$ is a valid integral solution to SDP($\inst$).

    There is a polynomial-time algorithm that returns an SDP solution $(\V V, \V \mu)$ with value $1$ such that for every satisfiable assignment $\V \alpha$ to the instance $\inst$, and for every $C\in \supp(\calC)$, $\mu_C(\alpha|_{\calV(C)}) >0$.
\end{lemma}
\begin{proof}
To prove this lemma, we will combine different SDP solutions using Claim~\ref{claim:combineSDP}. For every constraint $C$ and every $d = (a_1, a_2, a_3)\in \Sigma^3$ such that $C(d) = 1$, we augment the relaxation SDP($\inst$) with a constraint $\mu_C(d) = 1$; call this new relaxation SDP$_{C,d}$($\inst$). Let $(\V V_{C, d}, \V \mu_{C, d})$ be any arbitrary solution to SDP$_{C,d}$($\inst$). We then combine all the solutions $(\V V_{C, d}, \V \mu_{C, d})$ with value $1$ iteratively using Claim~\ref{claim:combineSDP} to get a solution $(\V V, \V \mu)$. Note that each application of Claim~\ref{claim:combineSDP} blows up the number of coordinates of the vectors by a multiplicative factor of $2$. However, since the SDP solution consists of $qn+1$ vectors, once the number of coordinates goes beyond $qn+1$, we can find a linear transformation that maps every vector to a vector with only the first $qn+1$ coordinates nonzero, while preserving all the inner products. We keep doing this dimension reduction after each application of Claim~\ref{claim:combineSDP}, and hence the overall process takes polynomial time. We now show that the final solution satisfies the required guarantee. To see this, fix any satisfiable assignment $\V \alpha$ to the instance $\inst$ and fix any $C\in \supp(\calC)$. As the solution $(\V V_{C, \alpha|_{\calV(C)}}, \V \mu_{C, \alpha|_{\calV(C)}})$ has $\mu_{C, \alpha|_{\calV(C)}}(\alpha|_{\calV(C)}) = 1$, we have $\mu_C(\alpha|_{\calV(C)}) >0$ using Claim~\ref{claim:combineSDP}. Furthermore, Claim~\ref{claim:combineSDP} guarantees that the value of the solution $(\V V, \V \mu)$ is $1$, proving this lemma.
\end{proof}

\begin{remark}  \label{remark:SDP_approximate_solution}
As mentioned at the beginning of this section, for every $\eta>0$, the SDP can be solved up to an additive accuracy of $\eta$ in time ${\sf poly}(n, \log(1/\eta))$ (see, for instance,~\cite{GrotschelLS12}). However,  Claim~\ref{claim:pqr_eq1_modifyGE} and Claim~\ref{claim:pqr_eq1_modifySDP} rely on solving the SDP program exactly as they rely on the exact support of $\mu_C$. One way to get around this issue is as follows. The solution constructed by Lemma~\ref{lemma:includes_all_sat_assignments} has the property that each valid entry where we want $\mu_C(\alpha|_{\calV(C)})>0$ can indeed be forced to have  $\mu_C(\alpha|_{\calV(C)})\geq 2^{-n^c}$ for a fixed constant $c$ depending on the number of constraints and the number of satisfying assignments to a constraint. This is because each iteration reduces the original weight in the distribution by $\frac{1}{2}$, and there will be at most $n^c$ iterations. Thus, if we solve the SDP with accuracy $2^{-n^C}$ for $C\gg c$ in time ${\sf poly}(n)$, then we can treat every weight $\mu_C(\alpha|_{\calV(C)}) \leq 2^{-n^C}$ as zero towards applying  Claim~\ref{claim:pqr_eq1_modifyGE} and Claim~\ref{claim:pqr_eq1_modifySDP}. In this way, we will end up preserving the local views of all the satisfying assignments in the SDP solution.
\end{remark}

The following claim shows that at each iteration in Algorithm~\ref{alg:massage_sdp_ge}, the modified SDP has all the satisfying assignments as valid integral solutions.
\begin{claim}
\label{claim:preserve_sat}
   If $\V \alpha\in \Sigma^N$ is a satisfying assignment to instance $\inst$, then it remains a valid integral solution to the SDP $\mathcal{S}$ that we get after running Algorithm~\ref{alg:massage_sdp_ge}.
\end{claim}
\begin{proof}
We will show that the procedure maintains the following two invariants: 1) every satisfying assignment from $\V \alpha\in [q]^N$ is also a valid satisfying assignment (after interpreting $\alpha_s$ with $r^s_{\alpha_s}$) to the {\gesystem} after the modification, and 2) every satisfying assignment from $\V \alpha\in [q]^N$ is a valid integral solution to SDP. These invariants are clearly satisfied at the beginning.

When we run $\mathtt{ModifySDP}$, the modified SDP has all the satisfying assignments to the instance $\inst$ as valid integral solutions. This follows as $\calT$ has the embedding of every satisfying assignment to start with, and we only excluded assignments from the local distributions that are not implied by the {\gesystem}.

Next, as the solution from step $3$ is guaranteed to keep $\alpha|_{\calV(C)}$ for every satisfying assignment $\V \alpha$ and $C$, every satisfying assignment $\V \alpha \in [q]^N$ (after interpreting $\alpha_s$ with $r^s_{\alpha_s}$) is also a valid satisfying assignment to the {\gesystem} after the modification. To see this, the elements $(r^{s_1}_{\alpha_{s_1}}, r^{s_2}_{\alpha_{s_2}}, r^{s_3}_{\alpha_{s_3}})$ satisfy all the additional equations of the from $y_{s_1}^{\chi_1}\cdot y_{s_2}^{\chi_2}\cdot y_{s_3}^{\chi_3} = 1$ added to the {\gesystem} related to the constraint $C$ for every $C$ with $\calV(C) = (s_1, s_2, s_3)$. Hence, $(r^s_{\alpha_s})_{s\in \calV}$ is in the support of $\calT$ even after modification of the {\gesystem}.
\end{proof}

\subsubsection{Getting Consistent Solutions to the SDP and the {\gesystem}}
\label{sec:getting_consistent_solutions}

We now have the following important lemma that shows that applying Algorithm~\ref{alg:massage_sdp_ge} on a CSP instance where every predicate is {\symm} returns an SDP solution and a system of linear equations that are consistent with respect to each other.

\begin{lemma}
\label{lemma:sdp_noZ_local_globalembed}
    Applying Algorithm~\ref{alg:massage_sdp_ge} on a satisfiable CSP instance $\inst(\calV, \calC)$ where every predicate $C\in \calC$ is {\symm} will return an SDP solution $(\V V^\star, \V \mu^\star)$ and a system of linear equations {\gesystem}$^\star$ that are consistent with each other.
    Furthermore, for every constraint $C\in \supp(\calC)$, the $\supp(\mu^\star_C)$ is pairwise connected and does not have a $(\mathbb{Z},+)$-embedding.
\end{lemma}
\begin{proof}
    We will rule out that the algorithm will never reach the Abort steps (steps $5$ and $8$). Lemma~\ref{lemma:alg_runtime} below shows that the algorithm terminates in polynomial time.  Using Claim~\ref{claim:pqr_eq1_modifyGE} and Claim~\ref{claim:pqr_eq1_modifySDP}, eventually, the algorithm will return an SDP solution and a {\gesystem} that are consistent with each other.

    Using Claim~\ref{claim:preserve_sat}, the algorithm will never reach step $5$. Next, we rule out that the algorithm will never reach step $8$. Fix a satisfying assignment $\V \alpha \in \Sigma^N$ to the instance. Consider the maps $\{ \tau_i:\Sigma\rightarrow \Sigma\mid 1\leq i\leq \ell\}$ from the Definition~\ref{def:symm_p} of the predicates being {\symm}. Using Property 1. from the Definition~\ref{def:symm_p}, we have that for all $1\leq i\leq \ell$, $\V \alpha^{(i)} \in \Sigma^N$, where $\alpha^{(i)}_j = \tau_i(\alpha_j)$ for $1\leq j\leq N$, is also a satisfying assignment to the instance.  Using this and Lemma~\ref{lemma:includes_all_sat_assignments}, for every $C\in \supp(\calC)$, the local distribution $\mu^\star_C$ in the SDP solution has in its support the set $Z_{\V \alpha, C}:=\{(\tau_i(\sigma_1), \tau_i(\sigma_2), \tau_i(\sigma_3)) \mid i\in [\ell]\}\subseteq \Sigma^3$ for $\sigma = \V \alpha|_C$. Now using Property 2. from the Definition~\ref{def:symm_p} of the predicates being {\symm}, the set $Z_{\V \alpha, C}$ (and hence $\supp(\mu^\star_C)$) is pairwise connected and does not have a $(\mathbb{Z},+)$-embedding. Hence, the algorithm will never reach step $8$.
\end{proof}

Finally, we compute the running time of Algorithm~\ref{alg:massage_sdp_ge}
\begin{lemma}
\label{lemma:alg_runtime}
    Algorithm~\ref{alg:massage_sdp_ge} on a satisfiable CSP instance $\inst(\calV, \calC)$ where every predicate $C\in \calC$ is {\symm} runs in time $O_{|\Sigma|}(1)\cdot|\calV|^{O(1)}$.
\end{lemma}
\begin{proof}
    First, the procedure $\mathtt{ModifySDP}$ runs in time  $O_{|\Sigma|}(1)\cdot|\calV|^{O(1)}$. As every modification to the SDP program strictly reduces the support of one of the local distributions, there can be at most $poly(|\calV|, |\Sigma|)$ calls to the $\mathtt{ModifySDP}$ procedure through the execution of Algorithm~\ref{alg:massage_sdp_ge}.

    In the procedure $\mathtt{ModifyGESystem}$, for a constraint $C$ such that ${\sf span}({\sf supp}(\mu_C)) \subsetneq \calT|_{C}$, then the {\gesystem}, as modified in Algorithm~\ref{alg:modifyGEsystem}, satisfies $ {\sf span}({\sf supp}(\mu_C))\supseteq \calT|_{C}$. The number of embedding functions is $O_{|\Sigma|}(1)$, and therefore, every call to $\mathtt{ModifyGESystem}$ takes at most $O_{|\Sigma|}(1)\cdot|\calV|^{O(1)}$ time.

    Therefore, the overall running time of Algorithm~\ref{alg:massage_sdp_ge} is $O_{|\Sigma|}(1)\cdot|\calV|^{O(1)}$
\end{proof}

\newcommand{\sme}{{\sigma^\star}}

\subsection{The Hybrid Approximation Algorithm}
\label{sec:hyb_approx}

 In this section, we give our hybrid algorithm for satisfiable instances of Max-$\mathcal{P}$-CSP where $\mathcal{P}$ is a collection of {\symm} predicates. \\

 \noindent {\bf Hybrid Algorithm ($\alg$):} the algorithm is given below.

\begin{algorithm}[t]
\caption{Hybrid Approximation Algorithm for {\symm} Predicates}
\label{alg:hybrid_algorithm}

\KwIn{An instance $\inst(\calV,\calC)$ of $\mathrm{Max}\text{-}\mathcal{P}\text{-}\mathrm{CSP}$ where $\mathcal{P}$ is a collection of $3$-ary {\symm} predicates.}
\KwOut{Accept or Reject.}

Perform the operations described in Algorithm~\ref{alg:massage_sdp_ge}.\;

\eIf{the procedure outputs an SDP solution $(\mathbf{V}^\star,\boldsymbol{\mu}^\star)$ and a system of linear equations ${\gesystem}^\star$ that are consistent with each other}{
    Accept.\;
}{
    Reject.\;
}

\end{algorithm}

It can be checked easily that Lemma~\ref{lemma:alg_runtime} implies the algorithm runs in polynomial time. We define the following quantity $ \alpha^\alg_{\mathcal{P}}$, and by definition, the approximation guarantee of the hybrid algorithm is $ \alpha^\alg_{\mathcal{P}}$.

\begin{align*}
    \alpha^\alg_{\mathcal{P}} = {\sf inf}_\beta \left\{ \begin{array}{l}
         \mbox{$\exists$ an instance $\inst$ of Max-$\mathcal{P}$-CSP such that}\\
         \mbox{1. OPT($\inst$) $\leq \beta$},\\
         \mbox{2. SDP value $= 1$},\\
         \mbox{3. The hybrid algorithm accepts}.
    \end{array}\right.
\end{align*}

\begin{theorem}
\label{thm:approx_alg}
    For any collection of {\symm} predicates $\mathcal{P}$, the algorithm $\alg$ for Max-$\mathcal{P}$-CSP distinguishes between the following two cases.
    \begin{enumerate}
        \item The input instance is satisfiable.
        \item The input instance has value at most $\alpha^\alg_{\mathcal{P}}$.
    \end{enumerate}
\end{theorem}
\begin{proof}
    If the instance is satisfiable, then by Lemma~\ref{lemma:sdp_noZ_local_globalembed}, $\alg$ accepts. If the instance has value strictly less than $\alpha_\alg$, then by definition, $\alg$ rejects such instances.
\end{proof}

 	In Section~\ref{section:dict_test_actual}, we study a dictatorship test for Max-$\mathcal{P}$-CSP whose soundness matches with the approximation guarantee of the hybrid algorithm. In other words, we design a test where the test accepts according to the predicates $\mathcal{P}$, has perfect completeness, and has soundness $\alpha^\alg_{\mathcal{P}}+\eps$, for every constant $\eps>0$.

 \subsection{The Dictatorship Test}
	\label{section:dict_test_actual}
	
\newcommand{\amap}{\varsigma}
\newcommand{\constraint}{C}

 Fix an $\inst = (\calV, \calC)$ of a Max-$\mathcal{P}$-CSP such that the algorithm $\alg$ accepts $\inst$. Let $\Sigma$ be the alphabet of the CSP, and identify it with the set $[q] = \{1, 2, \ldots, q\}$ where $q=|\Sigma|$. Let $(\V V, \V \mu)$ be a solution for the final SDP relaxation of $\inst$.

 Given a function $F: \Sigma^R\rightarrow \Sigma$, in our dictatorship test $\dict_{\V V, \V \mu}$, we will sample three queries according to the distribution $\mu_\constraint^{R}$ for $\constraint\sim \calC$. For the test $\dict_{\V V, \V \mu}$, there is no single natural choice of probability measure $\mu^n$ on $\Sigma^n$ using which we can define the function $F$ to be `far-from-dictator' required (as in Definition~\ref{def:low_influence_mixed_invariance}) for the application of our mixed invariance principle.  Therefore, we need to define quasirandomness of a function appropriately, and we do so in the next section.

\subsubsection{The Notion of Quasirandom Functions}

For each variable $s\in \calV$ in the original Max-$\mathcal{P}$-CSP instance $\inst = (\calV, \calC)$, there is a corresponding probability space $\Omega_s = (\Sigma, \mu_s)$. Therefore, we will define the relative notion of ``quasirandom with respect to $(\V V , \V \mu)$''. Roughly speaking, we shall call a function ``quasirandom'' if all its ``influences'', after removing the high-degree contributions, are low under all of the probability distributions corresponding to variables $s_i \in\calV$. We make this formal in the following definition.

\begin{definition}\label{def:quasirandom_function} (Quasirandom function w.r.t. $(\V V, \V \mu)$)
		Fix a function $F: \Sigma^n \rightarrow \Sigma$. Let $\alpha$ be the minimum probability of the atoms in the SDP solution. Consider the functions $(F_{s,1}, F_{s,2}, \ldots, F_{s,q})$, where for every $j\in[q]$, $F_{s,j}: ([q]^n, \mu_s^{\otimes n}) \rightarrow \mathbb{R}$ is defined as $F_{s,j}(\V x):=\mathbf{1}_{F(\V x) = j}$. Fix the group $G^s_{\sf master}$ for each $s\in \calV$ from the hybrid algorithm. Suppose that for each $s$ we have a collection of product functions $\calP_{s}$ over $G^s_{\sf master}$ such that there are functions $F'_{s,j}$ of the form
		$F'_{s,j}(x) = \sum\limits_{P\in\spn(\mathcal{P}_{s})} P(x) \cdot L_{P,s,j}(x)$
		satisfying:
		\begin{enumerate}
			\item $\norm{F_{s,j} - F'_{s,j}}_{\mathcal{D}_{s}, \alpha}\leq \eps$, where $\mathcal{D}_{s}=\{\mu_C\mid s\in \calV(C)\}$.
			\item $\card{\mathcal{P}_s}$, ${\sf deg}(L_{P,s,j})$ and $\norm{L_{P,s,j}}_2$ are all at most $d = O_{\alpha, q, \eps}(1)$.
            \item For every $C\in {\sf supp}(\calC)$ with $\calV(C) = (s_1, s_2, s_3)$, considering the collection of product functions $\mathcal{D_C} = \sett{D_P,D_Q,D_R}{P\in \mathcal{P}_{s_1},Q\in\mathcal{P}_{s_2},R\in\mathcal{P}_{s_3}}$ over $\supp(\mu_C)$ where $D_P(x,y,z) = P(x)$,
            $D_Q(x,y,z) = Q(y)$
            and $D_{R}(x,y,z) = R(z)$, we have that ${\sf rk}(\mathcal{D}_C)\geq T$.
   \end{enumerate}

		Then, $F$ is called $(d, \tau)$-quasirandom w.r.t. $(\V V, \V \mu)$ if for every $\eps>0$, there is arbitrary large enough $T\gg \eps^{-1}, d, q, \alpha^{-1}$ such that every $s\in \calV$, $j\in [q]$ the above properties hold and furthermore for every $P\in \spn(\calP_{s})$ and $j\in [q]$, $\max_{i\in [n]}I_i[L_{P,s,j}]\leq \tau$.
	\end{definition}

	\begin{remark}\label{remark:quasirandom_function}
	The standard notion of $(d, \tau)$-quasirandomness that is often used in hardness of approximation is that $I_i[F_{s,j}^{\leq d}]\leq \tau$ for all $j\in [q]$ and $i\in [n]$, where $F_{s,j}^{\leq d}$ is the truncation of $F_{s,j}$ up to degree $d$. In particular, the high-degree component of $F_{s,j}$ is often ignored since it can often be shown not to contribute much. In our case, since we care about perfect completeness, the high degree part of $F_{s,j}$ cannot be simply ignored. Thus, the above notion of quasirandomness can be seen as a natural analog of the usual one, where we enforce small influence conditions on all the low-degree components $L_{P,s,j}$ in $F_{s,j}$s. At this point, we do not know how to use this notion of quasirandomness to convert a dictatorship test to a (conditional) hardness result, and we leave this as an open challenge for future research.
	\end{remark}
	
	One may wonder if it is even possible to find the functions $F'_{s,j}$ such that all the above (except for the low influence property) properties hold for small enough $\eps$ and large $T$. Lemma~\ref{lem:regularity_fancy_mult} and Lemma~\ref{lem:regularity_fancy_highrank} already show the existence of product functions $\calP_{s}$ for individual $F_{s,j}$ satisfying the above conditions in some setting. The following lemma generalizes that argument to our current setting.
	\begin{lemma}\label{lem:collection_highrank_triples}
		For all $\alpha>0$, $m, M\in\mathbb{N}$ and $\xi>0$ there exist $r\in\mathbb{N}$ and $\eps_0>0$ such that
		the following holds for any decay function $w\colon [0,1]\to[0,1]$.
		Let $\Sigma$ be an alphabet of size at most $m$, and let $\mathcal{M}$ be a family of subsets of $\{1, 2, \ldots, M\}$ of size $3$. Let $\{\mu_{(s_1, s_2, s_3)}~\mid~ (s_1, s_2, s_3)\in \mathcal{M}\}$ be distributions over $\Sigma^3$ in which the probability of
		each atom is at least $\alpha$ and the distribution $\mu_{(s_1, s_2, s_3)}$ is pairwise connected and has no $\mathbb{Z}$ embedding for all $(s_1, s_2, s_3)\in \mathcal{M}$. Suppose the marginals of $\mu_{(s_1, s_2, s_3)}$ are $\mu_{s_1}$, $\mu_{s_2}$, $\mu_{s_3}$ on the three coordinates, respectively. Then there are
        Abelian groups $G^s$ whose sizes depend only on $|\Sigma|$, such that for any $1$-bounded functions $F_{s,j}\colon (\Sigma^n, \mu_s^n)\to \mathbb{C}$, for $1\leq s\leq M$ and $1\leq j\leq |\Sigma|$,
		there is $\eps\geq \eps_0$ and a collection $\mathcal{P}_s$ of product functions of size at most $r$ over the group $G^s$ such that:
		\begin{enumerate}
			\item For all $\eps'\in (w(\eps),\eps)$ we have that
			$\norm{F_{s,j} - \mathrm{T}_{\calP_s, 1-\eps'} F_{s,j}}_{\mathcal{D}_{s},\alpha}\leq\xi$, where  $\mathcal{D}_s = \{\mu_{(s_1, s_2, s_3)}~|~s_1=s, s_2=s\text{ or }s_3=s\}$.

            \item For every $(s_1, s_2, s_3)\in \mathcal{M}$, considering the collection of product functions
            $$\mathcal{D}_{(s_1, s_2, s_3)} = \sett{D_P,D_Q,D_R}{P\in \mathcal{P}_{s_1},Q\in\mathcal{P}_{s_2},R\in\mathcal{P}_{s_3}}$$ over $\supp(\mu_{(s_1, s_2, s_3)})$ where $D_P(x,y,z) = P(x)$,
            $D_Q(x,y,z) = Q(y)$
            and $D_{R}(x,y,z) = R(z)$, we have that ${\sf rk}(\mathcal{D}_{(s_1, s_2, s_3)})\geq  \frac{1}{w(\eps)}$.
		\end{enumerate}
        Furthermore, if the collection of distributions  $\{\mu_{(s_1, s_2, s_3)}~\mid~ (s_1, s_2, s_3)\in \mathcal{M}\}$ are coming from the valid SDP solution from Lemma~\ref{lemma:sdp_noZ_local_globalembed}, then the groups $G^s$ can be taken to be the groups $G^s_{\sf master}$ from the hybrid algorithm.
	\end{lemma}
	\begin{proof}
		Fix a decay function $w': [0,1] \rightarrow [0,1]$ depending on $w$ (to be determined). We first apply Lemma~\ref{lem:regularity_fancy_mult} to get a collection of product functions $\calP_s$ such that for all $s,j$ and $\eps'\in (w'(\eps), \eps)$,
        \[
        \|F_{s,j} - T_{\calP_s, 1-\eps'}F_{s,j}\|_{\mathcal{D}_s, \alpha} \leq \xi.
        \]
        Inspecting the proof of Lemma~\ref{lem:regularity_fancy_mult}, we see that because we apply it with the collection $\mathcal{D}_s$, the collection $\mathcal{P}_s$ consists of product functions over $G^{s}_{{\sf master}}$.
        Let $r = \max_{s}\card{\mathcal{P}_s}$, and note that $r = O_{\alpha,m,M,\xi}(1)$.
        Let $\eps_0$ be from Lemma~\ref{lem:regularity_fancy_mult} and let $\eps\geq \eps_0$.

        Next, we modify collection $\calP_s$
        as in Lemma~\ref{lem:regularity_fancy_highrank} to achieve the second item of the lemma. Denote $A = \sum_{(s_1,s_2,s_3)\in\mathcal{M}}\card{\spn(\mathcal{D}_{(s_1,s_2,s_3)})}$, and note that $A = O_{\alpha,m,M,\xi}(1)$. Define the decay functions $w_1(z) = 0.5w(z)$
        and $w_{i+1}(z) = 0.5 w(w_{i}(z))$ for $i\geq 1$.
        Choose $w'(z) = w_{6A}(z)$ and inspect the intervals
        $I_i = (w_{3i}(\eps)^{-1},w_{3(i+1)}(\eps)^{-1})$ for $i=1,\ldots, 2A-1$.
        By the pigeonhole principle, we may find $i$ such that for each $(s_1,s_2,s_3)\in\mathcal{M}$ and
        $D\in \spn(\mathcal{D}_{s_1,s_2,s_3})$ we have that $\Delta_{{\sf symbolic}}(D,1)\not\in I_i$.
        We take
        \[
        J = \sett{\ell\in [n]}
        {\exists (s_1,s_2,s_3)\in\mathcal{M},D\in \spn(\mathcal{D}_{s_1,s_2,s_3})
        \text{ with }
        \Delta_{{\sf symbolic}}(D,1)\leq w_{3i}(\eps)^{-1}, D_{\ell}\neq 1
        }
        \]
        and set $\mathcal{P}_s' = \sett{P|_{\bar{J}}}{P\in\mathcal{P}_s}$ and $\mathcal{D}'_{s_1,s_2,s_3}$ accordingly. Then for all $(s_1,s_2,s_3)\in\mathcal{M}$ and
        $D\in \spn(\mathcal{D}'_{s_1,s_2,s_3})$ we either have $D\equiv 1$ or
        \[
        \Delta_{{\sf symbolic}}(D,1)
        \geq w_{3(i+1)}(\eps)^{-1}-\card{J}
        \geq w_{3(i+1)}(\eps)^{-1}
        -w_{3i}(\eps)^{-1}A
        \geq \frac{1}{w(w_{3i+2}(\eps))},
        \]
        so the second item holds.
        Applying the argument in the proof of Claim~\ref{claim:remains_small_norm} gives the first item.
 	\end{proof}
Thus, Lemma~\ref{lem:collection_highrank_triples} shows that given any function $f: \Sigma^n\rightarrow \Sigma^n$, we can find decompositions of $f$ guaranteeing properties 1., 2. and 3.~listed in  Definition~\ref{def:quasirandom_function}. Next, we will design a dictatorship test with perfect completeness, such that for a given $f$, if any such decompositions happen to satisfy the quasirandom property (namely, it additionally satisfies the low-influence property), the test passes with probability at most the optimal value of the instance. If for any such decompositions of $f$, some $L_P$ has an influential variable, then we hope such a structure can be used in the (conditional) hardness reductions, specifically in the {\em decoding} procedure of the soundness analysis of a reduction.

\subsubsection{The Dictatorship Test for {\symm} Predicates}
We are now ready to state the dictatorship test and prove its completeness and soundness. Let $(\V V, \V \mu)$ be a solution for the SDP relaxation of $\inst$ and $\mathcal{T} \subseteq \prod_{v\in \calV} H_v^\star$ be the subspace of the set of satisfying assignments of the instance $\inst(\calV, \calC)$ coming from the {\gesystem} that satisfy the conditions from Lemma~\ref{lemma:sdp_noZ_local_globalembed}. Let $\alpha>0$ be a lower bound on the non-zero probabilities from the distributions $\mu_C$; we treat $\alpha$ as a constant as the instance size is fixed.

 In Figure~\ref{fig:dict_test_noise}, we give the dictatorship test $\dict_{\V V, \V \mu}$ for functions $F: \Sigma^R \rightarrow\Sigma$ for large enough $R$ compared to $|\calV|, \alpha^{-1}, |\Sigma|$.
	\begin{figure}[!ht]
		\fbox{
			
			\parbox{460pt}{
				
				\begin{enumerate}
					\item Let $(\V V, \V \mu)$ be a solution for the basic SDP relaxation of $\inst$ that satisfy the conditions from Lemma~\ref{lemma:sdp_noZ_local_globalembed}.
					
					\item Sample a payoff $\constraint\sim \calC$. Let $\calV(\constraint) = \{s_1,s_2, s_3\}$.
					
					\item Sample $\V z_{\constraint} = \{ \V z_{s_1}, \V z_{s_2}, \V z_{s_3}\}$ from the product distribution $\mu_{\constraint}^{R}$, i.e.,  independently for each $i\in [R]$, $(\V z_{s_1}^{(i)}, \V z_{s_2}^{(i)}, \V z_{s_3}^{(i)}) \sim \mu_{\constraint}$.
					
					\item Query the function values $F({\V z}_{s_1}), F({\V z}_{s_2}), F({\V z}_{s_3})$.
					
					\item Accept iff $\constraint(F({\V z}_{s_1}), F({\V z}_{s_2}), F({\V z}_{s_3})) = 1$.
				\end{enumerate}
		}}\\
		\caption{SDP integrality gap to a dictatorship test $\dict_{\V V, \V \mu}$.}
		\label{fig:dict_test_noise}
	\end{figure}

        \noindent{\bf Completeness:} a function $F:\Sigma^R \rightarrow \Sigma$ is called a dictator function if $F(\V z) = \V z^{(i)}$ for some $i\in [R]$. The completeness of the test is defined as follows,
$$\mbox{Completeness}(\dict_{\V V, \V \mu}) = \min_{\substack{i\in [R],\\ F\mbox{ is the $i^{th}$ dictator}}} \Pr[F\mbox{ passes }\dict_{\V V, \V \mu}].$$

As the distribution $\mu_C$ is supported on the satisfying assignment of the constraint $C$, the test passes with probability $1$ when $F$ is a dictator function. Therefore, $\mbox{Completeness}(\dict_{\V V, \V \mu})=1$.\skipi

\noindent{\bf Soundness:} The soundness of the test is the maximum probability with which it accepts quasirandom functions. More formally, define the soundness of the test as follows.
 $$\mbox{Soundness}_{(d, \tau)}(\dict_{\V V, \V \mu})  = \sup_{\substack{F: \Sigma^R \rightarrow \Sigma\\ F \mbox{ is } (d, \tau)\mbox{-quasirandom w.r.t.} (\V  V, \V \mu)}} \Pr[F \mbox{ passes }\dict_{\V V, \V \mu}].$$

 We now state our main theorem regarding the soundness of the dictatorship test $\dict_{\V V, \V \mu}$.

	\begin{theorem}\label{thm:dictatorship_test}
		Fix any collection of {\symm} predicates $\mathcal{P}$. Given an instance $\inst = (\calV, \calC)$ of a Max-$\mathcal{P}$-CSP such that the algorithm $\alg$ accepts $\inst$, the test $\dict_{\V V, \V \mu}$ has completeness $1$ and soundness $$\mbox{Soundness}_{(d, \tau)}(\dict_{\V V, \V \mu}) = {\sf OPT}(\inst) + o_{\tau}(1),$$ where ${\sf OPT}(\inst)$ is the optimal value of the instance $\inst$ and $o_{\tau}(1)\rightarrow 0$ as $\tau\rightarrow 0$.
	\end{theorem}
	
	The following corollary follows from the above theorem and the approximation guarantee of our hybrid algorithm.

 \begin{corollary}
     \label{corr:dict_test}
     Fix any collection of {\symm} predicates $\mathcal{P}$. For every $\eps>0$, there exists a dictator vs.~quasirandom test with completeness $1$ and soundness $\alpha^\alg_\mathcal{P}+\eps$, and the accepting criterion of the test is from the set of predicates $\mathcal{P}$.
 \end{corollary}

Theorem~\ref{thm:approx_alg} along with Corollary~\ref{corr:dict_test} prove our main Theorem~\ref{thm:main}. We prove Theorem~\ref{thm:dictatorship_test} in the rest of the section. We first set up a few notations in Section~\ref{section:functions} that will be used in the analysis of the soundness of the test.\skipi

	\subsubsection{Functions on Product Spaces}
	\label{section:functions}
	Let $(\Omega, \mu)$ be a probability space with $|\Omega| = q$ and $\mu$ has full support on $\Omega$. Define the inner product between two functions $f,g: \Omega\rightarrow \mathbb{C}$ on this space as follows: $\langle f, g\rangle = \E_{x\sim \mu}[f(x)\overline{g(x)}]$.
	
	\begin{definition}
		An orthonormal ensemble consists of a basis of real orthonormal random variables $\calL = \{ \ell_0 \equiv \mathbf{1}, \ell_1, \ldots, \ell_{q-1}\}$, where $\mathbf{1}$ is the constant $1$ function.
	\end{definition}
	
	Henceforth, we will sometimes refer to orthonormal ensembles as just ensembles. For an ensemble $\calL = \{ \ell_0 \equiv \mathbf{1}, \ell_1, \ldots, \ell_{q-1}\}$ of random variables, we will use $\calL^R$ to denote the
	ensemble obtained by taking $R$ independent copies of $\calL$. Further $\calL^{(i)}= \{ \ell_0^{(i)}, \ell_1^{(i)}, \ldots, \ell_{q-1}^{(i)}\}$ will denote the $i^{th}$ independent copy of $\calL$.

	Fix an ensemble $\calL = \{ \ell_0 \equiv \mathbf{1}, \ell_1, \ldots, \ell_{q-1}\}$ that forms a basis for $L^2(\Omega)$. Given such a basis for $L^2(\Omega)$, it induces a basis for the space $L^2(\Omega^R)$, given by the random variables
	\[
    \left\{\ell_{\V \sigma} :=\prod_{i=1}^R \ell_{\sigma_i}^{(i)} \quad \middle| \quad {\V \sigma} \in \{0,1,\ldots,q-1\}^R\right\}.
    \]
	Therefore, any function $\calF: \Omega^R \rightarrow \mathbb{R}$ has a multilinear expansion
	\[
    \calF(\V z) = \sum_{{\V \sigma} \in \{0,1,\ldots,q-1\}^R} \hat{\calF}({\V \sigma}) \ell_{\V \sigma}(\V z),
    \]
	where $\ell_{\V \sigma}(\V z)= \prod_{i=1}^R \ell_{\sigma_i}(z_i)$.

	\begin{definition}
		A multi-index ${\V \sigma}$ is a vector $(\sigma_1, \sigma_2, \ldots. \sigma_R)\in \{0,1,\ldots,q-1\}^R$ and the degree of ${\V \sigma}$ is denoted by $|{\V \sigma}|$ which is equal to $|{\V \sigma}| = |\{i\in [R] \mid \sigma_i \neq 0\}|$. Given a set of indeterminates $\bcalX = \{x_j^{(i)} | j\in \{0,1,\ldots,q-1\} ,i\in [R]\}$ and a multi-index ${\V \sigma}$, define the monomial $x_{\V \sigma}$ as
		$x_{\V \sigma} = \prod_{i=1}^R x_{\sigma_i}^{(i)}$.
		The degree of the monomial is given by $|{\V \sigma}|$. A multilinear polynomial over such indeterminates is given by
		$$F(\V x) = \sum_{{\V \sigma}\in \{0,1,\ldots,q-1\}^R} \hat{F}_{\V \sigma} x_{\V \sigma}.$$
	\end{definition}

	Given a function $\calF: \Omega^R \rightarrow \mathbb{C}$ whose multilinear expansion with respect to the orthonormal ensemble $\calL = \{ \ell_0 \equiv \mathbf{1}, \ell_1, \ldots, \ell_{q-1}\}$ is given by $\calF(\V z) = \sum_{\sigma \in \{0,1,\ldots,q-1\}^R} \hat{\calF}(\sigma) \ell_\sigma(\V z)$ , we define a corresponding {\em formal polynomial} in the indeterminates $\bcalX = \{x_j^{(i)} | j\in \{0,1,\ldots,q-1\} , i\in [R]\}$, as follows:
	
	$$F(\V x) = \sum_{{\V \sigma}} \hat{\calF}({\V \sigma}) x_{\V \sigma}.$$
	
	%Note that with this correspondence, the Fourier degree and the degree of the corresponding formal polynomial coincide.
	
	We will always use the symbol $\calF$ to denote a function on a product probability space $\Omega^R$.
    Further $ F(\V x)$ will denote the formal multilinear polynomial corresponding to $\calF$. Hence, $ F(\calL^R)$ is
	a random variable obtained by substituting the random variables $\calL^R$ in place of $\V x$. For instance, the following equation holds in this notation:
	$$\E_{\V z\in \Omega^R}[\calF(\V z)] = \E[ F(\calL^R)].$$

\paragraph{Vector Valued Functions.}	We will always use the symbol $\bcalF = (\calF_1, \calF_2, \ldots, \calF_{q})$ to denote a vector-valued function on a product probability space $\Omega^R$. Further, $ \V F(\V x) = (F_1, F_2, \ldots, F_{q})$ will denote the formal multilinear polynomial corresponding to $\bcalF$.

	\subsubsection{Soundness Analysis of the Dictatorship Test}
	
 We now analyze the soundness of the test. Towards this, we assume that $F$ is quasirandom according to Definition~\ref{def:quasirandom_function} and we fix decompositions of $F$ that guarantee the quasirandomness of $F$ throughout the analysis.  For each $s\in \calV$, let $\Omega_s = (\Sigma, \mu_s)$ be a probability space with atoms in $\Sigma$ where the probability of $a\in \Sigma$ is $\|\V b_{s,a}\|_2^2$.

	We will fix an arbitrary mapping from $\Sigma$ to $\{1,2,\ldots, q\}$, denoted by $\amap: \Sigma \rightarrow \{1,2,\ldots, q\}$. The domain of the payoff $\constraint: \Sigma^3 \rightarrow \{0,1\}$ can be extended from $\Sigma^3$ to $\simplex_{q}^3$. To see this, by the abuse of notation, first define a $\Delta_q$-representation of a payoff $\constraint: \Sigma^3 \rightarrow \{0,1\}$ as $\constraint:\Delta_q^3 \rightarrow \{0,1\}$ where
	$$\constraint(\V e_{a_1},\V e_{a_2}, \V e_{a_3} ) = \constraint(\amap^{-1}(a_1), \amap^{-1}(a_2),\amap^{-1}(a_3)), \mbox{ for all } (a_1, a_2,a_3) \in \{1,2,\ldots, q\}^3.$$
	
	The function ${\constraint}$ can be extended to the domain $\simplex_q^3$ by its multi-linear extension. Again, by abusing the notation, define the extension ${\constraint}$ as:
	
	\begin{equation}
		\label{eq:payoff}
		{\constraint}(\V x_1, \V x_2, \V x_3)  = \sum_{\sigma\in \Sigma^3 } \constraint(\sigma) \prod_{i=1}^3 x_{i, \amap(\sigma_i)}, \mbox{ for all } \V x_1, \V x_2, \V x_3 \in \simplex_q.
	\end{equation}

	\paragraph{Extending $\constraint$ to $\mathbb{C}^{3q}$:} We will extend the payoff function $\constraint$ further to a real-valued
    function on $(\mathbb{C}^{q})^3$, by constructing a smooth extension of $\constraint$.  This smooth extension of $\constraint$ satisfies the following properties, controlled by a constant $C_0(q)$ that depends only on $q$.

	\begin{enumerate}
		\item All the partial derivatives of $\constraint$ up to order $3$ are uniformly bounded by $C_0(q)$.
		\item $C$ is a Lipschitz function with Lipschitz constant $C_0(q)$, i.e.,$\forall \{\V x_1, \V x_2, \V x_3\}, \{\V y_1, \V y_2, \V y_3\} \in (\mathbb{C}^{q})^{3}$,
		$$| \constraint(\V x_1, \V x_2, \V x_3) - \constraint(\V y_1, \V y_2, \V y_3)|\leq C_0(q) \sum_{i=1}^3 \|\V x_i - \V y_i\|_2.$$
	\end{enumerate}

	\paragraph{Local and Global Ensembles.}
	Fix a given SDP solution $(\V V, \V \mu)$ with value $1$. We define the following local and global orthonormal ensembles of random variables for every $s\in \calV$ as follows:
	
	\begin{itemize}
		\item {\bf Local Integral Ensembles $\calL$:}
		The Local	Integral Ensemble $\calL = \{\V \ell_s\mid s \in\calV\}$ for a variable $s\in \calV$, $\V \ell_s = \{\ell_{s,0}\equiv \mathbf{1}, \ell_{s,1}, \ldots, \ell_{s,q-1}\}$, is a set of random variables that are orthonormal ensembles for the space $\Omega_s$.
		\item {\bf Global Gaussian Ensembles $\calG$:}
		The Global Gaussian Ensembles $\calG = \{\V g_s\mid s \in\calV\}$ are generated by setting
		$\V g_s = \{g_{s,0} \equiv \mathbf{1}, g_{s,1}, \ldots, g_{s,q-1}\}$ where
		$${g}_{s,j} = \sum_{\omega\in \Sigma}\ell_{s,j}(\omega)\langle \V b_{s,\omega}, \V \zeta\rangle, \quad \forall j\in \{1, \ldots, q-1\},$$
		and $\V \zeta$ is a random normal Gaussian vector of appropriate dimension.
	\end{itemize}

	The following lemma states that the local integral ensemble and the global Gaussian ensemble have matching first and second moments. We need this to apply the invariance principle in our analysis below.
	\begin{lemma}(\cite{BKMcsp1})
		\label{lemma:matching_moments}
		For every $s\in \calV$, $\V g_s$ is an orthonormal ensemble w.r.t. the space $\Omega_s$. Also, for any payoff $\constraint\in \calC$, the global ensembles $\calG$ match the following moments of the local integral ensembles $\calL$:
		$$ \E_{\V \zeta}[g_{s,j}.g_{s',j'}] = \E_{(\omega, \omega')\sim \mu_{\constraint}|(s,s')}[\ell_{s,j}(\omega).\ell_{s',j'}(\omega')] \quad \forall j, j'\in \{1,\ldots, q-1\}, s, s'\in \calV(P'),$$
		where $\mu_{\constraint}|(s,s')$ is the marginal distribution of $\mu_{\constraint}$ on the coordinates of $s,s'$.
	\end{lemma}

\paragraph{Functions used in the soundness analysis.} In order to analyze the above expectation, we define functions related to $\calF_s$ for $s\in \calV$ analogously to Section~\ref{sec:mixed_invariance}.
     For small enough $\xi>0$, let $\eps>0$ be arbitrarily small constant compared to $\xi$ and take a parameter $T$ such that $T\gg 1/\xi,1/\eps, 1/\alpha, |\Sigma|$. Let $\mathcal{P}_{s}$ be the collection of product functions
	associated with $F'_{s,j}$ from the decomposition of $F:\Sigma^R\rightarrow \Sigma$ satisfying the quasirandomness property as in Definition~\ref{def:quasirandom_function} with the above parameters where $\tau$ is much smaller than the parameters $\xi, \eps$. Write each
	$P\in \spn(\mathcal{P}_{s})$ as
	$P = \prod\limits_{i=1}^{R}\chi_{P,i}(\sigma_s(x_i))$. Here, $\sigma_s: \Sigma\rightarrow H_s^\star$ is the embedding into the group $H_s^\star$ with respect to the set of distributions $M_{\mu_s} = \{\mu_C ~\mid~ s\in \calV(C)\}$. With these collections of product functions for each $s\in \calV$, we first define the noisy version of the functions $\calF_{s,j}$ for $j\in [q]$ as:
    \begin{equation}\label{eq:noisy_csp}
	\mathrm{T}_{\calP_s, 1-\eps}\calF_{s,j}(\V z)
	= \Expect{I\subseteq_{\eps} [n]}{\Expect{\V y\sim \mathrm{T}_{\mu_s, \mathcal{P}_s, I} \V z}{\calF_{s,j}(\V y)}},
	\end{equation}
	where the operator $\mathrm{T}_{\mu_s, \mathcal{P}_s, I}$ is defined in Definition~\ref{def:noise_op_P}. Using this, define the following decomposition which is close to the above function in the $\ell_2$ norm as proved in Lemma~\ref{lem:approx_formula_noise_op_fancy},

	\begin{equation}\label{eq:decompose_L_P_csp}
	\calF'_{s,j}(\V z)
	=\sum\limits_{P\in\spn(\mathcal{P}_{s})}
	L_{P,s,j}(\V z)\cdot \prod\limits_{i=1}^{R} \chi_{P,i}(\sigma_s(z_i)).
	\end{equation}
	Here  ${\sf deg}(L_{P,s,j})\leq d$ and $\norm{L_{P, s, j}}_2\leq d$, where $d$  depends on $\xi, \eps, \alpha, |\Sigma|$. These are the same decompositions used to define the quasirandomness of the function $F$ in Definition~\ref{def:quasirandom_function}. Again, by the $(d,\tau)$ quasirandomness condition, we have for every $P\in \spn(\calP_{s})$ and $j\in [q]$, $\max_{i\in [n]}I_i[L_{P,s,j}]\leq \tau$. Next, we define a decoupled version of the above function where we provide separate inputs for the low-degree part and the product function part
    (we recall that by Lemma~\ref{lem:decoupled_1} there exists a coupling between $\mu_s^R\times \mu_s^R$ and $\mu_s^R$ under which these two functions are close to each):
	\begin{equation}\label{eq:decoupled_csp}
	\tilde{\calF}^{\sf dec}_{s,j}(\V x,\V y)
	=\sum\limits_{P\in\spn(\mathcal{P}_{s})}
	L_{P,s,j}(\V y)\cdot \prod\limits_{i=1}^{R} \chi_{P,i}(\sigma_s(x_i)).
	\end{equation}
	Interpreting the first input over the group $(H_s^\star)^R$, we define ${\calF}^{\sf dec}_{s,j}\colon (H_s^\star)^R\times \Sigma^R \to\mathbb{C}$ by replacing $\sigma_s(x_i)$
	with the group element input. Namely,
	\begin{equation}\label{eq:decoupled_group_csp}
	{\calF}^{\sf dec}_{s,j}(\V a,\V y)
	=
	\sum\limits_{P\in\spn(\mathcal{P}_{s, j})}
	 L_{P,s,j}(\V y)\cdot \prod\limits_{i=1}^{R}\chi_{P,i}(a_i).
	\end{equation}
	Towards applying the invariance principle, we replace the low-degree polynomials with their corresponding multilinear extensions to get the following function
	${F}^{\sf dec}_{s, j}\colon (H_s^\star)^R\times \mathbb{R}^{(q-1)R }\to\mathbb{C}$:
        \[
	{F}^{\sf dec}_{s,j}(\V a,\V \ell)
	=
	\sum\limits_{P\in\spn(\mathcal{P}_{s, j})}
	 L_{P,s,j}(\V \ell)\cdot \prod\limits_{i=1}^{R}\chi_{P,i}(a_i).
	\]
 We will also need to work with the analogous function ${F}^{\sf dec}_{s, j}\colon \Sigma^R\times \mathbb{R}^{(q-1)R }\to\mathbb{C}$, where in the above definition $a_i$ is replaced with $\sigma(x_i)$, and we use the same notation to refer to this function as the distinction will be clear from the inputs.

 Finally, the function $\widetilde{F^{\sf dec}_{s, j}}\colon (H_s^\star)^R\times \mathbb{R}^{(q-1)R }\to\mathbb{R}$ is defined by taking only real part of the output of ${F}^{\sf dec}_{s, j}$ and truncating it using the function ${\sf trunc}_{[0,1]}$ defined in Figure~\ref{fig:rounding_scheme}.
 $$\widetilde{F^{\sf dec}_{s, j}}(\V a, \V \ell) = {\sf trunc}_{[0,1]}(\mathfrak{Re}({F}^{\sf dec}_{s,j}(\V a, \V \ell))).$$
We also use the notation $\widetilde{\qquad }$ on any function with the same definition as above. This completes the description of the various functions that we will encounter in the soundness analysis below.

	\paragraph{Soundness analysis.} The acceptance probability of the test for a given function $\bcalF$ is given by:
	\begin{equation}
		\label{eq:dict_test_passing}
		\Pr[\bcalF\mbox{ passes }\dict_{\V V, \V \mu}] = \E_{{\constraint}\sim \calC}  \underbrace{\E_{{\V z}_{\constraint} \sim \mu_{\constraint}^{R}}[{\constraint}(\bcalF_{s_1}({\V z}_{s_1}), \bcalF_{s_2}({\V z}_{s_2}), \bcalF_{s_3}({\V z}_{s_3}))]}_{(\rom{1})},
	\end{equation}
where $\calV(\constraint) = (s_2, s_2, s_3)$ and $C$ is the smooth extension of the payoff function. Note that for all $s$ and ${\V z} \in \Sigma^R$, we have $\sum_{j} \calF_{s, j}(\V z) = 1$. Furthermore, we have the same property for the noisy functions as given by the following claim.

 \begin{claim}
     \label{claim:noisy_sum_1_csp}
     For all $s$ and ${\V z} \in \Sigma^R$, we have $\sum_{j} \mathrm{T}_{\calP_s, 1-\eps}\calF_{s,j}(\V z) = 1$.
 \end{claim}
 \begin{proof}
     \[
     \sum_j \mathrm{T}_{\calP_s, 1-\eps}\calF_{s,j}(\V z)
	= \sum_j \Expect{I\subseteq_{\eps} [n]}{\Expect{\V y\sim \mathrm{T}_{\mu_s, \mathcal{P}, I} \V z}{\calF_{s,j}(\V y)}} = \Expect{I\subseteq_{\eps} [n]}{\Expect{\V y\sim \mathrm{T}_{\mu_s, \mathcal{P}_s, I} \V z}{\sum_j \calF_{s,j}(\V y)}} = 1.
    \qedhere\]
 \end{proof}

 We now apply Theorem~\ref{thm:basic_mixed_invariance} to the expression $(\rom{1})$, and we use $\constraint$ as $\Psi$ therein. We remark that in the current setting we use ${\sf trunc}_{[0,1]}$ instead of ${\sf trunc}$ as in the proof of Theorem~\ref{thm:basic_mixed_invariance}, but the argument goes through in the same way. The only difference is that in the argument in  Claim~\ref{claim:finish_mixed_inv} instead of using the fact that $f' = \mathrm{T}_{\calP, 1-\eps} f$ is $1$-bounded (as there), we use the fact that $\mathrm{T}_{\calP_s, 1-\eps}\calF_{s,j}$ is $[0,1]$-valued. Therefore, we conclude the following

\begin{equation}
    \label{eq:move_to_dec_fun}
    \card{
				(\rom{1})
				-
				\underbrace{\Expect{\substack{\V a_{\constraint}\sim \tilde{\mu}_{\constraint}^R\\ \calG^R}}{\constraint(\widetilde{\bF^{\sf dec}_{s_1}}({\V a}_{s_1}, {\V g}_{s_1}), \widetilde{\bF^{\sf dec}_{s_2}}({\V a}_{s_2}, {\V g}_{s_2}), \widetilde{\bF^{\sf dec}_{s_3}}({\V a}_{s_3}, {\V g}_{s_3}))}}_{(\rom{2})}
			}\lll \xi,
\end{equation}
where $a_{\constraint} = (\V a_{s_1}, \V a_{s_2}, \V a_{s_3})\sim \tilde{\mu}_{\constraint}^R$  is the following distribution on $(H_{s_1}^\star\times H_{s_2}^\star\times H_{s_3}^\star)^R$ given by first sampling $\V z_{\constraint} = (\V z_{s_1}, \V z_{s_2}, \V z_{s_3})\sim \mu_{\constraint}^R$ and then applying the embedding maps $\sigma_{s}: \Sigma\rightarrow H_s^\star$ coordinatewise to $z_{s_j}$.

At this point, we would like to make a couple of observations. In the expectation given in $(\rom{2})$, the distribution $\calG^R$ can be sampled globally using the vectors from the SDP solution. However, the distribution $\tilde{\mu}_{\constraint}^R$ is specific to the constraint $C$. We do not know how to sample the strings $\{\V a_s \}_{s\in \calV}$ globally that marginally match with the distributions $\tilde{\mu}_{\constraint}^R$ for every constraint $C$. To remedy this, we next show that we can move from the distribution $\tilde{\mu}_{\constraint}^R$ to a distribution that samples a random set of $R$ solutions to the {\gesystem} and outputs the restriction of these solutions to the variables $s_1, s_2, s_3$. In order to show the closeness of these two distributions, we use the property that the product functions that we choose satisfy the high-rank property and the SDP solution is consistent with the {\gesystem}, guaranteed by Lemma~\ref{lemma:sdp_noZ_local_globalembed}.

 Specifically, consider the following two distributions:
	\begin{enumerate}
		\item The distribution
		$\mathcal{D}_{\mathcal{T}}$: sample $R$  assignments $\V \alpha^{(1)}, \V \alpha^{(2)}, \ldots, \V \alpha^{(R)}$ from the set of satisfying assignment $\mathcal{T}\subseteq G_{\sf master}^{\calV}$ to the {\gesystem} independently and uniformly at random, consider $(a_i, b_i, c_i):=(\alpha^{(i)}_{s_1},\alpha^{(i)}_{s_2},\alpha^{(i)}_{s_3})\in H_{s_1}^\star\times H_{s_2}^\star\times H_{s_3}^\star$
		and output $(\V a, \V b, \V c)$.
		\item The distribution
		$\mathcal{D}_{\mu_{\constraint}}$: sample $(\V x,\V y,\V z)\sim \mu_{\constraint}^{R}$ and output
        $(\sigma_{s_1}(\V x), \sigma_{s_2}(\V y), \sigma_{s_3}(\V z))$.
	\end{enumerate}
    For any $s\in \mathcal{V}$ and $P\in\spn(\mathcal{P}_{s})$,  take $\chi_{P}\colon \left(H_{s_1}^{\star}\right)^R\to \mathbb{C}$ a character such that $P(x) = \chi_P(\sigma_{s}(x))$.
    For any distribution $\mathcal{D}$ on $(H_{s_1}^\star\times H_{s_2}^\star\times H_{s_3}^\star)^R$, consider the distribution $\chi(\mathcal{D})$ defined as follows: Sample $(\V a, \V b, \V c) \sim \mathcal{D}$ and output
    \[
		(\chi_{P}(\V a),
		\chi_Q(\V b),
		\chi_R(\V c))_{P\in\spn(\mathcal{P}_{s_1}),Q\in\spn(\mathcal{P}_{s_2}),
			R\in\spn(\mathcal{P}_{s_3})}.
		\]

	Note that although each coordinate $i\in [R]$ of the output of the distribution  $\chi(\mathcal{D})$ is a vector in $\mathbb{C}^D$ for some $D = O_{|\Sigma|}(1)$, the support of the distribution for each $i\in [R]$ is bounded by $O_{|\Sigma|}(1)$. The following lemma (Lemma~\ref{lem:bound_sd_mixed_inv_csp}) asserts that the distributions $\chi(\mathcal{D}_{\mu_{\constraint}})$, and $\chi(\mathcal{D}_{\mathcal{T}})$ are close in statistical distance.   Here, we crucially use the fact that the set of solutions $\mathcal{T}$ to the {\gesystem} and the SDP solution $(\V V, \V \mu)$ are consistent with each other as stated in Lemma~\ref{lemma:sdp_noZ_local_globalembed}. Before we state and prove the lemma, we need the following two facts.

	\begin{fact}\label{fact:trivial_equivalency2}
		Suppose $P \in \spn(\mathcal{P}_{s_1})$, $Q \in \spn(\mathcal{P}_{s_2})$, $R \in \spn(\mathcal{P}_{s_3})$ are such that $(P_iQ_iR_i)\not\equiv 1$
		for at least $T$ coordinates under the distribution $\mu_C$. Then letting $\alpha$ be the minimum probability of an atom in $\mu_C$, we have that
		\[
		\card{\Expect{(\V x,\V y,\V z)\sim \mu_C^R}{\chi_P(\sigma_{s_1}(\V x))\chi_Q(\sigma_{s_2}(\V y))\chi_R(\sigma_{s_3}(\V z))}}\leq (1-\Omega_{\alpha,\card{\Sigma}}(1))^{T}.
		\]
	\end{fact}
	\begin{proof}
		Let $\mathcal{J}$ be the set of $i$'s
		such that $(P_iQ_iR_i)\not\equiv 1$ under $\mu_C$.
		Then the left-hand side is equal to
		\[
        \prod\limits_{i\in \mathcal{J}}
		\card{\Expect{(x_i,y_i,z_i)\sim \mu_C}{\chi_{P_i}(\sigma_{s_1}(x_i))\chi_{Q_i}(\sigma_{s_2}(y_i))\chi_{R_i}(\sigma_{s_3}(z_i))}},
		\]
		and it suffices to upper-bound each
		one of the terms by $1-\Omega_{\alpha,\card{\Sigma}}(1)$. Fix $i\in \mathcal{J}$, and note that
		the values that $P_iQ_iR_i$
		may receive are discrete, all have absolute value $1$ and they are $\Omega_{\card{\Sigma}}(1)$ far apart in
		absolute value. Thus, there is a fixed constant $c_{\alpha, \card{\Sigma}}>0$ such
		that for $i\in \mathcal{J}$, we have
		$$\card{\Expect{(x_i,y_i,z_i)\sim \mu_C}{\chi_{P_i}(\sigma_{s_1}(x_i))\chi_{Q_i}(\sigma_{s_2}(y_i))\chi_{R_i}(\sigma_{s_3}(z_i))}}\leq 1-c_{\alpha, \card{\Sigma}},$$ and the fact follows.
	\end{proof}

A similar statement is also true if we replace the distribution with $\mathcal{D}_{\mathcal{T}}$.
\begin{fact}\label{fact:trivial_equivalency_group_2}
		Suppose $P \in \spn(\mathcal{P}_{s_1})$, $Q \in \spn(\mathcal{P}_{s_2})$, $R \in \spn(\mathcal{P}_{s_3})$ are such that $(P_iQ_iR_i)\not\equiv  1$
		for at least $T$ coordinates under the distribution $\mu_C$. If the SDP solution and the solutions to the {\gesystem} are consistent, we have that
		\[
		\card{\Expect{(\V a,\V b,\V c)\sim \mathcal{D}_{\mathcal{T}}}{\chi_P(\V a)\chi_Q(\V b)\chi_R(\V c)}}\leq (1-\Omega_{\card{\Sigma}}(1))^{T}.
		\]
	\end{fact}
	\begin{proof}
		Again, for the set $\mathcal{J}$ as before, we have

		\[
		\card{\Expect{(\V a,\V b,\V c)\sim \mathcal{D}_{\mathcal{T}}}{\chi_P(\V a)\chi_Q(\V b)\chi_R(\V c)}}
		=\prod\limits_{i\in \mathcal{J}}
		\card{\Expect{(a_i,b_i,c_i)\sim \mathcal{D}_{\mathcal{T}}^{(i)}}{\chi_{P_i}(a_i)\chi_{Q_i}(b_i)\chi_{R_i}(c_i)}},
		\]
		where we let $\mathcal{D}_{\mathcal{T}}^{(i)}$ be the tuple on $(s_1, s_2, s_3)$ corresponding to the assignment $\V \alpha^{(i)}$. Here, we used the fact that the SDP solution and the solutions to the {\gesystem} are consistent, and hence, for $i\notin \mathcal{J}$, $P_iQ_iR_i\equiv 1$ under $\mathcal{D}_{\mathcal{T}}^{(i)}$ by Lemma~\ref{lemma:PQR_local_global}. It suffices to upper-bound each
		one of the terms in the above product by $1-\Omega_{\card{\Sigma}}(1)$. Fix $i\in \mathcal{J}$. We know that $(P_iQ_iR_i)\not\equiv  1$ under the distribution $\mu_C$. As the SDP solution and the solutions to the {\gesystem} are consistent, again using Lemma~\ref{lemma:PQR_local_global}, we have $\chi_{P_i}(a)\chi_{Q_i}(b)\chi_{R_i}(c)\neq \chi_{P_i}(a')\chi_{Q_i}(b')\chi_{R_i}(c')$ for some $(a, b, c)\neq (a', b', c')$ in the support of $\mathcal{D}_{\mathcal{T}}^{(i)}$. The probability that this $(a, b, c)$ is sampled according to $\mathcal{D}_{\mathcal{T}}^{(i)}$ is at least $\Omega_{|\Sigma|}(1)$.

        Similar to the above proof, the distinct values of $\chi_{P_i}\chi_{Q_i}\chi_{R_i}$ are $\Omega_{\card{\Sigma}}(1)$ far apart in
		absolute value. Thus, there is a fixed constant $c_{\card{\Sigma}}>0$ such
		that for $i\in \mathcal{J}$, we have
		$$\card{\Expect{(a_i,b_i,c_i)\sim \mathcal{D}_{\mathcal{T}}^{(i)}}{\chi_{P_i}(a_i)\chi_{Q_i}(b_i)\chi_{R_i}(c_i)}}\leq 1-c_{\card{\Sigma}},$$ and the fact follows.
	\end{proof}

We now prove the lemma.

	\begin{lemma}\label{lem:bound_sd_mixed_inv_csp}
		Fix $\constraint\in \calC$. Suppose the sizes of each one of
		$\mathcal{P}_{s_1},\mathcal{P}_{s_2},\mathcal{P}_{s_3}$
		is at most $r$, and that the size of each one of $H_{s_1}^\star$ $H_{s_2}^\star$, $H_{s_3}^\star$ is at most $m$. Suppose further that for any $P\in\spn(\mathcal{P}_{s_1})$,
		$Q\in\spn(\mathcal{P}_{s_2})$
		and $R\in\spn(\mathcal{P}_{s_3})$,
		it holds that either $PQR\equiv 1$ under $\mu_C^R$
		or else $P_iQ_iR_i\not\equiv 1$ under $\mu_C$
		for at least $T$ of the coordinates
		$i\in [R]$. Then
		\[
		{\sf SD}(\chi(\mathcal{D}_{\mu_{\constraint}}), \chi(\mathcal{D}_{\mathcal{T}}))
		\lll_{m,r}
		(1-\Omega_{m,\alpha}(1))^{T}.
		\]
	\end{lemma}
 \begin{proof}
 Take any tuple
		$S = (a_{P}, b_{Q}, c_{R})_{P\in\spn(\mathcal{P}_{s_1}),Q\in\spn(\mathcal{P}_{s_2}),
			R\in\spn(\mathcal{P}_{s_3})}$ in either
		one of the supports. Consider their contribution to the two distributions.

        	\[
		\chi(\mathcal{D}_{\mathcal{T}})(S)
		=
		\Expect{(\V a,\V b,\V c)\sim \mathcal{D}_\mathcal{T}}{\prod\limits_{P\in\spn(\mathcal{P}_{s_1})}
			1_{\chi_P(a) = a_P}
			\prod\limits_{Q\in\spn(\mathcal{P}_{s_2})}
			1_{\chi_{Q}(b) = b_Q}
			\prod\limits_{R\in\spn(\mathcal{P}_{s_3})}
			1_{\chi_R(c) = c_R}}
		\]
		and
		\[
		\chi(\mathcal{D}_{\mu_{\constraint}})(S)
		=
		\Expect{(\V a, \V b, \V c)\sim \mathcal{D}_{\mu_{\constraint}}}{\prod\limits_{P\in\spn(\mathcal{P}_{s_1})}
			1_{\chi_P(a) = a_P}
			\prod\limits_{Q\in\spn(\mathcal{P}_{s_2})}
			1_{\chi_{Q}(b) = b_Q}
			\prod\limits_{R\in\spn(\mathcal{P}_{s_3})}
			1_{\chi_R(c) = c_R}}.
		\]
		Arithmetizing the indicators as
		$1_{\chi_P(a) = a_P} = \prod\limits_{\substack{a'\neq a_P\\ a'\in {\sf Image}(P)}}\frac{\chi_P(a) - a'}{a_P-a'}$ and similarly for the
		other ones, then opening things up,
		we get that there are coefficients
		$B(P,Q,R)$ that are at most $O_{m,r}(1)$ in absolute value such that
		\[
		\chi(\mathcal{D}_{\mathcal{T}})(S)
		=\sum\limits_{
			\substack{P\in\spn(\mathcal{P}_{s_1})
				\\
				Q\in\spn(\mathcal{P}_{s_2})
				\\
				R\in\spn(\mathcal{P}_{s_3})}}
		B(P,Q,R)\Expect{(\V a,\V b,\V c)\sim \mathcal{D}_\mathcal{T}}{\chi_P(a)\chi_Q(b)\chi_R(c)}
		\]
        and
		\[
        \chi(\mathcal{D}_{\mu_{\constraint}})(S)
		=\sum\limits_{
			\substack{P\in\spn(\mathcal{P}_{s_1})
				\\
				Q\in\spn(\mathcal{P}_{s_2})
				\\
				R\in\spn(\mathcal{P}_{s_3})}}
		B(P,Q,R)\Expect{(\V a,\V b,\V c)\sim \mathcal{D}_{\mu_{\constraint}}}{\chi_P(a)\chi_Q(b)\chi_R(c)}.
		\]
		Fix $P,Q,R$ and consider their contribution to $\chi(\mathcal{D}_{\mathcal{T}})(S)$
		and $\chi(\mathcal{D}_{\mu_{\constraint}})(S)$.
		For $P, Q, R$ such that $PQR\equiv 1$ in the support of $\mu^R$,
		the two contributions are the same. This follows from Lemma~\ref{lemma:PQR_local_global}.
		Else, by assumption $P_iQ_iR_i\not\equiv 1$ under $\mu_C$ for at least
		$T$ coordinates, and by Fact~\ref{fact:trivial_equivalency2} the second expectations is at most
		$(1-\Omega_{m,\alpha}(1))^{T}$ in
		absolute value. Since the SDP solution is consistent with the {\gesystem}, using Fact~\ref{fact:trivial_equivalency_group_2}, the same holds for the first expectation for the terms corresponding to such $P_iQ_iR_i$. It follows that
		\[
		\card{
			\chi(\mathcal{D}_{\mathcal{T}})(S)
			-
			\chi(\mathcal{D}_{\mu_{\constraint}})(S)}
		\leq \sum\limits_{
			\substack{P\in\spn(\mathcal{P}_{s_1})
				\\
				Q\in\spn(\mathcal{P}_{s_2})
				\\
				R\in\spn(\mathcal{P}_{s_3})}}
		\card{B(P,Q,R)}(1-\Omega_{m,\alpha}(1))^{T}
		\lll_{m,r}
		(1-\Omega_{m,\alpha}(1))^{T}.
		\]
		Therefore, we get that
		${\sf SD}(\chi(\mathcal{D}_{\mu_{\constraint}}), \chi(\mathcal{D}_{\mathcal{T}}))\lll_{m,r}
		\sum\limits_{S}(1-\Omega_{m,\alpha}(1))^{T}\lll_{m,r}(1-\Omega_{m,\alpha}(1))^{T}$.
 \end{proof}

 Now, consider the following expectation where both the inputs to the function $\widetilde{\bF^{\sf dec}_{s, \sigma}}$ are sampled using global distributions.

$$
(\rom{3})= \Expect{\mathcal{D}_{\mathcal{T}}, \calG^R}{\constraint(\widetilde{\bF^{\sf dec}_{s_1}}({\V a}_{s_1}, {\V g}_{s_1}), \widetilde{\bF^{\sf dec}_{s_2}}({\V a}_{s_2}, {\V g}_{s_2}), \widetilde{\bF^{\sf dec}_{s_3}}({\V a}_{s_3}, {\V g}_{s_3}))}.
$$
As the quantities inside the expectations from $(\rom{2})$ and $(\rom{3})$ are bounded by $O_q(1)$, using the above lemma, we have

\begin{equation}
    \label{eq:move_to_global_dist}
\card{(\rom{2}) - (\rom{3})} \lll_q {\sf SD}(\chi(\mathcal{D}_{\mu_{\constraint}}), \chi(\mathcal{D}_{\mathcal{T}}))\lll_{m,r}(1-\Omega_{m,\alpha}(1))^{T} \lll \xi,
\end{equation}
by choosing an appropriately large value of $T$.\skipi

Combining Equations (\ref{eq:dict_test_passing}), (\ref{eq:move_to_dec_fun})and (\ref{eq:move_to_global_dist}), we get

 \begin{align}\label{eq:testpassing_after_invariance}
		\card{\Pr[\bcalF\mbox{ passes }\dict_{\V V, \V \mu}] -  \Expect{{\constraint}\sim \calC}  {\Expect{(\mathcal{D}_\mathcal{T}, \calG^{R}) } {{\constraint}(\widetilde{\bF^{\sf dec}_{s_1}}({\V a}_{s_1}, {\V g}_{s_1}), \widetilde{\bF^{\sf dec}_{s_2}}({\V a}_{s_2}, {\V g}_{s_2}), \widetilde{\bF^{\sf dec}_{s_3}}(\V a_{s_3}, {\V g}_{s_3}))]}}} \lll \xi.
\end{align}

The second expectation in the above expression almost looks like the approximation guarantee of our rounding algorithm given in Figure~\ref{fig:rounding_scheme}. The only difference is that we scale the vector $\widetilde{\bF^{\sf dec}_{s}}({\V a}, {\V g})\in \mathbb{R}^q$ in step $4$.  We now analyze the loss that occurs because of this scaling. Let $\bF^{\sf dec}_{s}({\V a}, {\V g})^\star\in \simplex_q$ be the vector obtained by scaling it using the function $\scale$ defined in step $4$ of the rounding scheme $\Round_{\bcalF}$ below. The expected value of the solution returned by the rounding scheme $\Round_{\bcalF}$ is given by:
\begin{equation}\label{eq:rounding_scheme}
    \Round_{\bcalF}(\V V, \V \mu) = \Expect{{\constraint}\sim \calC}{ \Expect{(\mathcal{D}_\mathcal{T}, \calG^{R})}{{\constraint}(\bF^{{\sf dec}}_{s_1}({\V a}_{s_1}, {\V g}_{s_1})^\star, \bF^{{\sf dec}}_{s_2}({\V a}_{s_2}, {\V g}_{s_2})^\star, \bF^{{\sf dec}}_{s_3}({\V a}_{s_3}, {\V g}_{s_3})^\star) }}.
\end{equation}

Let $\delta = o_{\xi}(1)$. Here, the notation $o_{\xi}(1)$ means that the expression goes to $0$ as $\xi$ goes to $0$. Fix a constraint $C$ and a variable $s\in \calV(C)$. Let $({\V z}_{s}, {\V z'_{s}}, {\V Z_{s}})$ to be distributed according to the coupled distribution from Lemma~\ref{lem:coupling_decoupled} for $\kappa$ such that $T^{-1}\ll\kappa\ll d^{-1},\xi$. Let $E_s$ be the event that $\sum_j \mathfrak{Re}({\tilde{\calF}^{{\sf dec}}_{s, j}}({\V z}_{s}, {\V z'_{s}}) )\in [1- \delta, 1+\delta]$. The next claim shows that the event $E_s$ occurs with high probability over $({\V z}_{s}, {\V z'_{s}})$. We will use this later to show that the effect of truncation and scaling has little effect on the expression (\ref{eq:rounding_scheme}).

 \begin{claim}
 \label{claim:close_to_1_csp}
    For every $s\in \calV$, $\Pr[E_s]\geq 1-\delta$.
 \end{claim}
 \begin{proof}
Let $1\gg\delta'>0$ to be fixed later. Define $\varrho'(x)$ where  $\varrho'(x)= |\mathfrak{Re}(x)|$ if $\mathfrak{Re}(x)\in [-\delta'^2, \delta'^2]$ and $1$ otherwise. Consider the function $\zeta(b) = |1-b|$ such that the function. Let $({\V z}_{s}, {\V z'_{s}}, {\V Z_{s}})$ to be distributed according the coupled distribution from Lemma~\ref{lem:coupling_decoupled}. As $\zeta(1+b) = |b|$ for every $ b\in \mathbb{C}$ we get
\begin{align*}
  &\E_{(\V z_s, \V z'_s, \V Z_s)} \left[\varrho'\circ\zeta\left(\sum_j{\tilde{\calF}^{{\sf dec}}_{s, j}}({\V z_s}, \V z'_s)\right) \right]\\
  & \qquad= \E_{(\V z_s, \V z'_s, \V z_s)} \left[\varrho'\circ\zeta\left(\sum_j{\tilde{\calF}^{{\sf dec}}_{s, j}}({\V z_s}, \V z'_s) + \sum_j \mathrm{T}_{\calP_s, 1-\eps}\calF_{s,j}(\V Z_s) -  \sum_j \mathrm{T}_{\calP_s, 1-\eps}\calF_{s,j}(\V Z_s)\right) \right]\\
  & \qquad=  \E_{(\V z_s, \V z'_s, \V z_s)}\left[ \varrho'\left(\left|\sum_j{\tilde{\calF}^{{\sf dec}}_{s, j}}({\V z_s}, \V z'_s) - \sum_j \mathrm{T}_{\calP_s, 1-\eps}\calF_{s,j}(\V Z_s) \right|\right)\right]. \tag{Using Claim~\ref{claim:noisy_sum_1_csp}}
  \end{align*}
  Now, using the fact that $\varrho'(|x|)\leq \delta'^{-2}|x|$ for all $x$, we have
  \begin{align*}
    \E_{(\V z_s, \V z'_s, \V Z_s)} \left[\varrho'\circ\zeta\left(\sum_j{\tilde{\calF}^{{\sf dec}}_{s, j}}({\V z_s}, \V z'_s)\right) \right]&\leq \delta'^{-2}  \E_{(\V z_s, \V z'_s, \V Z_s)}\left[ \left|\sum_j{\tilde{\calF}^{{\sf dec}}_{s, j}}({\V z_s}, \V z'_s) - \sum_j \mathrm{T}_{\calP_s, 1-\eps}\calF_{s,j}(\V Z_s) \right|\right] \\
&\leq  \delta'^{-2} \E_{(\V z_s, \V z'_s, \V Z_s)}\left[ \left|\sum_j{\tilde{\calF}^{{\sf dec}}_{s, j}}({\V z_s}, \V z'_s) - \sum_j \mathrm{T}_{\calP_s, 1-\eps}\calF_{s,j}(\V Z_s) \right|^2\right]^{1/2}\\
&\lll_q \delta'^{-2} \sum_j   \E_{(\V z_s, \V z'_s, \V Z_s)}\left[ \left| \tilde{\calF}^{{\sf dec}}_{s, j}({\V z_s}, \V z'_s) -  \mathrm{T}_{\calP_s, 1-\eps}\calF_{s,j}(\V Z_s) \right|^2\right]^{1/2},
 \end{align*}
where the second line is the Cauchy-Schwarz inequality. Now,
\begin{align*}
      &\E_{(\V z_s, \V z'_s, \V Z_s)}\left[ \left| \tilde{\calF}^{{\sf dec}}_{s, j}({\V z_s}, \V z'_s) -  \mathrm{T}_{\calP_s, 1-\eps}\calF_{s,j}(\V Z_s) \right|^2\right] \\
      &\qquad = \sum_{j=1}^q \Expect{}{\card{\tilde{\calF}^{{\sf dec}}_{s,j}({\V z}_{s}, {\V z'_{s}})-\mathrm{T}_{\calP_s, 1-\eps}\calF_{s, j}({\V Z}_{s})}_2^2}\\
       &\qquad = \sum_{j=1}^q \Expect{}{\card{\tilde{\calF}^{{\sf dec}}_{s,j}({\V z}_{s}, {\V z'_{s}})- \calF'_{s, j}({\V Z}_{s}) + \calF'_{s, j}({\V Z}_{s})-  \mathrm{T}_{\calP_s, 1-\eps}\calF_{s, j}({\V Z}_{s})}_2^2}\\
       &\qquad \lll \sum_{j=1}^q \Expect{}{\card{\tilde{\calF}^{{\sf dec}}_{s,j}({\V z}_{s}, {\V z'_{s}})- \calF'_{s, j}({\V Z}_{s})}^2} + \Expect{}{\card{\calF'_{s, j}({\V Z}_{s})-  \mathrm{T}_{\calP_s, 1-\eps}\calF_{s, j}({\V Z}_{s})}_2^2}.
      \end{align*}

  Using Lemma~\ref{lem:approx_formula_noise_op_fancy}, the second expectation is upper bounded by $\xi$ and using Lemma~\ref{lem:decoupled_2}, the first expectation is upper bounded by $\lll_{d,q} {\kappa}\ll \xi$. Thus, we have for every $1\leq j\leq q$,
\[
\E_{(\V z_s, \V z'_s, \V Z_s)}\left[ \left|\tilde{\calF}^{{\sf dec}}_{s, j}({\V z_s}, \V z'_s) -  \mathrm{T}_{\calP_s, 1-\eps}\calF_{s,j}(\V Z_s) \right|^2\right] = o_{\xi}(1).
\]
Therefore, we get
\begin{align}
  \E_{(\V z_s, \V z'_s, \V Z_s)} \left[\varrho'\circ \zeta\left(\sum_j {\tilde{\calF}^{{\sf dec}}_{s, j}}({\V z_s}, \V z'_s)\right) \right]
\lll_q  \frac{o_{\xi}(1)}{\delta'^2}. \label{eq:scale_3}
\end{align}

For some setting of $\delta \approx  \delta'^{-2}o_{\xi}(1)$, we get that with probability at least $1-\delta$ it holds that $\sum_j \mathfrak{Re}({\tilde{\calF}^{{\sf dec}}_{s, j}}({\V z}_{s}, {\V z'_{s}}) )\in [1-\delta, 1+\delta]$. This finishes the proof of the claim.
 \end{proof}

We now remove the truncation and scaling, and see the effect of this on the expression from (\ref{eq:rounding_scheme}). Towards this, we have the following simple claim.

\begin{claim}\label{claim:trunc_scale_vec}
    Let ${\V a}\in \mathbb{R}^q$ such that $\sum_i {a}_i \in [1-\delta,1+\delta]$, and let ${\V a}^\star$ be the vector that we get after scaling the vector $\tilde{\V a} := {\sf trunc}_{[0,1]}(\V a)$ (as defined in Figure~\ref{fig:rounding_scheme}), then
    \[
    \sum_i (\tilde{a}_i - a^\star_i)^2 \leq q \sum_i (a_i - \tilde{a}_i)^2 + O_q(\delta).
    \]
\end{claim}
\begin{proof} First, note that $\sum_i \tilde{a}_i >0 $ if $\sum_i {a}_i \in [1-\delta,1+\delta]$. We have
    \[
    \sum_i (\tilde{a}_i - a^\star_i)^2 = \sum_i \left(\tilde{a}_i - \frac{\tilde{a}_i}{\sum_j \tilde{a}_j}\right)^2 = \left(\sum_i \tilde{a}_i - 1\right)^2 \sum_i\frac{\tilde{a}_i^2}{(\sum_j \tilde{a}_j)^2} \leq \left(\sum_i \tilde{a}_i - 1\right)^2,
    \]
    where we used
    $(\sum_j \tilde{a}_j)^2
    \geq \sum_j \tilde{a}_j^2$
    which holds as $\tilde{a}_j$ are non-negative. Now,
    \[
    \left(\sum_i \tilde{a}_i - 1\right)^2 \leq \left(\sum_i \tilde{a}_i - a_i\right)^2 + O_q(\delta)\leq q\sum_i (\tilde{a}_i - a_i)^2+ O_q(\delta),
    \]
    where the last inequality is the Cauchy-Schwarz inequality.
\end{proof}

\begin{claim}
For every constraint $C$ on $(s_1, s_2, s_3)$, if $(\V a_{s_1}, \V a_{s_2}, \V a_{s_3})$ and $(\V g_{s_1}, \V g_{s_2}, \V g_{s_3})$ are distributed according to the distribution from Equation (\ref{eq:rounding_scheme}), then
\begin{align*}
\left|
    \begin{array}{c}
    \E\left[{\constraint}(\bF^{{\sf dec}}_{s_1}({\V a}_{s_1}, {\V g}_{s_1})^\star, \bF^{{\sf dec}}_{s_2}({\V a}_{s_2}, {\V g}_{s_2})^\star, \bF^{{\sf dec}}_{s_3}({\V a}_{s_3}, {\V g}_{s_3})^\star) \right] -\\
   \E\left[{\constraint}(\widetilde{\bF^{{\sf dec}}_{s_1}}({\V a}_{s_1}, {\V g}_{s_1}), \widetilde{\bF^{{\sf dec}}_{s_2}}({\V a}_{s_2}, {\V g}_{s_2}), \widetilde{\bF^{{\sf dec}}_{s_3}}({\V a}_{s_3}, {\V g}_{s_3}))\right]
\end{array}
    \right|\leq o_\tau(1) + o_{\xi}(1)+ O_q(\delta).
\end{align*}
\end{claim}
\begin{proof}
Throughout this proof, we will use the following simple fact related to adding and removing conditioning on the event happening with high probability. If $X$ is a non negative random variable which is $c$-bounded and $E$ be the event that $X\in [1-\delta, 1+\delta]$ with $\Pr[E] \geq 1-\eta \geq 1/2$, then $\Expect{}{X} \leq  \Expect{}{X\mid E}  + \eta\cdot c$. Also, $\Expect{}{X\mid E} \leq 2\cdot \Expect{}{X}$.

We now begin with the proof. Using the fact that  $C$ is a Lipschitz function with Lipschitz constant $C_0(q)$, we get that the LHS is upper bounded by
    \[
    C_0(q)\cdot \sum_{j=1}^3 \E\left[ \| \bF^{{\sf dec}}_{s_j}({\V a}_{s_j}, {\V g}_{s_j})^\star -\widetilde{\bF^{{\sf dec}}_{s_j}}({\V a}_{s_j}, {\V g}_{s_j}) \|_2\right].
    \]
    Using Lemma~\ref{lem:bound_sd_mixed_inv_csp}, the fact that the expectation is bounded by $O_q(1)$, we get that the LHS can be upper bounded by
     \[
     C_0(q)\cdot \sum_{j=1}^3 \underbrace{\E\left[ \| \bF^{{\sf dec}}_{s_j}({\V z}_{s_j}, {\V g}_{s_j})^\star -\widetilde{\bF^{{\sf dec}}_{s_j}}({\V z}_{s_j}, {\V g}_{s_j}) \|_2 \right]}_{(\rom{1})}+o_\xi(1).
     \]
     For each fixing of ${\V z_{s_j}}$, the corresponding function has all the influences bounded by $\lll_{q} \tau$. Applying the invariance principle from Theorem~\ref{thm:invariance_principle}, we get
     \begin{align*}
\E\left[ \| \bF^{{\sf dec}}_{s_j}({\V z}_{s_j}, {\V g}_{s_j})^\star -\widetilde{\bF^{{\sf dec}}_{s_j}}({\V z}_{s_j}, {\V g}_{s_j}) \|_2 \right]
     &\leq \E\left[ \| \tilde{\bcalF}^{{\sf dec}}_{s_j}({\V z}_{s_j}, {\V z'_{s_j}})^\star -\widetilde{\tilde{\bcalF}^{{\sf dec}}_{s_j}}({\V z}_{s_j}, {\V z'_{s_j}}) \|_2\right] + o_\tau(1)
      % &\leq 2\E\left[ \| \bF^{{\sf dec}}_{s_j}({\V z}_{s_j}, {\V z'_{s_j}})^\star -\widetilde{\bF^{{\sf dec}}_{s_j}}({\V z}_{s_j}, {\V z'_{s_j}}) \|_2 \mid \wedge_{j=1}^3 E''_{s_j}  \right] + O_q(\delta)+ o_\tau(1),
     \end{align*}
    Overall, the LHS from the claim is upper bounded by
     $$C_0(q)\cdot \sum_{j=1}^3 \E\left[ \| \tilde{\bcalF}^{{\sf dec}}_{s_j}({\V z}_{s_j}, {\V z'_{s_j}})^\star -\widetilde{\tilde{\bcalF}^{{\sf dec}}_{s_j}}({\V z}_{s_j}, {\V z'_{s_j}}) \|_2 \right]+o_\xi(1) + o_\tau(1).$$
     In this expression, $({\V z}_{s_j}, {\V z'_{s_j}})$ are distributed according to the distribution $\mu_{s_j}^R\times \mu_{s_j}^R$. Consider $({\V z}_{s_j}, {\V z'_{s_j}}, {\V Z_{s_j}})$ the distributed from Lemma~\ref{lem:coupling_decoupled} with the same $\kappa$ as required for Claim~\ref{claim:close_to_1_csp}. As $T\gg d,\kappa^{-1}$, using Lemma~\ref{lem:coupling_decoupled}, the statistical distance between these two distributions on $({\V z}_{s_j}, {\V z'_{s_j}})$ can be upper bounded by $o_\xi(1)$. The quantity inside the expectation is $O_q(1)$ bounded, and hence, we can move to the coupled distribution by incurring an error of $o_\xi(1)$. Thus, we assume this distribution on $({\V z}_{s_j}, {\V z'_{s_j}})$ now onwards. Hence, the above quantity is upper bounded by
     \[
     C_0(q)\cdot \sum_{j=1}^3 \Expect{({\V z}_{s_j}, {\V z'_{s_j}}, {\V Z_{s_j}})}{ \| \tilde{\bcalF}^{{\sf dec}}_{s_j}({\V z}_{s_j}, {\V z'_{s_j}})^\star -\widetilde{\tilde{\bcalF}^{{\sf dec}}_{s_j}}({\V z}_{s_j}, {\V z'_{s_j}}) \|_2 }+o_\xi(1) + o_\tau(1),
     \]
     %where $({\V z}_{s_j}, {\V z'_{s_j}}, {\V Z_{s_j}})$ is distributed according the coupling from Lemma~\ref{lem:coupling_decoupled}.
     It follows from Claim~\ref{claim:close_to_1_csp} that $\Pr[E_{s_j}]\geq 1-\delta$. Hence, the above quantity is upper bounded by
    $$O_q(1)\cdot  \sum_{j=1}^3 \E\left[ \| \tilde{\bcalF}^{{\sf dec}}_{s_j}({\V z}_{s_j}, {\V z'_{s_j}})^\star -\widetilde{\tilde{\bcalF}^{{\sf dec}}_{s_j}}({\V z}_{s_j}, {\V z'_{s_j}})\|_2  ~\big|~  E_{s_j}\right]+o_\xi(1) + o_\tau(1)+ O_q(\delta).$$
    Using Claim~\ref{claim:trunc_scale_vec}, the above quantity (and hence the LHS from the claim) is upper bounded by
     $$O_q(1)\cdot  \sum_{j=1}^3 \E\left[ \| \mathfrak{R}(\tilde{\bcalF}^{{\sf dec}}_{s_j}({\V z}_{s_j}, {\V z'_{s_j}})) -\widetilde{\tilde{\bcalF}^{{\sf dec}}_{s_j}}({\V z}_{s_j}, {\V z'_{s_j}})\|_2  ~\big|~  E_{s_j}\right]+o_\xi(1) + o_\tau(1)+ O_q(\delta).$$
We have:
  \begin{align*}
  &\E\left[ \| \mathfrak{Re}(\tilde{\bcalF}^{{\sf dec}}_{s_j}({\V z}_{s_j}, {\V z'_{s_j}})) -\widetilde{\tilde{\bcalF}^{{\sf dec}}_{s_j}}({\V z}_{s_j}, {\V z'_{s_j}})\|_2  ~\big|~  E_{s_j}\right]\\
   &\quad\quad\quad\quad\quad\quad
   \leq \frac{1}{\Pr[E_{s_j}]}\E\left[ \| \mathfrak{Re}(\tilde{\bcalF}^{{\sf dec}}_{s_j}({\V z}_{s_j}, {\V z'_{s_j}})) -\widetilde{\tilde{\bcalF}^{{\sf dec}}_{s_j}}({\V z}_{s_j}, {\V z'_{s_j}})\|_2\right]\\
           &\quad\quad\quad\quad\quad\quad \leq 2\E\left[ \| \mathfrak{Re}(\tilde{\bcalF}^{{\sf dec}}_{s_j}({\V z}_{s_j}, {\V z'_{s_j}})) -\widetilde{\tilde{\bcalF}^{{\sf dec}}_{s_j}}({\V z}_{s_j}, {\V z'_{s_j}})\|_2\right],
  \end{align*}
  using the fact that $\Pr[E_{s_j}] = 1-
  \delta\geq 1/2$.
  Therefore, the LHS from the claim is upper bounded by
  \begin{equation}\label{eq:truncation_re_error_1}
      O_q(1)\cdot  \sum_{j=1}^3 \E\left[ \| \mathfrak{Re}(\tilde{\bcalF}^{{\sf dec}}_{s_j}({\V z}_{s_j}, {\V z'_{s_j}})) -\widetilde{\tilde{\bcalF}^{{\sf dec}}_{s_j}}({\V z}_{s_j}, {\V z'_{s_j}})\|_2 \right] +o_\xi(1) + o_\tau(1)+ O_q(\delta).
  \end{equation}
  Now, for a given variable $s\in \calV$,
 define
 \begin{align}
  \Theta_s:=\E\left[ \| \mathfrak{Re}(\tilde{\bcalF}^{{\sf dec}}_{s}({\V z}_{s}, {\V z'_{s}})) -\widetilde{\tilde{\bcalF}^{{\sf dec}}_{s}}({\V z}_{s}, {\V z'_{s}})\|_2\right]= \Expect{}{\truncerr(\tilde{\bcalF}^{{\sf dec}}_{s}({\V z}_{s}, {\V z'_{s}}))}, \label{eq:truncation_re_error_2}
  \end{align}
  where $\truncerr(a_1,\ldots,a_q) = \sqrt{\sum\limits_{i=1}^{q}\card{{\sf trunc}_{[0,1]}(\mathfrak{Re}(a_i)) - \mathfrak{Re}(a_i)}^2}$.
 Using the fact that $\truncerr$ is $O(1)$ Lipschitz,
 $$
 \E_{}\left[\truncerr\left(\tilde{\bcalF}^{{\sf dec}}_{s}({\V z}_{s}, {\V z'_{s}})\right)\right] \lll \E_{}\left[\truncerr\left(\mathrm{T}_{\calP_s, 1-\eps}{\bcalF}_{s}({\V Z_{s}})\right) + \norm{\tilde{\bcalF}^{{\sf dec}}_{s}({\V z}_{s}, {\V z'_{s}})-\mathrm{T}_{\calP_s, 1-\eps}{\bcalF}_{s}({\V Z}_{s})}_2\right].
 $$
  As $\mathrm{T}_{\calP_s, 1-\eps}{\calF}_{s, j} \in [0,1]$ for $1\leq j\leq q$,
  $\truncerr\left(\mathrm{T}_{\calP_s, 1-\eps}{\bcalF}_{s}({\V z}_{s})\right)=0$. Therefore, we have
  \begin{align*}
      \Theta_s^2 &\lll \Expect{}{\norm{\tilde{\bcalF}^{{\sf dec}}_{s}({\V z}_{s}, {\V z'_{s}})-\mathrm{T}_{\calP_s, 1-\eps}{\bcalF}_{s}({\V Z}_{s})}_2}^2\\
      & \leq \Expect{}{\norm{\tilde{\bcalF}^{{\sf dec}}_{s}({\V z}_{s}, {\V z'_{s}})-\mathrm{T}_{\calP_s, 1-\eps}{\bcalF}_{s}({\V Z}_{s})}_2^2}\\
      & = \sum_{j=1}^q \Expect{}{\card{\tilde{\calF}^{{\sf dec}}_{s,j}({\V z}_{s}, {\V z'_{s}})-\mathrm{T}_{\calP_s, 1-\eps}\calF_{s, j}({\V Z}_{s})}_2^2}\\
       & = \sum_{j=1}^q \Expect{}{\card{\tilde{\calF}^{{\sf dec}}_{s,j}({\V z}_{s}, {\V z'_{s}})- \calF'_{s, j}({\V Z}_{s}) + \calF'_{s, j}({\V Z}_{s})-  \mathrm{T}_{\calP_s, 1-\eps}\calF_{s, j}({\V Z}_{s})}_2^2}\\
       &\lll \sum_{j=1}^q \Expect{}{\card{\tilde{\calF}^{{\sf dec}}_{s,j}({\V z}_{s}, {\V z'_{s}})- \calF'_{s, j}({\V Z}_{s})}^2} + \Expect{}{\card{\calF'_{s, j}({\V Z}_{s})-  \mathrm{T}_{\calP_s, 1-\eps}\calF_{s, j}({\V Z}_{s})}_2^2}.
      \end{align*}
      Using Lemma~\ref{lem:approx_formula_noise_op_fancy} and Lemma~\ref{lem:decoupled_1} we can bound the first and second expectations from each summand by $o_\xi(1)$ respectively to get that $\Theta_s \lll o_\xi(1)$.
      Therefore,
  \begin{equation}\label{eq:truncation_re_error_3}
     \Expect{}{\truncerr(\tilde{\bcalF}^{{\sf dec}}_{s}({\V z}_{s}, {\V z'_{s}}))}\leq  o_{\xi}(1).
  \end{equation}
  Combining Equations (\ref{eq:truncation_re_error_1}), (\ref{eq:truncation_re_error_2}), and (\ref{eq:truncation_re_error_3}), we get that the LHS from the claim is upper bounded by $o_{\tau}(1) + o_{\xi}(1) + O_q(\delta)$ as claimed.
\end{proof}
\begin{proof}[Proof of Theorem~\ref{thm:dictatorship_test}]
Using the above claim and Equation(\ref{eq:rounding_scheme}), and the fact that $\tau\ll \xi$, we get
\begin{equation}
    \Round_{\bcalF}(\V V, \V \mu) \geq  \Expect{{\constraint}\sim \calC}{ \Expect{(\mathcal{D}_\mathcal{T}, \calG^{R})}{{\constraint}(\widetilde{\bF^{{\sf dec}}_{s_1}}({\V a}_{s_1}, {\V g}_{s_1}), \widetilde{\bF^{{\sf dec}}_{s_2}}({\V a}_{s_2}, {\V g}_{s_2}), \widetilde{\bF^{{\sf dec}}_{s_3}}({\V a}_{s_3}, {\V g}_{s_3}))}} - o_{\xi}(1).
\end{equation}

Combining with~\eqref{eq:testpassing_after_invariance} gives

\begin{equation*}
    \Pr[\bcalF\mbox{ passes }\dict_{\V V, \V \mu}] \leq \Round_{\bcalF}(\V V, \V \mu) + o_{\xi}(1).
\end{equation*}

As the integral value of the instance $\inst$ is at most ${\sf OPT}(\inst)$, we have $\Round_{\bcalF}(\V V, \V \mu)\leq  {\sf OPT}(\inst)$ it follows that
\[
		\Pr[\bcalF\mbox{ passes }\dict_{\V V, \V \mu}] \leq {\sf OPT}(\inst) + o_\xi(1).
        \qedhere
\]
\end{proof}
	% finishing the proof of Theorem~\ref{thm:dictatorship_test}.%\qed

\begin{figure}
        \begin{center}
\fbox{
  \begin{minipage}{0.95\linewidth}
    	{\bf Setup:} For each $s\in \calV$, the probability space $\Omega_s = (\Sigma, \mu_s)$ consists of atoms in $\Sigma$ with the distribution $\mu_s(a) = \|b_{s,a}\|^2$. Let $\bcalF_s$ denote the function obtained by interpreting the
				function $\bcalF : \Sigma^R\rightarrow \blacktriangle_q$ as a function over $\Omega_s^R$.
				
				{\bf Input:} An instance $\inst = (\calV, \calC)$ of a Max-$\mathcal{P}$-CSP such that the algorithm $\alg$ accepts $\inst$. Let $(\V V, \V \mu)$ be a solution for the basic SDP relaxation of $\inst$ and $\mathcal{T} \subseteq G_{\sf master}^\calV$ be the subspace of the set of satisfying assignments of the instance $\inst(\calV, \calC)$ over $G_{\sf master}^{\calV}:=\prod_{v\in \calV} H_v^\star$ that satisfy the conditions from Lemma~\ref{lemma:sdp_noZ_local_globalembed}. \\

				{\bf Rounding Scheme:}
				
				Step I: Sample $R$ Gaussian vectors $\V \zeta^{(1)}, \V \zeta^{(2)}, \ldots, \V \zeta^{(R)}$ with the same dimension as $\V V$.\\
				
				Step II: For each $s\in \calV$, do the following:

				\begin{enumerate}
					\item For each $j\in [R]$, let $g^{(j)}_{s,0}\equiv {\mathbf{1}}$ and for $c\in \{1,\ldots, q-1\}$, set
					$$g^{(j)}_{s,c}  =\sum_{\omega\in \Sigma} \ell_{s,c}(\omega)\langle \V b_{s,\omega}, \V \zeta^{(j)}\rangle.$$
					Let ${\bg^{(j)}_{s}} = ({g^{(j)}_{s,0}}\equiv \mathbf{1}, {g^{(j)}_{s,1}}, \ldots, {g^{(j)}_{s,q-1}})$ and $\bg_s = (\bg^{(1)}_s, \bg^{(2)}_s, \ldots, \bg^{(R)}_s)$.
					
					\item Sample $R$ uniformly random assignments ${\V \alpha}^{(1)}, {\V \alpha}^{(2)},\ldots, {\V \alpha}^{(R)}$ from the set of satisfying assignments to the instance $\inst$ over $G_{\sf master}^\calV$. Let ${\V a}_s = ({\V \alpha}^{(1)}_s, {\V \alpha}^{(2)}_s ,\ldots, {\V \alpha}^{(R)}_s)$.
					
					\item Evaluate the polynomial $ \bF^{\sf dec}_{s}$ with $({\V a}_s, {\bg_s})$ as inputs to obtain ${\V p}_s\in \mathbb{C}^q$, and let $\tilde{\V p}_s\in \mathbb{C}^q$ be such that $(\tilde{\V p}_s)_j = {\sf trunc}_{[0,1]}(\mathfrak{Re}(({\V p}_s)_j))$	where
					\[
					{\sf trunc}_{[0,1]}(x) = \left\{ \begin{array}{ll}
						0 &\mbox{ if } x<0\\
						x &\mbox{ if } 0\leq x\leq 1,\\
						1 &\mbox{ if } x>1,\\
					\end{array}\right. \]
					\item Round $\tilde{\V p}_s$ to $\V p^\star_s$.
					$$\V p^\star_s = \scale((\tilde{\V p}_s)_1, (\tilde{\V p}_s)_2, \ldots,(\tilde{\V p}_s)_{q}),$$
					where
					\[\scale(x_1, x_2, \ldots, x_{q}) = \left\{ \begin{array}{ll}
						\frac{1}{\sum_i x_i}(x_1, x_2, \ldots, x_{q}) & \mbox{ if } \sum_i x_i \neq 0,\\
						(1,0,0,\ldots, 0) & \mbox{ if } \sum_i x_i = 0.\\
					\end{array}\right.\\
					\]
					\item Assign the variable $s\in \calV$ a value $a\in \Sigma$ with probability $(\V p^\star_s)_{\amap^{-1}(a)}$.
					
				\end{enumerate}
				
				Step III: Output the assignment from Step II.
    \end{minipage}
  }
\end{center}
\caption{Rounding Scheme $\Round_{\bcalF}$.}
\label{fig:rounding_scheme}
\end{figure}

\clearpage
\printbibliography

@INPROCEEDINGS{RaghavendraS09,
	author={Raghavendra, Prasad and Steurer, David},
	booktitle={Proceedings of the IEEE Annual Symposium on Foundations of Computer Science (FOCS 2009)}, 
	title={How to Round Any {CSP}}, 
	year={2009},
	volume={},
	number={},
	pages={586-594},
	doi={10.1109/FOCS.2009.74}}

@inproceedings{BKMSTOC2025,
author = {Bhangale, Amey and Khot, Subhash and Minzer, Dor},
title = {On Approximability of Satisfiable k-{CSP}s: {V}},
year = {2025},
publisher = {Association for Computing Machinery},
address = {New York, NY, USA},
url = {https://doi.org/10.1145/3717823.3718127},
doi = {10.1145/3717823.3718127},
booktitle = {Proceedings of the Annual ACM Symposium on Theory of Computing},
pages = {62–71},
numpages = {10},
location = {Prague, Czechia},
series = {STOC '25}
}

@article{BKMcsp1,
  author       = {Amey Bhangale and
                  Subhash Khot and
                  Dor Minzer},
  title        = {On approximability of Satisfiable k-{CSP}s: {I}},
  journal      = {Computational Complexity},
  volume       = {34},
  number       = {2},
  pages        = {8},
  year         = {2025},
  url          = {https://doi.org/10.1007/s00037-025-00267-6},
  doi          = {10.1007/S00037-025-00267-6},
  bibsource    = {dblp computer science bibliography, https://dblp.org}
}

@article{BKM,
  author  = {Bhangale, Amey and Khot, Subhash and Minzer, Dor},
  title   = {Effective Bounds for Restricted $3$-Arithmetic Progressions in $\mathbb{F}_p^n$},
  journal = {Discrete Analysis},
  year    = {2024},
  number  = {16},
  pages   = {22},
  doi     = {10.19086/da.125858},
}

@article{BKMcsp2,
  author       = {Amey Bhangale and
                  Subhash Khot and
                  Dor Minzer},
  title        = {On approximability of satisfiable $k$-{CSP}s: {II}},
  journal      = {Combinatorics, Probability and Computing},
  volume       = {34},
  number       = {6},
  pages        = {857--926},
  year         = {2025},
  url          = {https://doi.org/10.1017/s0963548325100114},
  doi          = {10.1017/S0963548325100114},
  bibsource    = {dblp computer science bibliography, https://dblp.org}
}

@inproceedings{BKMcsp3,
  author       = {Amey Bhangale and
                  Subhash Khot and
                  Dor Minzer},
  title        = {{On Approximability of Satisfiable $k$-{CSP}s: {III}}},
  booktitle    = {Proceedings of the Annual ACM Symposium on Theory of Computing},
  pages        = {643--655},
  year         = {2023},
  url          = {https://doi.org/10.1145/3564246.3585121},
  doi          = {10.1145/3564246.3585121},
  bibsource    = {dblp computer science bibliography, https://dblp.org}
}

@inproceedings{BKMcsp4,
  author       = {Amey Bhangale and
                  Subhash Khot and
                  Dor Minzer},
  title        = {On Approximability of Satisfiable $k$-{CSP}s: {IV}},
  booktitle    = {Proceedings of the Annual ACM Symposium on Theory of Computing},
  pages        = {1423--1434},
  publisher    = {{ACM}},
  year         = {2024},
  url          = {https://doi.org/10.1145/3618260.3649610},
  doi          = {10.1145/3618260.3649610},
  bibsource    = {dblp computer science bibliography, https://dblp.org}
}

@article{MOO,
  author  = {Mossel, Elchanan and O'Donnell, Ryan and Oleszkiewicz, Krzysztof},
  title   = {Noise stability of functions with low influences: invariance and optimality},
  journal = {Annals of Mathematics},
  volume  = {171},
  number  = {1},
  pages   = {295--341},
  year    = {2010},
  doi     = {10.4007/annals.2010.171.295}
}

@article{tao2010inverse,
  author  = {Tao, Terence and Ziegler, Tamar},
  title   = {The inverse conjecture for the {G}owers norm over finite fields via the correspondence principle},
  journal = {Analysis \& PDE},
  volume  = {3},
  number  = {1},
  pages   = {1--20},
  year    = {2010},
  doi     = {10.2140/apde.2010.3.1}
}

@article{bergelson2010inverse,
  author  = {Bergelson, Vitaly and Tao, Terence and Ziegler, Tamar},
  title   = {An inverse theorem for the uniformity seminorms associated with the action of $\mathbb{F}_p^\infty$},
  journal = {Geometric and Functional Analysis},
  volume  = {19},
  number  = {6},
  pages   = {1539--1596},
  year    = {2010},
  doi     = {10.1007/s00039-010-0051-1}
}

@article{tao2012inverse,
  author  = {Tao, Terence and Ziegler, Tamar},
  title   = {The inverse conjecture for the {G}owers norm over finite fields in low characteristic},
  journal = {Annals of Combinatorics},
  volume  = {16},
  number  = {1},
  pages   = {121--188},
  year    = {2012},
  doi     = {10.1007/s00026-011-0124-3}
}

@article{DBLP:conf/soda/BrakensiekHPZ21,
author = {Brakensiek, Joshua and Huang, Neng and Potechin, Aaron and Zwick, Uri},
title = {On the Mysteries of MAX NAE-SAT},
journal = {SIAM Journal on Discrete Mathematics},
volume = {39},
number = {1},
pages = {267-313},
year = {2025},
doi = {10.1137/23M1591578},

URL = { 
    
        https://doi.org/10.1137/23M1591578
    
    

},
eprint = { 
    
        https://doi.org/10.1137/23M1591578
    
    

}    
}

@article{Lander,
  author  = {Ladner, Richard E.},
  title   = {On the Structure of Polynomial Time Reducibility},
  journal = {Journal of the ACM},
  volume  = {22},
  number  = {1},
  pages   = {155--171},
  year    = {1975},
  doi     = {10.1145/321864.321877}
}

@book{ODonnell,
  author    = {O'Donnell, Ryan},
  title     = {Analysis of Boolean Functions},
  publisher = {Cambridge University Press},
  year      = {2014},
  doi       = {10.1017/CBO9781139814782}
}

@misc{green2007montreal,
      title={Montreal Lecture Notes on Quadratic Fourier Analysis}, 
      author={Ben Green},
      year={2007},
      eprint={math/0604089},
      archivePrefix={arXiv},
      primaryClass={math.CA},
      url={https://arxiv.org/abs/math/0604089}, 
}

@inproceedings{Schoenebeck,
  author    = {Grant Schoenebeck},
  title     = {Linear Level {L}asserre Lower Bounds for Certain $k$-CSPs},
  booktitle = {Proceedings of the IEEE Annual Symposium on Foundations of Computer Science (FOCS 2008)},
  pages     = {593--602},
  year      = {2008},
  doi       = {10.1109/FOCS.2008.74}
}

@article{Grigoriev,
  author  = {Grigoriev, Dima},
  title   = {Linear lower bound on degrees of Positivstellensatz calculus proofs for the parity},
  journal = {Theoretical Computer Science},
  volume  = {259},
  number  = {1--2},
  pages   = {613--622},
  year    = {2001},
  doi     = {10.1016/S0304-3975(00)00157-2}
}

@article{Has01,
  author    = {Johan H{\aa}stad},
  title     = {Some optimal inapproximability results},
  journal   = {Journal of the ACM},
  volume    = {48},
  number    = {4},
  pages     = {798--859},
  year      = {2001},
  url       = {https://doi.org/10.1145/502090.502098},
  doi       = {10.1145/502090.502098},
  bibsource = {dblp computer science bibliography, https://dblp.org}
}

@article{AroraLMSS1998,
	author =	 {Sanjeev Arora and Carsten Lund and Rajeev Motwani
	and Madhu Sudan and Mario Szegedy},
	journal =	 {Journal of the ACM},
	number =	 3,
	pages =	 {501--555},
	title =	 {Proof Verification and the Hardness of Approximation
	Problems},
	volume =	 45,
	year =	 1998,
	doi =		 {10.1145/278298.278306},
	eccc =	 {1998/TR98-008},
}

@article{FederV98,
	author = {Feder, Tomás and Vardi, Moshe Y.},
	title = {The Computational Structure of Monotone Monadic {SNP} and Constraint Satisfaction: A Study through Datalog and Group Theory},
	journal = {SIAM Journal on Computing},
	volume = {28},
	number = {1},
	pages = {57-104},
	year = {1998},
	doi = {10.1137/S0097539794266766}
}

@inproceedings{Rag08,
  author    = {Prasad Raghavendra},
  title     = {Optimal algorithms and inapproximability results for every {CSP}?},
  booktitle = {Proceedings of the Annual ACM Symposium on Theory of Computing (STOC 2008)},
  pages     = {245--254},
  year      = {2008},
  publisher = {ACM},
  doi       = {10.1145/1374376.1374414}
}

@article{AroraS1998,
	author =	 {Sanjeev Arora and Shmuel Safra},
	journal =	 {Journal of the ACM},
	number =	 1,
	pages =	 {70--122},
	title =	 {Probabilistic Checking of Proofs: A New
	Characterization of~{NP}},
	volume =	 45,
	year =	 1998,
	doi =		 {10.1145/273865.273901},
}

@article{EngebretsenHR,
  author  = {Engebretsen, Lars and Holmerin, Jonas and Russell, Alexander},
  title   = {Inapproximability results for equations over finite groups},
  journal = {Theoretical Computer Science},
  volume  = {312},
  number  = {1},
  pages   = {17--45},
  year    = {2004},
  doi     = {10.1016/S0304-3975(03)00401-8}
}

@article{GW95,
	author = {Goemans, Michel X. and Williamson, David P.},
	title = {Improved Approximation Algorithms for Maximum Cut and Satisfiability Problems Using Semidefinite Programming},
	year = {1995},
	publisher = {Association for Computing Machinery},
	address = {New York, NY, USA},
	volume = {42},
	number = {6},
	issn = {0004-5411},
	doi = {10.1145/227683.227684},
	journal = {Journal of the ACM},
	pages = {1115–1145},
	numpages = {31}
}

@article{FGLSS96,
	title={{Interactive proofs and the hardness of approximating cliques}},
	author={Feige, Uriel and Goldwasser, Shafi and Lov{\'a}sz, Laszlo and Safra, Shmuel and Szegedy, Mario},
	journal={Journal of the ACM},
	volume={43},
	number={2},
	pages={268--292},
	year={1996},
	publisher={ACM},
	doi = {10.1145/226643.226652}
}

@article{Mossel,
  author  = {Mossel, Elchanan},
  title   = {Gaussian Bounds for Noise Correlation of Functions},
  journal = {Geometric and Functional Analysis},
  volume  = {19},
  number  = {6},
  pages   = {1713--1756},
  year    = {2010},
  doi     = {10.1007/s00039-010-0047-x}
}

@inproceedings{Schaefer78,
	author = {Schaefer, Thomas J.},
	title = {The Complexity of Satisfiability Problems},
	year = {1978},
	isbn = {9781450374378},
	publisher = {Association for Computing Machinery},
	address = {New York, NY, USA},
	url = {https://doi.org/10.1145/800133.804350},
	doi = {10.1145/800133.804350},
	booktitle = {Proceedings of the Annual ACM Symposium on Theory of Computing (STOC 1978)},
	pages = {216–226},
	numpages = {11},
	location = {San Diego, California, USA}
}

@article{Zhuk20,
	author = {Zhuk, Dmitriy},
	title = {A Proof of the {CSP} Dichotomy Conjecture},
	year = {2020},
	issue_date = {October 2020},
	publisher = {Association for Computing Machinery},
	address = {New York, NY, USA},
	volume = {67},
	number = {5},
	issn = {0004-5411},
	url = {https://doi.org/10.1145/3402029},
	doi = {10.1145/3402029},
	journal = {Journal of the ACM},
	month = aug,
	articleno = {30},
	numpages = {78}
}

@INPROCEEDINGS{Bulatov17,
	author={Bulatov, Andrei A.},
	booktitle={Proceedings of the IEEE Annual Symposium on Foundations of Computer Science (FOCS 2017)},
	title={A Dichotomy Theorem for Nonuniform {CSP}s},
	year={2017},
	volume={},
	number={},
	pages={319-330},
	doi={10.1109/FOCS.2017.37}}

@inproceedings{Khot02UG,
	title={On the power of unique 2-prover 1-round games},
	author={Khot, Subhash},
	booktitle = {Proceedings of the Annual ACM Symposium on Theory of Computing (STOC 2002)},
	year = {2002},
	publisher = {{ACM}},
	pages={767--775},
	doi = {10.1145/509907.510017},
}

@article{DHJ3,
  title={A new proof of the density {H}ales-{J}ewett theorem},
  author={Polymath, DHJ},
  journal={Annals of Mathematics},
  pages={1283--1327},
  year={2012},
  publisher={JSTOR}
}

@article{DHJ1,
  author  = {Furstenberg, H. and Katznelson, Y.},
  title   = {A density version of the {H}ales-{J}ewett theorem for $k=3$},
  journal = {Discrete Mathematics},
  volume  = {75},
  pages   = {227--241},
  year    = {1989},
  doi     = {10.1016/0012-365X(89)90089-7}
}

@article{DHJ2,
  author  = {Furstenberg, H. and Katznelson, Y.},
  title   = {A density version of the {H}ales-{J}ewett theorem},
  journal = {Journal d'Analyse Math\'ematique},
  volume  = {57},
  number  = {1},
  pages   = {64--119},
  year    = {1991},
  doi     = {10.1007/BF03041066}
}

@article{multidimensional,
  author  = {Furstenberg, H. and Katznelson, Y.},
  title   = {An ergodic Szemer{\'e}di theorem for commuting transformations},
  journal = {Journal d'Analyse Math\'ematique},
  volume  = {34},
  number  = {1},
  pages   = {275--291},
  year    = {1978},
  doi     = {10.1007/BF02790016}
}

@InProceedings{BravermanKM21,
	author =	{Mark Braverman and Subhash Khot and Dor Minzer},
	title =	{{On Rich 2-to-1 Games}},
	booktitle =	{Proceedings of the Innovations in Theoretical Computer Science Conference (ITCS 2021)},
	pages =	{27:1--27:20},
	year =	{2021},
	doi =		{10.4230/LIPIcs.ITCS.2021.27}
}

@incollection{hales2009regularity,
  title={Regularity and positional games},
  author={Hales, Alfred W and Jewett, Robert I},
  booktitle={Classic Papers in Combinatorics},
  pages={320--327},
  year={2009},
  publisher={Springer}
}

@inproceedings{BKnonabelian,
  author       = {Amey Bhangale and
                  Subhash Khot},
  title        = {Optimal inapproximability of satisfiable k-{LIN} over non-abelian groups},
  booktitle    = {Proceedings of the Annual ACM Symposium on Theory of Computing (STOC 2021)},
  pages        = {1615--1628},
  publisher    = {{ACM}},
  year         = {2021},
  url          = {https://doi.org/10.1145/3406325.3451003},
  doi          = {10.1145/3406325.3451003},
  bibsource    = {dblp computer science bibliography, https://dblp.org}
}

@article{GuruswamiL18,
	author = {Guruswami, Venkatesan and Lee, Euiwoong},
	da = {2018/06/01},
	doi = {10.1007/s00493-016-3383-0},
	id = {Guruswami2018},
	isbn = {1439-6912},
	journal = {Combinatorica},
	number = {3},
	pages = {547--599},
	title = {Strong Inapproximability Results on Balanced Rainbow-Colorable Hypergraphs},
	ty = {JOUR},
	url = {https://doi.org/10.1007/s00493-016-3383-0},
	volume = {38},
	year = {2018}}

@article{KV,
  author  = {Khot, Subhash A. and Vishnoi, Nisheeth K.},
  title   = {The Unique Games Conjecture, Integrality Gap for Cut Problems and Embeddability of Negative-Type Metrics into $\ell_1$},
  journal = {Journal of the ACM},
  volume  = {62},
  number  = {1},
  pages   = {8:1--8:39},
  year    = {2015},
  doi     = {10.1145/2629614}
}

@inproceedings{Tulsiani,
  author       = {Madhur Tulsiani},
  editor       = {Michael Mitzenmacher},
  title        = {{CSP} gaps and reductions in the {L}asserre hierarchy},
  booktitle    = {Proceedings of the Annual ACM Symposium on Theory of Computing (STOC 2009)},
  pages        = {303--312},
  publisher    = {{ACM}},
  year         = {2009},
  url          = {https://doi.org/10.1145/1536414.1536457},
  doi          = {10.1145/1536414.1536457},
  bibsource    = {dblp computer science bibliography, https://dblp.org}
}

@article{BKMcspDHJ,
  author       = {Amey Bhangale and
                  Subhash Khot and
                  Yang P. Liu and
                  Dor Minzer},
  title        = {Reasonable Bounds for Combinatorial Lines of Length Three},
  journal      = {Computing Research Repository},
  volume       = {abs/2411.15137},
  year         = {2024},
  url          = {https://doi.org/10.48550/arXiv.2411.15137},
  doi          = {10.48550/ARXIV.2411.15137},
  eprinttype    = {arXiv},
  eprint       = {2411.15137},
  bibsource    = {dblp computer science bibliography, https://dblp.org}
}

@article{BKMcsp7,
  author       = {Amey Bhangale and
                  Subhash Khot and
                  Yang P. Liu and
                  Dor Minzer},
  title        = {On Approximability of Satisfiable $k$-{CSP}s: {VII}},
  journal      = {Computing Research Repository},
  volume       = {abs/2411.15136},
  year         = {2024},
  url          = {https://doi.org/10.48550/arXiv.2411.15136},
  doi          = {10.48550/ARXIV.2411.15136},
  eprinttype    = {arXiv},
  eprint       = {2411.15136},
  bibsource    = {dblp computer science bibliography, https://dblp.org}
}

@article{BKMcsp6,
  author       = {Amey Bhangale and
                  Subhash Khot and
                  Yang P. Liu and
                  Dor Minzer},
  title        = {On Approximability of Satisfiable $k$-{CSP}s: {VI}},
  journal      = {Computing Research Repository},
  volume       = {abs/2411.15133},
  year         = {2024},
  url          = {https://doi.org/10.48550/arXiv.2411.15133},
  doi          = {10.48550/ARXIV.2411.15133},
  eprinttype    = {arXiv},
  eprint       = {2411.15133},
  bibsource    = {dblp computer science bibliography, https://dblp.org}
}

@inproceedings{Zwick99,
author = {Zwick, Uri},
title = {Outward rotations: a tool for rounding solutions of semidefinite programming relaxations, with applications to MAX CUT and other problems},
year = {1999},
isbn = {1581130678},
publisher = {Association for Computing Machinery},
address = {New York, NY, USA},
url = {https://doi.org/10.1145/301250.301431},
doi = {10.1145/301250.301431},
booktitle = {Proceedings of the Annual ACM Symposium on Theory of Computing (STOC 1999)},
pages = {679–687},
numpages = {9},
location = {Atlanta, Georgia, USA}
}

@article{KMS98,
author = {Karger, David and Motwani, Rajeev and Sudan, Madhu},
title = {Approximate graph coloring by semidefinite programming},
year = {1998},
issue_date = {March 1998},
publisher = {Association for Computing Machinery},
address = {New York, NY, USA},
volume = {45},
number = {2},
issn = {0004-5411},
url = {https://doi.org/10.1145/274787.274791},
doi = {10.1145/274787.274791},
journal = {Journal of the ACM},
month = {3},
pages = {246–265},
numpages = {20}
}

@article{CharikarMM09,
author = {Charikar, Moses and Makarychev, Konstantin and Makarychev, Yury},
title = {Near-optimal algorithms for maximum constraint satisfaction problems},
year = {2009},
issue_date = {July 2009},
publisher = {Association for Computing Machinery},
address = {New York, NY, USA},
volume = {5},
number = {3},
issn = {1549-6325},
url = {https://doi.org/10.1145/1541885.1541893},
doi = {10.1145/1541885.1541893},
journal = {ACM Transactions on Algorithms},
month = {7},
articleno = {32},
numpages = {14}
}

@inproceedings{Cook71,
author = {Cook, Stephen A.},
title = {The complexity of theorem-proving procedures},
year = {1971},
isbn = {9781450374644},
publisher = {Association for Computing Machinery},
address = {New York, NY, USA},
url = {https://doi.org/10.1145/800157.805047},
doi = {10.1145/800157.805047},
booktitle = {Proceedings of the Annual ACM Symposium on Theory of Computing},
pages = {151–158},
numpages = {8},
location = {Shaker Heights, Ohio, USA},
series = {STOC '71}
}

@article{Levin73,
  title={Universal sequential search problems},
  author={Levin, Leonid Anatolevich},
  journal={Problemy peredachi informatsii},
  volume={9},
  number={3},
  pages={115--116},
  year={1973},
  publisher={Russian Academy of Sciences, Branch of Informatics, Computer Equipment and~…}
}

@inproceedings{Zwick98,
	title={Approximation Algorithms for Constraint Satisfaction Problems Involving at Most Three Variables per Constraint.},
	author={Zwick, Uri},
	booktitle={Proceedings of the Annual ACM-SIAM Symposium on Discrete Algorithms (SODA 1998)},
	volume={98},
	pages={201--210},
	year={1998}
}

@article{Hastad14,
  author  = {H{\aa}stad, Johan},
  title   = {On the {NP}-Hardness of Max-Not-2},
  journal = {SIAM Journal on Computing},
  volume  = {43},
  number  = {1},
  pages   = {179--193},
  year    = {2014},
  doi     = {10.1137/120882718}
}

@inproceedings{OW_dict_09,
  author    = {O'Donnell, Ryan and Wu, Yi},
  title     = {3-Bit Dictator Testing: 1 vs. 5/8},
  booktitle = {Proceedings of the Annual ACM-SIAM Symposium on Discrete Algorithms (SODA 2009)},
  pages     = {365--374},
  year      = {2009},
  doi       = {10.1137/1.9781611973068.41}
}

@inproceedings{OW_hard_09,
author = {O'Donnell, Ryan and Wu, Yi},
title = {Conditional hardness for satisfiable 3-{CSP}s},
year = {2009},
isbn = {9781605585062},
publisher = {Association for Computing Machinery},
address = {New York, NY, USA},
url = {https://doi.org/10.1145/1536414.1536482},
doi = {10.1145/1536414.1536482},
booktitle = {Proceedings of the Annual ACM Symposium on Theory of Computing},
pages = {493–502},
numpages = {10},
location = {Bethesda, MD, USA},
series = {STOC '09}
}

@inproceedings{KSaket06,
  author    = {Khot, Subhash and Saket, Rishi},
  title     = {A 3-Query Non-Adaptive {PCP} with Perfect Completeness},
  booktitle = {Proceedings of the Conference on Computational Complexity (CCC 2006)},
  pages     = {159--169},
  year      = {2006},
  doi       = {10.1109/CCC.2006.5}
}

@inproceedings{CiardoZ24,
author = {Ciardo, Lorenzo and \v{Z}ivn\'{y}, Stanislav},
title = {Semidefinite Programming and Linear Equations vs. Homomorphism Problems},
year = {2024},
isbn = {9798400703836},
publisher = {Association for Computing Machinery},
address = {New York, NY, USA},
url = {https://doi.org/10.1145/3618260.3649635},
doi = {10.1145/3618260.3649635},
booktitle = {Proceedings of the Annual ACM Symposium on Theory of Computing (STOC 2024)},
pages = {1935–1943},
numpages = {9},
keywords = {Diophantine equations, affine integer programming relaxations, approximate graph colouring, approximate graph homomorphism, promise constraint satisfaction, semidefinite programming relaxations},
location = {Vancouver, BC, Canada},
%series = {Proceedings of the Annual ACM Symposium on Theory of Computing}
}

@article{CiardoZ_CLAP23,
author = {Ciardo, Lorenzo and \v{Z}ivn\'{y}, Stanislav},
title = {CLAP: A New Algorithm for Promise {CSP}s},
journal = {SIAM Journal on Computing},
volume = {52},
number = {1},
pages = {1-37},
year = {2023},
doi = {10.1137/22M1476435},
URL = {https://doi.org/10.1137/22M1476435},
eprint = { https://doi.org/10.1137/22M1476435}
}

@article{BrakensiekGWZ20,
author = {Brakensiek, Joshua and Guruswami, Venkatesan and Wrochna, Marcin and \v{Z}ivn\'{y}, Stanislav},
title = {The Power of the Combined Basic Linear Programming and Affine Relaxation for Promise Constraint Satisfaction Problems},
journal = {SIAM Journal on Computing},
volume = {49},
number = {6},
pages = {1232-1248},
year = {2020},
doi = {10.1137/20M1312745},
URL = {https://doi.org/10.1137/20M1312745},
eprint = {https://doi.org/10.1137/20M1312745}}

@inproceedings{BrakensiekG19,
  author       = {Joshua Brakensiek and
                  Venkatesan Guruswami},
  title        = {An Algorithmic Blend of {LP}s and Ring Equations for Promise {CSP}s},
  booktitle    = {Proceedings of the Annual ACM-SIAM Symposium on Discrete Algorithms (SODA 2019)},
  pages        = {436--455},
  publisher    = {{SIAM}},
  year         = {2019},
  url          = {https://doi.org/10.1137/1.9781611975482.28},
  doi          = {10.1137/1.9781611975482.28},
  bibsource    = {dblp computer science bibliography, https://dblp.org}
}

@INPROCEEDINGS{BartoK09,
  author={Barto, Libor and Kozik, Marcin},
  booktitle={Proceedings of the IEEE Annual Symposium on Foundations of Computer Science (FOCS 2009)}, 
  title={Constraint Satisfaction Problems of Bounded Width}, 
  year={2009},
  volume={},
  number={},
  pages={595-603},
  doi={10.1109/FOCS.2009.32}}

@inproceedings{KTW14,
author = {Khot, Subhash and Tulsiani, Madhur and Worah, Pratik},
title = {A characterization of strong approximation resistance},
year = {2014},
isbn = {9781450327107},
publisher = {Association for Computing Machinery},
address = {New York, NY, USA},
url = {https://doi.org/10.1145/2591796.2591817},
doi = {10.1145/2591796.2591817},
booktitle = {Proceedings of the Annual ACM Symposium on Theory of Computing (STOC 2014)},
pages = {634–643},
numpages = {10},
location = {New York, New York}
}

@article{BenabbasGMT12,
 author = {Benabbas, Siavosh and Georgiou, Konstantinos and Magen, Avner and Tulsiani, Madhur},
 title = {{SDP} Gaps from Pairwise Independence},
 year = {2012},
 pages = {269--289},
 doi = {10.4086/toc.2012.v008a012},
 publisher = {Theory of Computing},
 journal = {Theory of Computing},
 volume = {8},
 number = {12},
 URL = {https://theoryofcomputing.org/articles/v008a012},
}

@InProceedings{BrakensiekG14,
  author =	{Brakensiek, Joshua and Guruswami, Venkatesan},
  title =	{{New Hardness Results for Graph and Hypergraph Colorings}},
  booktitle =	{Proceedings of the Conference on Computational Complexity},
  pages =	{14:1--14:27},
  series =	{Leibniz International Proceedings in Informatics (LIPIcs)},
  ISBN =	{978-3-95977-008-8},
  ISSN =	{1868-8969},
  year =	{2016},
  volume =	{50},
  publisher =	{Schloss Dagstuhl -- Leibniz-Zentrum f{\"u}r Informatik},
  address =	{Dagstuhl, Germany},
  URL =		{https://drops.dagstuhl.de/entities/document/10.4230/LIPIcs.CCC.2016.14},
  URN =		{urn:nbn:de:0030-drops-58291},
  doi =		{10.4230/LIPIcs.CCC.2016.14}
}

@book{GrotschelLS12,
  author    = {Gr{\"o}tschel, Martin and Lov{\'a}sz, L{\'a}szl{\'o} and Schrijver, Alexander},
  title     = {Geometric Algorithms and Combinatorial Optimization},
  series    = {Algorithms and Combinatorics},
  volume    = {2},
  publisher = {Springer},
  address   = {Berlin},
  year      = {1988},
  doi       = {10.1007/978-3-642-78240-4}
}

\newpage
\appendix
 \section{Omitted Proofs}\label{sec:proof_that_symm}
 Let $p\geq 5$ be a prime,
 take $\Sigma = \mathbb{F}_p$
 and consider the predicate $P\colon \Sigma^3\to\{0,1\}$
 defined as $P(x,y,z) = 1_{x,y,z\text{ are distinct}}$.
 \begin{claim}
     The collection $\{P\}$ is
     $\textsc{Mildly-Symmetric}$.
 \end{claim}
 \begin{proof}
     For each $a\in\mathbb{F}_p\setminus \{0\}$ and $b\in\mathbb{F}_p$
     define the map
     $\tau_{a,b}\colon \Sigma\to\Sigma$ by
     $\tau_{a,b}(u)=au+b$. It is clear that each one of these maps preserves the satisfying assignments of $P$, so it remains to check
     the second condition in Definition~\ref{def:symm_p}.

     Let $(x,y,z)$ be any satisfying assignment.
     We can find $a\neq 0$ and $b$
     such that $ax+b = 0$
     and $ay+b = 1$, so we may assume without loss of generality that $(x,y,z) = (0,1,z')$ for some $z'\neq 0,1$; this is justified as the orbits of $(0,1,az+b)$ and $(x,y,z)$ under $\{\tau_{a',b'}\}_{a'\neq 0,b'}$ are the same,
      and $az+b$ is some element in $\mathbb{F}_p$ not equal to $0,1$. 

     Let $A = \sett{(b,a+b,az+b)}{a\neq 0, b\in\mathbb{F}_p}$
     be the orbit of $(0,1,z)$
     under $\{\tau_{a,b}\}_{a\neq 0, b}$ for $z\neq 0,1$,
     and suppose that $\sigma,\gamma,\phi$ is
     an embedding of $A$ into $(\mathbb{Z},+)$. By applying an affine shift
     to all embeddings we may
     assume that
     $\gamma(0)=\phi(0) = 0$, and by the definition of embeddings we get that
     \begin{equation}\label{eq:omitted1}
         \sigma(b) + \gamma(a+b) + \phi(az+b) = 0\qquad
         ~\forall a\neq 0, b.
     \end{equation}
     Taking
     $b=-a$ and using $\gamma(0)=0$ we get that
     $\sigma(-a) + \phi(a(z-1)) = 0$ for $a\neq 0$,
     so $\phi(a(z-1)) = -\sigma(-a)$. Hence we get that
     \[
        -\phi(-b(z-1))
        +\gamma(a+b)
        +\phi(az+b) = 0
        \qquad
         ~\forall a,b\neq 0.
     \]
     Taking $b = -az$ and using $\phi(0) = 0$ we get
     $-\phi(az(z-1))+\gamma(-a(z-1)) = 0$,
     meaning that $\gamma(y) = \phi(-zy)$ for all $y\neq 0$.
     Equality also holds for $y=0$
     (as both values are $0$)
     and we conclude that
     \[
        -\phi(-b(z-1))
        +\phi(-z(a+b))
        +\phi(az+b) = 0
                \qquad
         ~\forall a,b\neq 0.
     \]
     Note that
     the image of $(-b(z-1), -z(a+b))$ under $a,b\neq 0$
     consists of $(\alpha,\beta)$
     such that $\alpha\neq 0$
     and $(z-1)\beta-z\alpha\neq 0$, and we get that for any such $\alpha,\beta$ it holds that $\phi(\alpha) = \phi(\beta) + \phi(\alpha-\beta)$.
     We now use the idea of
     local self-correction
     to argue that $\phi(\alpha) = \phi(\beta) + \phi(\alpha-\beta)$ in fact holds
     for all $\alpha,\beta\in\mathbb{F}_p$.

    Take any $\alpha,\beta$
     satisfying
     $(z-1)\beta-z\alpha= 0$.
     If $\beta = 0$ the equality $\phi(\alpha) = \phi(\beta) + \phi(\alpha-\beta)$ is clear, so assume otherwise. It follows that
     $\alpha = \frac{z-1}{z}\beta$, so
     $\alpha\neq 0,\beta$.
     Choose
     $\alpha'\in\mathbb{F}_p\setminus \{\alpha,0,\frac{z}{z-1}{\alpha}\}$ arbitrarily and then $\beta'\in\mathbb{F}_p$ uniformly.
     With probability at least
     $1-\frac{3}{p}>0$ we
     have that
     $(z-1)\beta'-z\alpha' \neq 0$,
     $(z-1)(\beta-\beta')-z(\alpha-\alpha') \neq 0$,
     $(z-1)\beta'-z\beta \neq 0$,
     $(z-1)(\alpha'-\beta')-z(\alpha-\beta) \neq 0$.\footnote{Note that we have $4$ conditions, so a naive application of the union bound only gives a lower bound of $1-\frac{4}{p}$
     on the probability that all of these events hold. However, since $(z-1)\beta-z\alpha = 0$, the first two conditions can be seen to be equivalent, so the result of the union bound improves to $1-\frac{3}{p}$.}.
     Also, by choice of $\alpha'$ we
     have $(z-1)\alpha'-z\alpha \neq 0$
     and $\alpha'\neq 0,\alpha,\beta$,
     $\beta\neq 0$ and $\alpha\neq 0$. Fix $\alpha',\beta'$
     satisfying all of these inequalities.
     Adding up the constraints we get from the first two conditions we get that
     \[
        \phi(\alpha') + \phi(\alpha-\alpha')
        =
        (\phi(\beta')+\phi(\alpha'-\beta'))
        +
        (\phi(\beta-\beta')+\phi(\alpha-\beta-\alpha'+\beta'))
     \]
     Using the third and fourth conditions
     we have $\phi(\beta')+\phi(\beta-\beta') = \phi(\beta)$
     and $\phi(\alpha'-\beta') + \phi(\alpha-\beta-\alpha'+\beta') = \phi(\alpha-\beta)$
     so that the right hand side simplifies to $\phi(\beta)+\phi(\alpha-\beta)$. Using the fifth
     condition the left hand
     side simplifies to $\phi(\alpha)$, altogether giving that
     $\phi(\alpha) = \phi(\beta)+\phi(\alpha-\beta)$.

     We conclude that
     $\phi(\alpha) = \phi(\beta)+\phi(\alpha-\beta)$
     for all $\alpha,\beta\in\mathbb{F}_p$.
     Thus, we get that
     $\phi(x) = x\phi(1)$ for all $x\in\mathbb{F}_p$ and
     also that
     $\phi(1) = \phi(2\cdot (p-1)/2 + 2)=2\phi((p-1)/2)+\phi(2)=(p+1)\phi(1)$,
     hence $\phi(1)=0$. This implies that
     $\phi\equiv 0$, and using
     $\gamma(y) = \phi(-zy)$ that holds for $y\neq 0$ and $\gamma(0) = 0$ it follows that $\gamma\equiv 0$.
     Plugging this into~\eqref{eq:omitted1} gives that $\sigma\equiv 0$, and the proof is concluded.
 \end{proof}

\end{document}